\documentclass[a4paper,USenglish,cleveref,autoref,thm-restate]{article}
\usepackage{fullpage}

\usepackage{amsthm}
\usepackage{amsmath}
\usepackage{graphicx}
\usepackage[hang]{subfigure}
\usepackage{multirow}
\usepackage{comment}
\usepackage{algorithm}
\usepackage[noend]{algpseudocode}
\usepackage{etoolbox}
\usepackage{cleveref}
\usepackage{enumitem}
\usepackage{url}

\newtheorem{theorem}{Theorem}[section]

\newenvironment{definition}[1][Definition]{\begin{trivlist}
		\item[\hskip \labelsep {\bfseries #1}]}{\end{trivlist}}

\begin{document}

\newcommand{\NAMNECACHE}{K-Way}
\newcommand{\LOCKSET}{KW-LS}
\newcommand{\FREEAR}{KW-WFA}
\newcommand{\COUNTER}{KW-WFSC}

\newtoggle{MW2020}
\toggletrue{MW2020}

\newtoggle{SMALL}
\togglefalse{SMALL}

\newtoggle{MEDIUM}
\toggletrue{MEDIUM}

\newcommand*{\ARXIV}{}
\newcommand*{\ICPP}{}
	
	
\title{Limited Associativity Makes Concurrent Software Caches a Breeze}
	
\author{
	Dolev Adas\\
	Computer Science\\
	Technion\\
	\texttt{sdolevfe@cs.technion.ac.il}
\and
	Gil Einziger\\
	Computer Science\\
	Ben-Gurion University of the Negev\\
	\texttt{gilein@bgu.ac.il}
\and
	Roy Friedman\\
	Computer Science\\
	Technion\\
	\texttt{roy@cs.technion.ac.il}
} 

\maketitle

\begin{abstract}
Software caches optimize the performance of diverse storage systems, databases and other software systems.
Existing works on software caches automatically resort to fully associative cache designs. 
Our work shows that limited associativity caches are a promising direction for concurrent software caches.
Specifically, we demonstrate that limited associativity enables simple yet efficient realizations of multiple cache management schemes that can be trivially parallelized.
We show that the obtained hit ratio is usually similar to fully associative caches of the same management policy, but the throughput is improved by up to x$5$ compared to production-grade caching libraries, especially in multi-threaded executions.
\end{abstract}

\section{Introduction}
\emph{Caching} is a fundamental design pattern for boosting systems' performance.
Caches store part of the data in a fast and close memory to the program's execution. 
Accessing data from such a memory, called a \emph{cache}, reduces waiting times, thereby improving running times. 
A \emph{cache hit} means that an accessed data item is already located in the cache; otherwise, it is a \emph{cache miss}. 
Usually, only a small fraction of the data can fit inside a given cache. 
Hence, a \emph{cache management scheme} decides which data items should be placed in the cache to maximize the \emph{cache hit ratio}, i.e., the ratio of cache hits to all accesses. 
More so, such a management scheme should be lightweight; otherwise, its costs would outweigh its benefits.


At coarse granularity, we can distinguish between \emph{hardware caches} and \emph{software caches}. A CPU's hardware cache uses fast SRAM memory as a cache for the slower main memory, which is made of DRAM~\cite{LRU}. 
In addition, the main memory (DRAM) is often utilized as a cache for secondary storage (disks and SSDs) or network/remote storage.
Similarly, most storage systems include a software-managed cache, e.g., the Velocix media platform stores video chunks in a multi-layer cache~\cite{Hifi}.
Additional popular examples include storage systems like Redis~\cite{redis-lru}, Cassandra~\cite{cassandra}, HBase~\cite{hbase} and Accumulo~\cite{accumulo}, graph databases like DGraph~\cite{dgraph} and neo4j~\cite{neo4j}, etc.
Web caching is another popular example.
\nottoggle{MW2020}{
In addition, caches are often hierarchical, e.g., L1, L2, and L3 hardware caches and multiple levels of \emph{proxy caches} in the case of network software cache.
}

An important aspect of the cache management scheme is to select a \emph{victim} for eviction whenever there is not enough room in the cache to admit a newly accessed data item.
In \emph{fully associative} caches, the cache management schemes can evict any of the cached items, as illustrated in Figure~\ref{fig:motivation}. 
In contrast, \emph{limited associativity} caches restrict the selection of data items that may be evicted.
Specifically, limited associativity caches are portioned into independent \emph{sets}.
Each set consists of $k$ \emph{ways} (places), each of which can store a data item. 
When admitting an item to a limited associativity cache, one must select the cache victim from the same set as the newly admitted item.
In such caches, known as \emph{k-way set associative caches}, items' IDs are mapped into sets using a hash function and each set is an \emph{independent} sub-cache of size $k$.

In hardware, limited associativity caches are simpler to implement, consume less power to operate and are often significantly faster than fully associative designs.
Thus, the vast majority of hardware-based caches have limited associativity. 
In contrast, non-distributed software-based caches are almost always fully associative. 
Specifically, cache implementations usually employ hash tables as building blocks. 
This \nottoggle{MW2020}{design choice} is often so obvious that papers describing software caches do not even explicitly mention that their algorithms are fully associative~\cite{Hyperbolic,EFM17,ARC,AdaptiveCacheReplacement}. 
A possible reason for this lack of previous consideration stems from the understanding that fully associative caches should yield higher hit ratios than the same sized limited associativity caches.
Further, it is relatively easy to implement such caches from fully associative building blocks such as concurrent hash tables.

\begin{figure}[t]
	\center{
\ifdefined\ARXIV
		\subfigure[Fully Associative Cache]{\includegraphics[width=0.35\columnwidth]{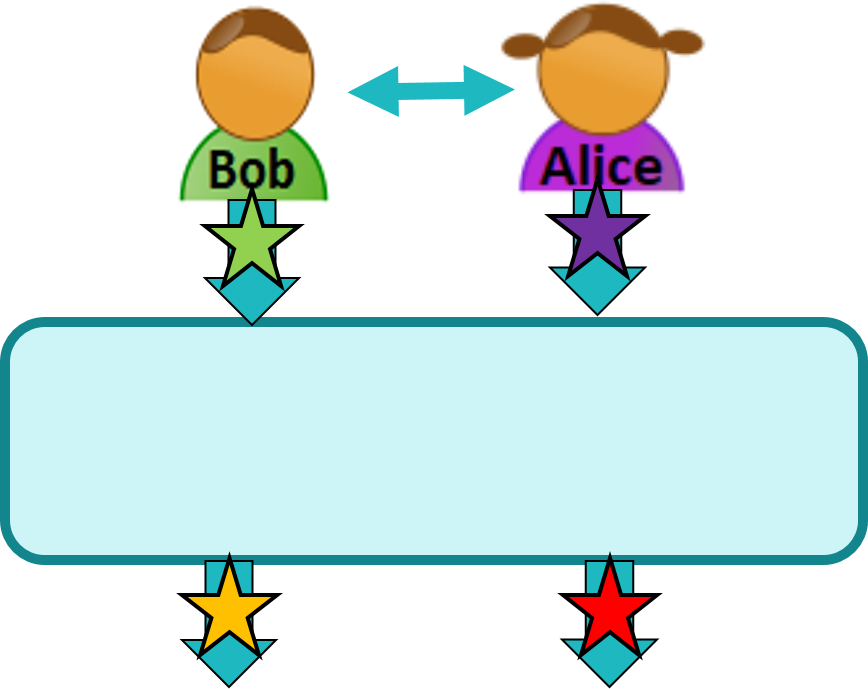}}\hfill
		\subfigure[\mbox{Limited Associativity Cache}]{\includegraphics[width=0.35\columnwidth]{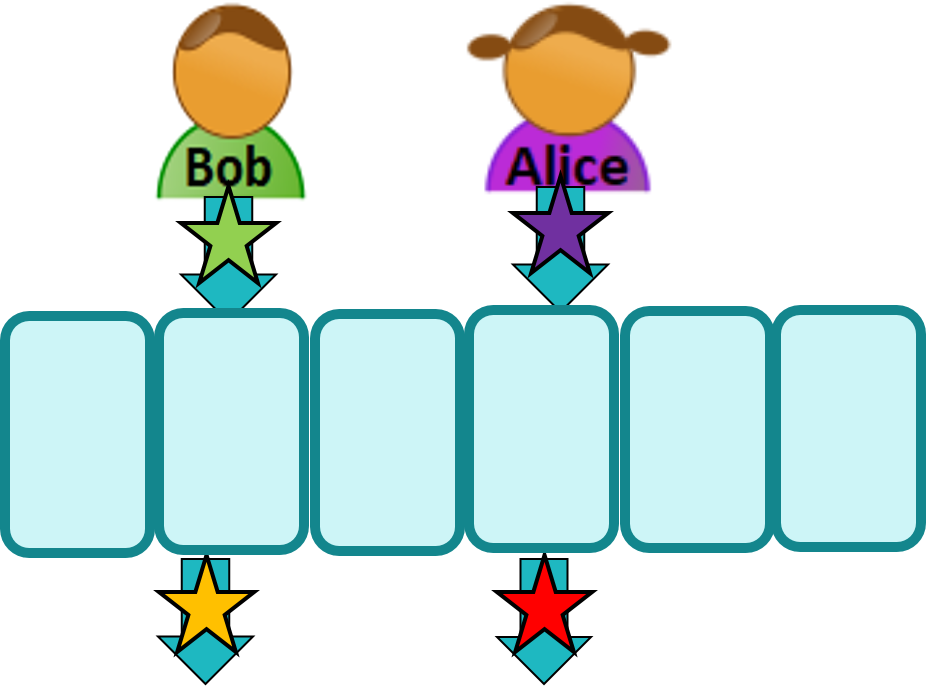}}
\else
		\subfigure[Fully Associative Cache]{\includegraphics[width=0.45\columnwidth]{FullyAssociativeCache}}\hfill
		\subfigure[\mbox{Limited Associativity Cache}]{\includegraphics[width=0.45\columnwidth]{K-WayCache}}
\fi
	}
	\caption{A schematic overview of a fully associative cache (a) and a limited associativity cache (b). In a fully associative cache, Alice and Bob operate on a shared resource and require some form of synchronization between them. In the limited associativity case, the cache is partitioned into many independent sub-caches, also known as sets, and updates are assigned to sets through hashing. Most times, Alice and Bob operate on different sets requiring no synchronization at~all.}
	\label{fig:motivation}	
\end{figure}

However, our work argues that software caches significantly benefit from a limited associativity design due to the following advantages:
First, limited associativity caches are embarrassingly parallel and allow for concurrent operations on different sets without any form of synchronization (as illustrated in Figure~\ref{fig:motivation}).
In contrast, it is far from trivial to introduce concurrency to most fully associative cache policies as often these cache management policies rank cached data items and then choose the ``worst'' item for eviction, c.f.,~\cite{SLRU,SurveyOfCacheReplecmentStrategies}. 
Thus, parallel actions may select the same cache victim or update the ranks of cached data items.
For example, in LRU, we get contention on the head of the LRU list~\cite{FM2020-fast}.
Aside from parallelism, cache implementations rely on elaborate data structures that often require an excessive number of pointers.
Also, hash tables almost always have a constant fraction of unused space. 
In comparison, limited associativity caches offer a denser representation of data. When $k$ is reasonably small, we can scan all the items in the set without requiring any auxiliary data structures.  

The above limitations of full associativity have motivated the development of reduced accuracy management policies such as Clock~\cite{clock,clock-pro}, sampled LRU~\cite{redis-lru} and sampled dynamic priority based policies~\cite{Hyperbolic,LHD}.
As elaborated in this work, we claim that limited associativity is a beneficial design alternative compared to Clock and sampling, especially for parallel systems.

\subsection{Contributions}

In this work, we explore and promote limited associativity designs for software caches.
Specifically, our first contribution is the implementation of several leading cache replacement policies in a limited associativity manner.
As we show, limited associativity leads to simplified data structures and algorithms.
These can be trivially implemented in static memory with minimal overhead and operate in constant time.
Further, these realizations trivially support wait-free concurrent manipulation and avoid hot-spots.

Our second contribution is an evaluation of the hit-ratio obtained by multiple cache management policies on real traces for $k$-way set associative caches and for fully associative caches.
Our findings indicate that, even for relatively small values of $k$, the difference between the hit-ratio obtained by a $k$-way associative scheme and the corresponding fully associative scheme is marginal.

Since the runtime of each operation depends on the size of each set, low associativity is preferred for speed. 
Low associativity also implies more independent sets which reduces the contention on each set, and allows for better parallelism. 
However, this is a tradeoff as associativity also impacts the hit-ratio and if we set the associativity too low then we would suffer a non-negligible drop in hit-ratio.  
Our work demonstrates that a reasonable associativity value of $8$ provides the best of both worlds as it yields a very similar hit-ratio to fully associative caches, while being considerably faster. 

These results echo previous studies from hardware caching~\cite{LRU,HS89}, where limited associativity was found to be nearly as effective as full associativity.
However, those studies were mostly conducted using hardware workloads, which tend to exhibit higher locality than software workloads.

We stress that limited associativity is a design principle rather than a cache management policy.
Hence, we can adjust most existing policies to their respective limited associativity version with small changes. 
This assertion is especially straightforward for sampling based policies such as sampled LRU~\cite{redis-lru}, Hyperbolic~\cite{Hyperbolic} and LHD~\cite{LHD} that randomly select a small number of cached items and choose the victim only from these items.  Notice that limited associativity only requires computing a single hash function per cache miss compared to invoking the PRNG multiple times in sampled approaches.
The sampled approach also requires accessing multiple random memory locations, which is not friendly for the hardware cache.
In contrast, in limited associativity, we access a short continuous region of memory.

Further, contemporary cache management schemes, including ARC~\cite{ARC}, LIRS~\cite{LIRS}, FRD~\cite{FRD} and W-TinyLFU~\cite{EFM17,EEFM18} maintain two or more cache regions, each of which handled in a fully associative manner.
We argue that each cache region could be treated as a corresponding limited associativity region for these schemes.

Last, we evaluate the speedups obtained by our concurrent implementations compared to leading alternatives.
These demonstrate the performance benefit of limited associativity designs for software caches in modern multi-core computers.

\paragraph*{Paper Organization}
Section~\ref{sec:related} surveys related work and provides the background necessary to position our work within the broader context of relevant designs. 
Section~\ref{sec:tiny} suggests several ways to implement $k$-way set associative caches efficiently in software. 
\ifdefined\ICPP
\else
Section~\ref{sec:analysis} provides a brief analysis that establishes that limited associativity caches can mimic (smaller) fully associative caches. 
\fi
Section~\ref{sec:eval} shows an evaluation of the hit ratio obtained by limited associativity caches for a wide selection of caching algorithms and of their throughput compared to the Guava and Caffeine Java caching libraries. 
We conclude with a discussion in Section~\ref{sec:discussion}. 

\section{Related Work}
\label{sec:related}

\subsection{Cache Management Policies}
A cache management scheme's design often involves certain (possibly implicit) assumptions on the characteristics of the ``typical'' workload. 
For example, the \emph{Least Recently Used (LRU)} policy always admits new items to the cache and evicts the least recently used items~\cite{LRU}. 
LRU works well for \emph{recency biased workloads}, where recently accessed items are most likely to be accessed again.

Alternatively, the \emph{Least Frequently Used (LFU)} policy, also called \emph{Perfect LFU}, assumes that the access distribution is fixed over time, meaning that frequently accessed data items are more likely to be reaccessed.
Hence, Perfect LFU evicts the least frequently used item from the cache and admits a new one if it is more frequent than the cache victim.
For synthetic workloads with static access distributions, Perfect LFU is the optimal cache management policy~\cite{WLFU}.
Real workloads are dynamic, so practical LFU policies often apply aging mechanisms to the items' frequency. 
Such mechanisms can be calculating frequency on sliding windows~\cite{WLFU}, using exponential decay~\cite{LFUAGING,LFUDA}, or periodic aging~\cite{EFM17}.

Operation complexity is another essential factor in cache design.
Both LRU and LFU need to maintain a priority queue or a heap, and can be implemented in constant time~\cite{LFUIMPl,SpaceSavings,WCSS}. LRU and variants of LFU also maintain a hash table for fast access to items and their metadata.
Another limitation of priority queue based LRU implementations is the high contention on the head of queue due to the need to update it on each cache access (both hits and misses).

The realization that a good cache management mechanism needs to combine recency and frequency led to the development of more sophisticated approaches~\cite{FM2020-fast}.
LRU-K~\cite{LRUK} maintains the last $K$ occurrences of each item and orders all items based on the recency of their last $K^\mathrm{th}$ access.
The drawbacks are significant memory and computational overheads.
Also, most benefit comes when $K=2$, with quickly diminishing returns for larger values of $K$.
2Q~\cite{2Q} and SLRU~\cite{SLRU} approximate LRU-2 by maintaining two fixed size LRU segments, a \emph{probation segment} and a \emph{protected segment}.
New items are always admitted into the probation segment and when an item in the probation segment is accessed again, it is moved to the protected segment.
An insertion to a full probation segment results in the least recently used item (in that segment) being evicted from the cache.
Similarly, any insertion to a full protected segment results in moving its LRU victim to the probation segment (which may then select its own LRU victim).

\emph{Adaptive Replacement Cache} (ARC)~\cite{ARC} also maintains two LRU lists, but addresses their respective ratio using a dynamically adaptive mechanism.
Further, each of these LRU lists maintains an additional LRU lists of \emph{ghost entries}, which include only the meta-data of recently evicted items.
A cache miss which is a hit in one of the ghost entries indicates that its corresponding  LRU list should be enlarged at the expense of the other.
Ghost entries have become very common in other state-of-the-art cache management policies.
\emph{Clock with Adaptive Replacement} (CAR)~\cite{car} replaces the LRU caches in ARC with Clock caches to mitigate the contention problems caused by LRU.

\emph{Low Inter-reference Recency Set} (LIRS)~\cite{LIRS} is a page replacement algorithm that attempts to directly predict the next occurrence of an item using a metric named \emph{reuse distance}.
To realize this, LIRS maintains two cache regions and a large number of ghost entries.
FRD~\cite{FRD} can be viewed as a practical variant of LIRS, which overcomes many of the implementation related limitations of~LIRS.

\iftoggle{MW2020}
{
\emph{Window-TinyLFU} (W-TinyLFU)~\cite{EFM17} also balances between recency and frequnecy using two cache regions, a \emph{window cache} and the \emph{main cache}.
Yet, it controls which items are admitted to the main cache using an optimized counting Bloom filter~\cite{CountingBloom} based admission filter and thereby avoids the need to maintain ghost entries.
The relative size of the two caches is dynamically adjusted at runtime~\cite{EEFM18}.
}
{
\emph{TinyLFU}~\cite{TinyLFU} is an LFU like admission policy that avoids the use of ghost entries. 
The goal of TinyLFU is only to admit a new item only if it is more frequent than the cache victim. 
Instead of ghost entries, TinyLFU estimates the frequency of each element using a counting Bloom filter~\cite{CountingBloom}. 
\emph{Window-TinyLFU} (W-TinyLFU) extends TinyLFU by utilizing two caches; a TinyLFU managed SLRU cache and an LRU cache.
New items are always inserted into the LRU cache, which is also called \emph{window cache} and the victim of the window cache is suggested to the TinyLFU managed cache. In the \emph{Hill Climber W-TinyLFU (HC-W-TinyLFU)}~\cite{EEFM18} we use a simple hill climber to periodically reconfigure the sizes of the LRU and the SLRU caches. 
}

As we mentioned before, what is common to ARC, CAR, LIRS, SLRU,  2Q, FRD and W-TinyLFU is that they all maintain two or more caches regions, each of which is maintained as a fully associative cache.
We promote implementing each such cache \iftoggle{MW2020}{region}{region, regardless of the overall high level policy,} as a limited associativity cache.

Hyperbolic caching~\cite{Hyperbolic} logically maintains a dynamically evolving priority for each cached item, which is the number of times that item has been accessed divided by the duration of time since the item was inserted into the cache.
To make this practical and efficient, priorities are only computed during the eviction and only for a small sample of items.
Among the sampled items, the one whose priority is smallest becomes the victim.
As we discuss later, this policy can be trivially implemented using an associative design.

\nottoggle{MEDIUM}{
Other variants of the problem include different size items~\cite{SIZE,berger2017adaptsize,LHD,LruSP,GDSF,GD}.
For example, SIZE~\cite{SIZE} offers to remove the largest item first, while LRU orders the same size items. 
Alternatively, LRU-SP~\cite{LruSP} weighs both the size and the frequency of an item when picking a cache victim.
}


\subsection{Associativity in Caches}
The only limited associativity software cache we are aware of is HashCache~\cite{hashcache}.
Yet, that work focused on very memory constraint Web caches in which the meta-data cannot fit in main memory and the goal is to minimize disk accesses to the meta-data.
Further, it only considered the LRU replacement policy and only very particular workloads.
In this paper, we perform a more systematic study of limited associativity's effectiveness to software cache design.

\nottoggle{SMALL}{
	Current caching libraries 
	use fully associative caches. 
	Yet, in principle, any distributed cache (e.g., memcacheD) can be viewed as relying on limited associativity, since each node manages its caches independently to the others.
	When an item is added to a given cache, the victim must be selected from the same host.
	The main difference though is that in distributed caching, each host is typically responsible for a very large range of objects (equivalent to having a very large $k$ value).
	In contrast, in $k$-way associativity, our goal, as explained before, is to have a large number of small sets (low $k$ values).
	Further, \nottoggle{MW2020}{hardware caches such as} the}{The} CPU's L1, L2 and L3 caches are almost always implemented through limited associativity, with the exception of TCAM based caches. 
In hardware, it is clear that fully associative caches require fewer circuits to implement and that they work faster. 
\ifdefined\ICPP
\else
In particular, techniques like d-left hashing~\cite{dlCountingBloom} can be seen as efficient ways to simulate a fully associative structure from limited associativity building blocks. 
\fi

\subsection{On Hash Tables and Caches}
As mentioned above, most fully associative software caches use fully associative hash tables to quickly find the cached items. 
In principle, hash tables incur non-negligible overheads. 
Specifically, open addressing hash tables such as the classic Linear Probing, or the modern Hopscotch hash table~\cite{hopscotchhashing}, are implemented on top of partially full arrays. 
In all open address hash tables, a constant fraction of these arrays must remain empty. 
Thus, relying on such tables results in underutilized space that could potentially store useful items.  
Closed address hash tables such as Java's concurrent hash map~\cite{javaMap} are another alternative.
Yet, they incur significant overheads for pointers, which results in a suboptimal data density. 

In contrast, limited associativity caches can be implemented using the same arrays as open address hash tables, but we can use 100\% of the array's cells. 
In principle, we do not even have to use pointers and achieve a very dense representation of data. 
Such a trade-off is potentially attractive since our caches can store more items within a given memory unit and having more items increases the cache hit ratio. 
On the other hand, there is less freedom in selecting the cache victim, which may reduce the hit ratio.

Let us emphasize that hashtables and caches are not the same. 
While a hashtable can be used to store cached objects, one would still need additional data-structures and algorithms to implement the cache management policy, and this is the main source of pain in software caches.
On the other hand, caches are allowed to evict items while hashtables retain all the items that were not explicitly removed. 

\subsection{Parallel Caches}
Caches need to operate fast and thus parallelism is a vital consideration. 
Specifically, LRU is known to yield high contention on the head of the list as a simple implementation moves elements to the head of the list on each access. 
\iftoggle{SMALL}{
	This contention motivated alternative approximate LRU designs such as Clock~\cite{clock} and MemC3~\cite{MemC3}.	
}
{
	This contention motivated alternative designs such as Clock~\cite{clock}.
	In Clock, each entry receives a single bit, indicating whether it was recently accessed. 
	Items are always admitted by Clock and the cache victim is the first item with an unset Clock bit.
	The Clock hand clears set bits as it searches for the cache victim. 
	Although the worst-case complexity of Clock is linear, it parallelizes well and creates little contention and is therefore used in practice as an operating system page replacement cache.
	Another interesting approach is taken by~\cite{MemC3}, which uses Bloom filter based approximations of LRU to enable better concurrency.
}

Fully associative cache management often relies on a hash table that needs to be parallelized. 
Such tables are a shared resource and are usually not embarrassingly parallel~\cite{michael2002high},~\cite{MaierSD19}. 
On the other hand, limited associativity caches are embarrassingly parallel as actions on different sets are completely independent and require no synchronization. 
This property holds regardless of the (limited associativity) eviction policy.  

\section{Limited Associativity Architecture}
\label{sec:tiny}

A \NAMNECACHE{} cache supports two operations: get/read and put/write, as customary in other caches such as Caffeine~\cite{CaffeineProject} and Guava~\cite{guava-cache}.
A get/read operation retrieves an item's value from the cache or returns null if this item is not in the cache.
The put/write operation inserts an item into the cache; if the item already exists in the cache, this operation overwrites its value.
Both types of operations update the item's metadata that is used by the management policy.

\begin{figure}[t]
	\center{
\ifdefined\ARXIV
		\includegraphics[width=0.7\columnwidth]{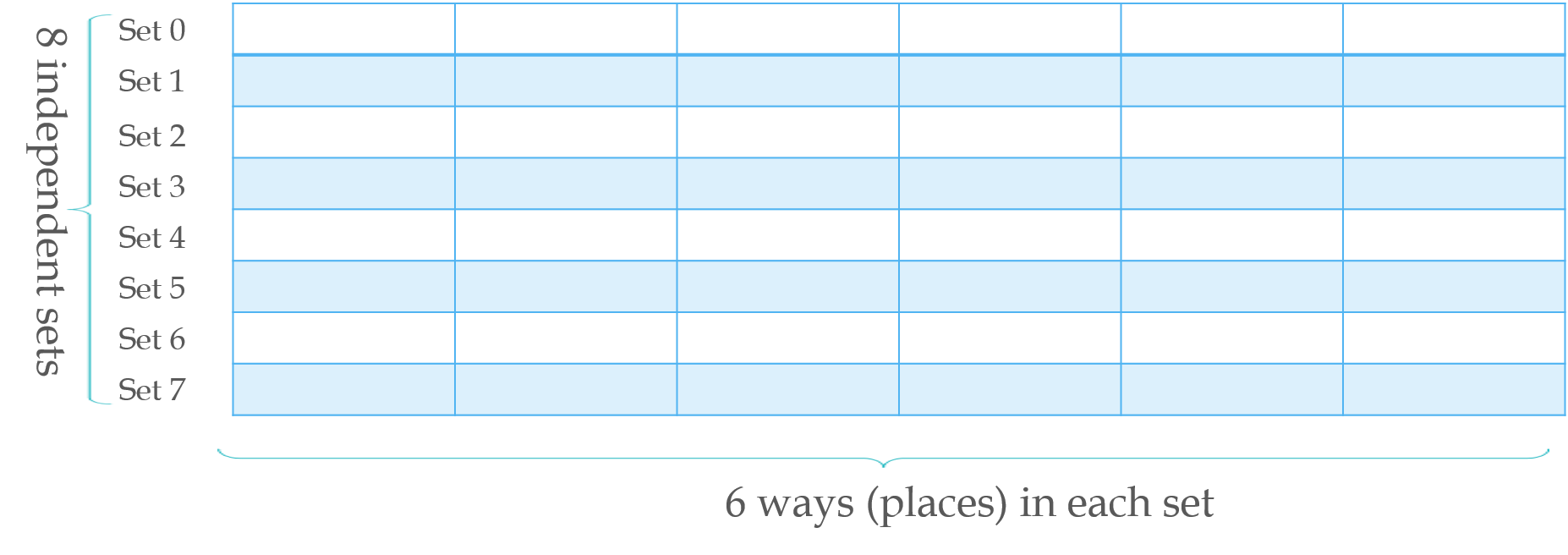}
\else
		\includegraphics[width=\columnwidth]{kway-organization}
\fi
	}
	\vspace{-0.5cm}
	\caption{An example of a 48 items cache organized in a 6-way set associative design. The cache is broken into 8 independent sets such that each contains a up to 6 items.  }
	\label{fig:k_way}	
\end{figure}

Figure~\ref{fig:k_way} illustrates such cache organization consisting of $8$ independent sets that are held in an array.
Each set is an array of $6$ ways (nodes) that store items with metadata that may be required by various eviction policies.

Upon insertion of an item to the cache, we use a hash function to determine its set.
If the set is full, we also need to select a victim from the same set to be evicted. 
\NAMNECACHE{} supports efficient implementations of numerous eviction policies without relying on expensive auxiliary data structures. 
For example, to implement LRU or LFU, we only need a short counter for each cached item in the set. 
We select the victim by reading all $K$ items in the set and evicting the one with the smallest counter. 
The difference between LRU and LFU is merely in the way the counter is updated. 
In LRU, we use it to record the (logical) timestamp of the last access, whereas in LFU, we use it to count how many times the entry was accessed in the past. 
To read an item, we scan the appropriate set. 
Once found, we update the corresponding item's metadata (updating the access time and frequency) atomically and returning the~value.

Clearly, the \NAMNECACHE{} time complexity of both get/read and put/write is O($K$), which is only dependent on the number of ways ($K$).
As $K$ is a small constant number, we treat it as constant time. 

Similarly, for Hyperbolic caching~\cite{Hyperbolic} we maintain two short counters for each item, one recording the time it was inserted ($t_0$) and the other ($n$) counting \iftoggle{SMALL}{its}{the} number of \iftoggle{SMALL}{accesses}{times it was accessed} (initialized to $1$).
For an eviction at time $t_e$, we scan the $K$ items at the corresponding set and select the item whose $n/(t_e-t_0)$ value is~minimal.

We implemented \NAMNECACHE{} in Java.
It supports five eviction policies: LRU, LFU, FIFO, Random and Hyperbolic.
The eviction policy is chosen at the constructor.
We have realized several concurrency control variants as detailed below.

\paragraph{$K$ Way Cache Wait Free Array - \FREEAR{}}
Each set is an array of Node references.
Algorithm~\ref{alg:ArrayClass} depicts the internal classes and variables.
Pseudo-code for the get operation is given in Algorithm~\ref{alg:getArray} and for the put operation in Algorithm~\ref{alg:putArray}.
To replace a victim node with a newly arriving one, we perform a CAS (compare and swap) operation on the address of this node. 
 
\begin{algorithm}[t]
	\caption{Internal \FREEAR{} classes and fields}
	\label{alg:ArrayClass}
	\scriptsize
	\begin{algorithmic}[1]
		\State class Node $\{$ \label{line:node}
		\State K \ key;
		\State V  value;
		\State AtomicInteger  counter;
		\State int  index;  $\}$
	
		\State class  set $\{$
		\State AtomicReferenceArray$<$Node$<$K, V$>$$>$  setArray;
		\State  AtomicLong  time;  $\}$  // time is only present in LRU 
	
	\end{algorithmic}
\end{algorithm}

\begin{algorithm}[t]
	\caption{\FREEAR{} get(read) operation}
	\label{alg:getArray}
	\scriptsize
	\begin{algorithmic}[1]
		\Function {get}{K key}
		\State int set =(int) hash(key) \& (numberOfSets-1); 	
			\For{int i = 0; i $<$ ways; ++i }
				\State Node$<$K, V$>$ n = cache[set].setArray.get(i);
				\If{ n != null  \&\&  n.key.equals( key)} 
					\State update(n.counter);
					\State return  n.value;
				\EndIf
			
			\EndFor	
				\State return null;			
			\EndFunction
	\end{algorithmic}
\end{algorithm}

\begin{algorithm}[t]
	\caption{\FREEAR{} put(write) operation }	
	\label{alg:putArray}
	\scriptsize
	\begin{algorithmic}[1]
		\Function {put}{K key , V value}
		\State int set =(int) hash(key) \& (numberOfSets-1));  
		
			\For{int \  i = 0; i $<$ ways; ++i}
				\State Node$<$K, V$>$ n = cache[set].setArray.get(i);
				\If{n != null \ \&\& \ n.key.equals( key)} 
					\State cache[set].setArray.compareAndSet(i,n,update);
					\State return ;
				\EndIf
			\EndFor	
		
		\State Node$<$K, V$>$ victim =policy.select(cache[set].setArray, key)

			\If{ victim != null } 
				\State {Node$<$K,V$>$ updateNode = new Node$<$K,V$>$(key,value,victim.index,cache[set].readTime());}
				\State {cache[set].setArray.compareAndSet(i,victim,update);}		
			
				\Else
					\For{int i = 0; i $<$ ways; ++i }
						\If{Node$<$K, V$>$) cache[set].setArray.get(i)== null}
							\State {Node$<$K, V$>$ updateNode = new \ Node$<$K, V$>$(key,value,victim.index,cache[set].readTime());}
						\State cache[set].setArray.compareAndSet(i ,null,updateNode);	
							\State return;	
							\EndIf	
						\EndFor	
			\EndIf	
		\EndFunction
	\end{algorithmic}
\end{algorithm}



To exchange an entire node atomically, we maintain each set of \NAMNECACHE{} as a reference array. 
Alas, this has the following shortcoming.
To search an item and select a victim, we must scan the entire corresponding set reading the metadata of each node.
With a reference array implementation, this means reading multiple different memory locations, one for each item in the set, which is expensive.

To improve this and enable scans to access continuous memory regions, we offer a second implementation in which we separated the counters and fingerprints from the nodes.

\paragraph{K Way Cache Wait Free Separate Counters - \COUNTER{}}
We store an array of counters and an array of fingerprints with the reference array of the nodes for each set.
Algorithm~\ref{alg:CounterClass} depicts the internal classes and variables.
To insert an item, we first scan the fingerprints array.
If a certain entry's fingerprint matches the key's fingerprint, we check the key inside the corresponding node.
If they are the same, we update this node with the new value and return.
If the item is new, we scan the counter array to find the lowest/highest counter, depending on the eviction policy.
We then replace the victim without accessing the node at all.
Pseudo-code for the get operation is listed in Algorithm~\ref{alg:getCounter} and for the put operation can be found in Algorithm~\ref{alg:putCounter}.

\begin{algorithm}[t]
	\caption{Internal \COUNTER{} classes and fields}
	\label{alg:CounterClass}
	\scriptsize
	\begin{algorithmic}[1]
		\State class Node $\{$
		\State  K  key;
		\State  V  value;
		\State  AtomicInteger  counter;
		\State 	int  index; $\}$
		
		\State class set $\{$
		\State AtomicReferenceArray$<$Node$<$K, V$>$$>$ setArray;
		\State AtomicIntegerArray counters;
		\State AtomicLongArray  fingerPrint;  $\}$
		
	\end{algorithmic}
\end{algorithm}

\begin{algorithm}[t]
	\caption{\COUNTER{} get operation }
\label{alg:getCounter}
	\scriptsize
	\begin{algorithmic}[1]
		\Function {get}{K key}
	\State {int set =(int) hash(key) \& (numberOfSets-1);}	
	\For{int i = 0; i $<$ ways; ++i}
		\State Long fp = cache[set].fingerPrint.get(i);
		\If{ fp.equals((Long)key)} 
			\State Node$<$K, V$>$ n = cache[set].setArray.get(i);
			\If{ n != null \&\& n.key.equals( key)} 
				\State update(n.counter);
				\State return  n.value;
			\EndIf
		\EndIf
	\EndFor	
		\State return null;		
	\EndFunction
	\end{algorithmic}
\end{algorithm}

\begin{algorithm}[t]
	\caption{\COUNTER{} put operation }	
	\label{alg:putCounter}
	\scriptsize
	\begin{algorithmic}[1]
	\Function {put}{$K \ key ,\ V \  value$}
	\State int  set =(int) hash(key) \& (numberOfSets-1));  
		
	\For{ int  i = 0; i $<$ ways; ++i }
		\State Long   fp = cache[set].fingerPrint.get(i);
		\If{ fp.equals((Long)key)} 
			\State Node$<$K, V$>$ n = cache[set].setArray.get(i);
			\If{ n != null  \&\&  n.key.equals( key)} 
				\State Node$<$K, V$>$ update = new Node$<$K, V$>$(key, value, i);
				\State  cache[set].setArray.compareAndSet(i, n, update);
				\State return;
			\EndIf
		\EndIf
	\EndFor	
		\State {Node$<$K,V$>$ victim =policy.select(cache[set].counters, key);}		
		\State {Node$<$K,V$>$ update = new Node<K,V>(key,value,victim.index,cache[set].readTime());}	
			\If{cache[set].setArray.compareAndSet(victim.index ,victim,update);}	
				\State cache[set].fingerPrint.set(victim,index,(Long)key);		
				\State cache[set].counters.set(victim,index, 0);			
			\EndIf
			
		\EndFunction
	\end{algorithmic}
\end{algorithm}

\paragraph{$K$ Way Cache Lock Set - \LOCKSET{}}
In our third implementation, we utilize one lock per set.
Algorithm~\ref{alg:LockClass} depicts the internal classes and variables. To read an item from the set, we lock the set and scan it.
If the item is in the set, we try to change the lock to a write lock to update the counter, return the value and finish. We treat writes in a similar manner.
Pseudo-code for the get operation is given in Algorithm~\ref{alg:getLock} and for the put operation in Algorithm~\ref{alg:putLock}.

\begin{algorithm}[t]
	\caption{Internal \LOCKSET{} classes and fields}
	\label{alg:LockClass}
	\scriptsize
	\begin{algorithmic}[1]
		\State class Node $\{$ 
		\State $\ \ K\ key;$
		\State $\ \ V\ value;$
		\State 	$\ \ Integer \ counter;\}$
		\State 	$\ \ int \ index;\}$
		
		\State class set $\{$
		\State $\ \   Node<K, V>[] setArray;$
		\State $\ \  StampedLock \  lock;\}$
	\end{algorithmic}
\end{algorithm}

\begin{algorithm}[t]
	\caption{\LOCKSET{} get operation }
	\label{alg:getLock}
	\scriptsize
	\begin{algorithmic}[1]
		\Function {get}{K key}
		\State int set =(int) hash(key) \& (numberOfSets-1));	
		\State long stamp =cache[set].lock.readLock(); 	
		\For{ int i = 0; i $<$ ways; ++i}
			\State Node$<$K, V$>$ n = cache[set].setArray.get(i);
			\If{n != null  \&\& n.key.equals( key)} 
				\State long stampConvert=cache[set].lock.tryConvertToWriteLock(stamp);
				\If{ stampConvert == 0  }
					\State  cache[set].lock.unlockRead(stamp);
					\State return n.value;
				\EndIf
				
				\State update(n.counter);
					\State  cache[set].lock.unlockWrite(stampConvert);
				\State return n.value;
			\EndIf
		
	\EndFor	
		\State  cache[set].lock.unlockRead(stamp);
		\State return null;		
		\EndFunction
	\end{algorithmic}
\end{algorithm}

\begin{algorithm}[t]
	\caption{\LOCKSET{} put operation }	
	\label{alg:putLock}
	\scriptsize
	\begin{algorithmic}[1]
		\Function {put}{ K key , V value}
		\State int set =(int) hash(key) \& (numberOfSets-1));
			\State long stamp =cache[set].lock.readLock();  	
			
			\For{ int i = 0; i $<$ ways; ++i}
				\State Node$<$K, V$>$ n = cache[set].setArray[i];
				\If{ n != null \&\& n.key.equals( key)} 
					\State long stampConvert=cache[set].lock.tryConvertToWriteLock(stamp);
					\If{ stampConvert == 0}
						\State  cache[set].lock.unlockRead(stamp);
						\State return ;
					\EndIf
					\State  n.value=value;
					\State update(n.counter)
					\State  cache[set].lock.unlockWrite(stampConvert);
					\State return; 
				\EndIf
	
			\EndFor

		\State Node$<$K, V$>$ victim =policy.select(cache[set].setArray, key)
				\State int \ victimIndex = 0
		\If{victim != null } 
			\State victimIndex = victim,index
				\Else
			\For{ int i = 0; i $<$ ways; ++i }
			\If{Node$<$K, V$>$ cache[set].setArray.get(i)== null}
				\State victimIndex =i
				\EndIf	
			\EndFor	
		\EndIf	
			\State Node$<$K, V$>$ updateNode = new Node$<$K, V$>$(key,value,victim.index,cache[set].readTime());
			\State cache[set].setArray[victimIndex].key=key;  		
			\State  cache[set].setArray[victimIndex]value.=value;	
			\State cache[set].setArray[victimIndex].counter=0;		
			\State cache[set].lock.unlockWrite(stampConvert);
			\State return;	
	
		\EndFunction
	\end{algorithmic}
\end{algorithm}

\ifdefined\ARXIV
\section{Formal Analysis}
\label{sec:analysis}
Assume that a fully associative cache policy would like to store $C$ items in the cache. 
We can guarantee that any specific $C$ items can be kept in a larger limited associativity cache.
Specifically, if we denote the number of sets $n$, then if $C\ge n\log(n)$ we get that the maximum load (among the $C$ items) is: $\frac{C}{n} + \Theta\left({\sqrt{\frac{C \log(n)}{n}}}\right)$.

The next theorem guarantees that a K-way set associative cache can mimic a (smaller) fully associative cache with high probability.
Thus, such caches are capable of comparable performance as their fully associative counterparts. 
Our theorem uses the following Chernoff bound (Theorem 4.4, page 66) from~\cite{Mitzenmacher2005}: 
Let $X_1,X_2,...X_n$ be independent Poission trails, with $E(X_i)=p_i$, let $X=\Sigma^{n}_{i=1} X_i$ and let $\mu = E(X)$  then: \iftoggle{MW2020}{$\Pr\left[X\ge (1+\delta)E(X)\right]\le e^{\frac{-E(X)\delta^2}{3}}.$}{$$\Pr\left[X\ge (1+\delta)E(X)\right]\le e^{\frac{-E(X)\delta^2}{3}}.$$}

\begin{theorem}
	Consider a \NAMNECACHE{} set associative cache of size $C'\ge 2C$, then 
	the probability for the cache to be able to store the desired $C'$ items is at most: $\frac{C'}{k} e^{\frac{k}{6}}$.
	\label{thm:kway}  
\end{theorem}

\begin{proof}
When storing $C$ items in a $C'$ sized limited associativity cache partitioned into $\frac{C'}{k}$ sets, each set receives $\frac{Ck}{C'}=\frac{k}{(1+\delta)}$ of those $C$ items. 
We look at an arbitrary set and denote $X_i$ the random variable of the $i^\mathrm{th}$ item. 
That is, $X_i = 1$ if the $i'th$ item is placed on the set in question and zero otherwise. 
We denote by $X$, the total number of items (among the $C$) that are placed in the set. 
By definition, we get that: $E(X) =\frac{k}{(1+\delta)}$. 

From the Chernoff bound we get: $\Pr\left[X\ge (1+\delta)E(X)\right]\le e^{\frac{k\delta^2}{(1+\delta)3}}$. 
That is, the probability for each set to not be able to store one of the items allocated to it is: $e^{\frac{k\delta^2}{(1+\delta)3}}$.
We apply the union bound on all sets and obtain that the probability that none of the sets receives more items than it can store is at most: $\frac{C'}{k} e^{\frac{k\delta^2}{(1+\delta)3}}$.
Substituting $\delta =1$ concludes the proof. 
\end{proof}

For example, Theorem~\ref{thm:kway} shows that a 64-way cache of size 200k items can store any desired 100k items with a probability of over $99\%$. 
Alternatively, we can store in a 2M sized 128 way set associative cache any 1M items with a probability of over $99.999\%$. 

In practice, even when limited associativity caches miss out on some desired items, they store enough other items to yield comparable hit ratios as same size fully associative caches.
This is due to the following two reasons:
First, the above bound is not tight.
Second, it is quite reasonable that slight content divergence between caches would still result in similar hit ratios.
That is, the exact cache content is just a mean to obtain good hit ratio, not a goal by itself.


\label{anal:OperationComplexity}

\fi

\section{Evaluation}
\label{sec:eval}
\begin{figure}[t]
	\center{
		\includegraphics[width=0.6\columnwidth]{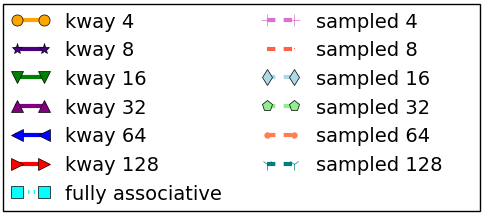}}
	
	\caption {Legend for all graphs. }
	\vspace{-0.3cm}
	\label{fig:legend}
\end{figure}

\begin{figure*}[t]
	\begin{center}
		\offinterlineskip
	\subfigure[LRU]{\includegraphics[width=0.45\columnwidth]{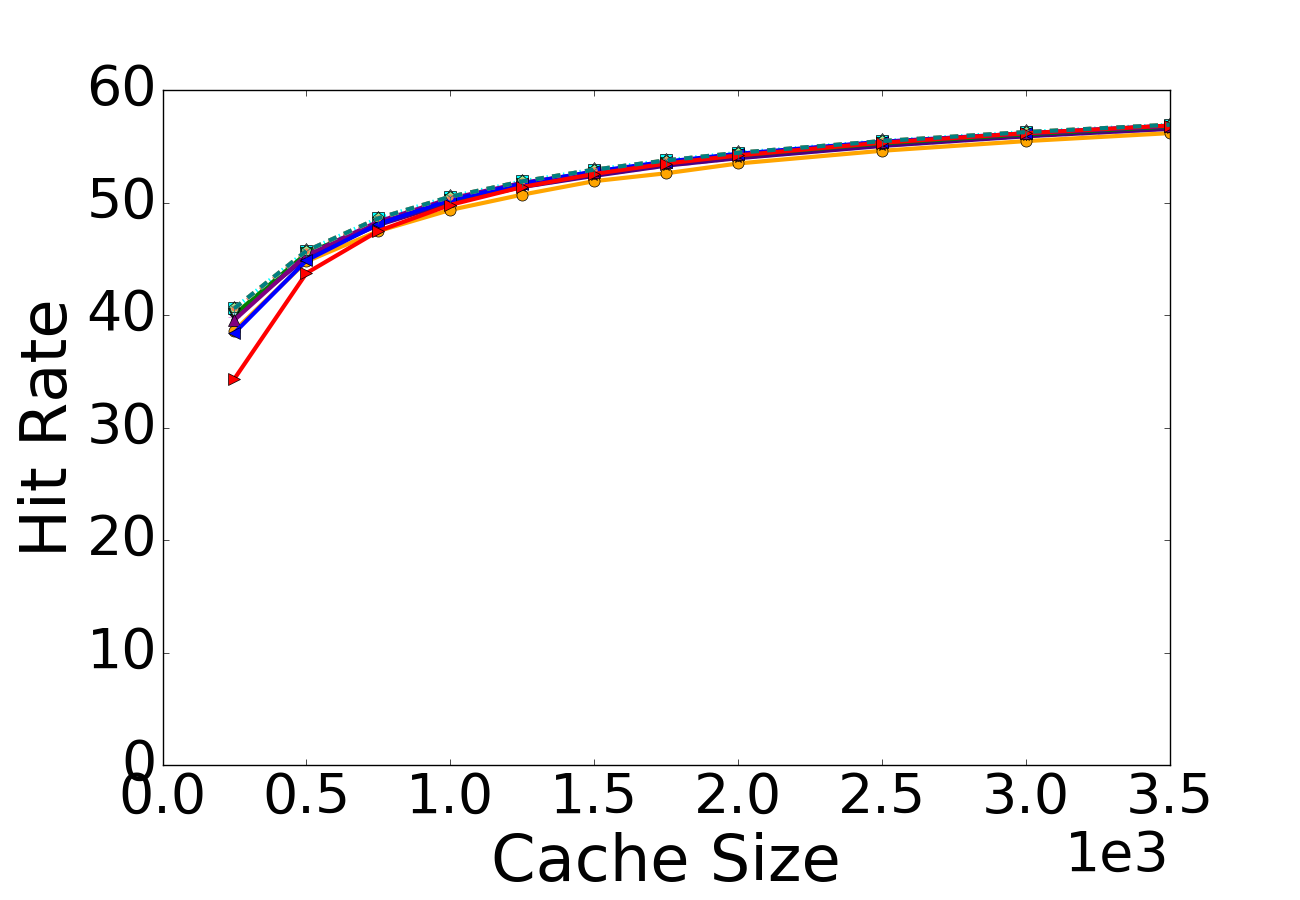}}	
	\subfigure[LFU +TinyLFU]{\includegraphics[width=0.45\columnwidth]{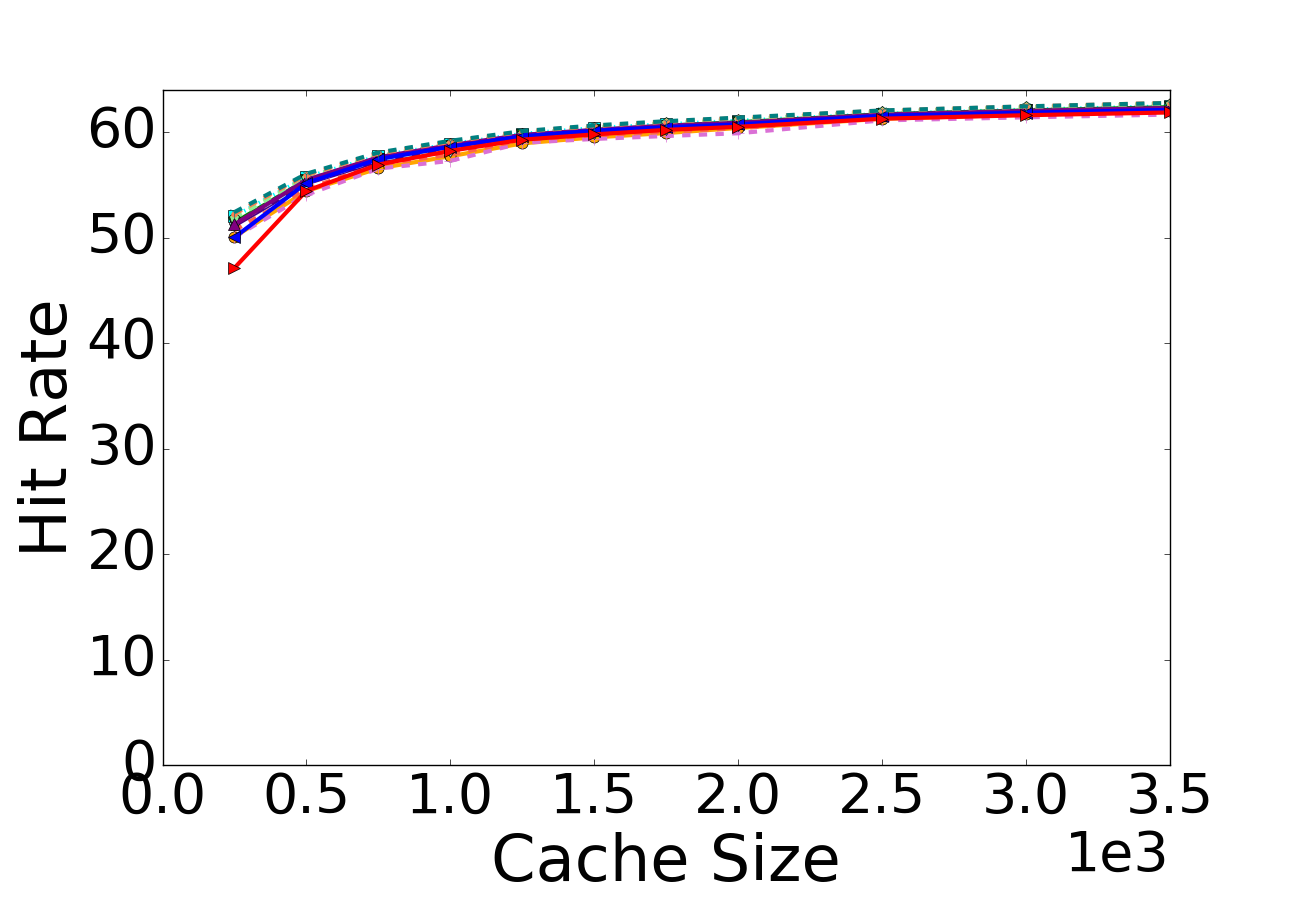}}
	\subfigure[Product]{\includegraphics[width=0.45\columnwidth]{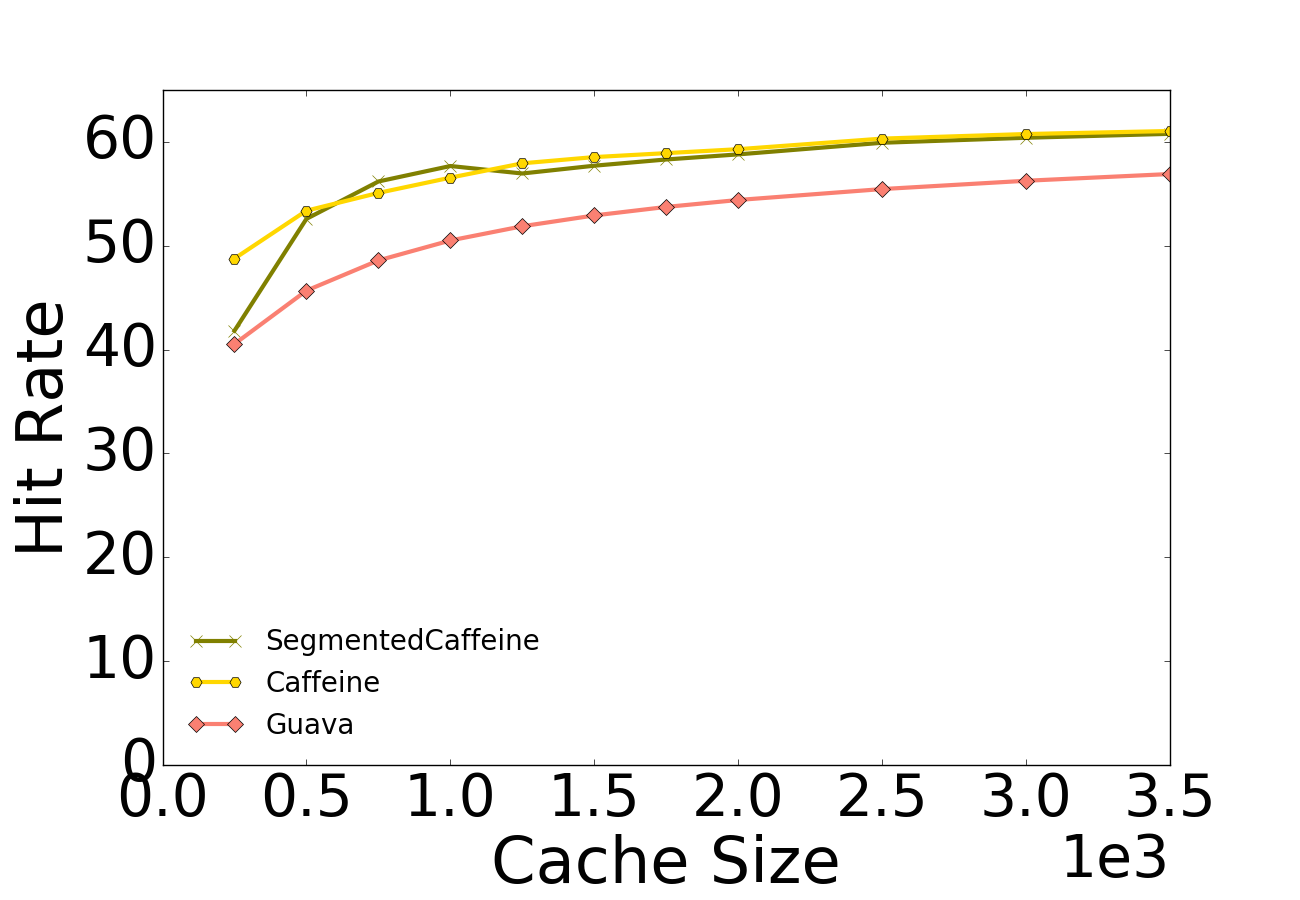}}
	\subfigure[LRU +TinyLFU]{\includegraphics[width=0.45\columnwidth]{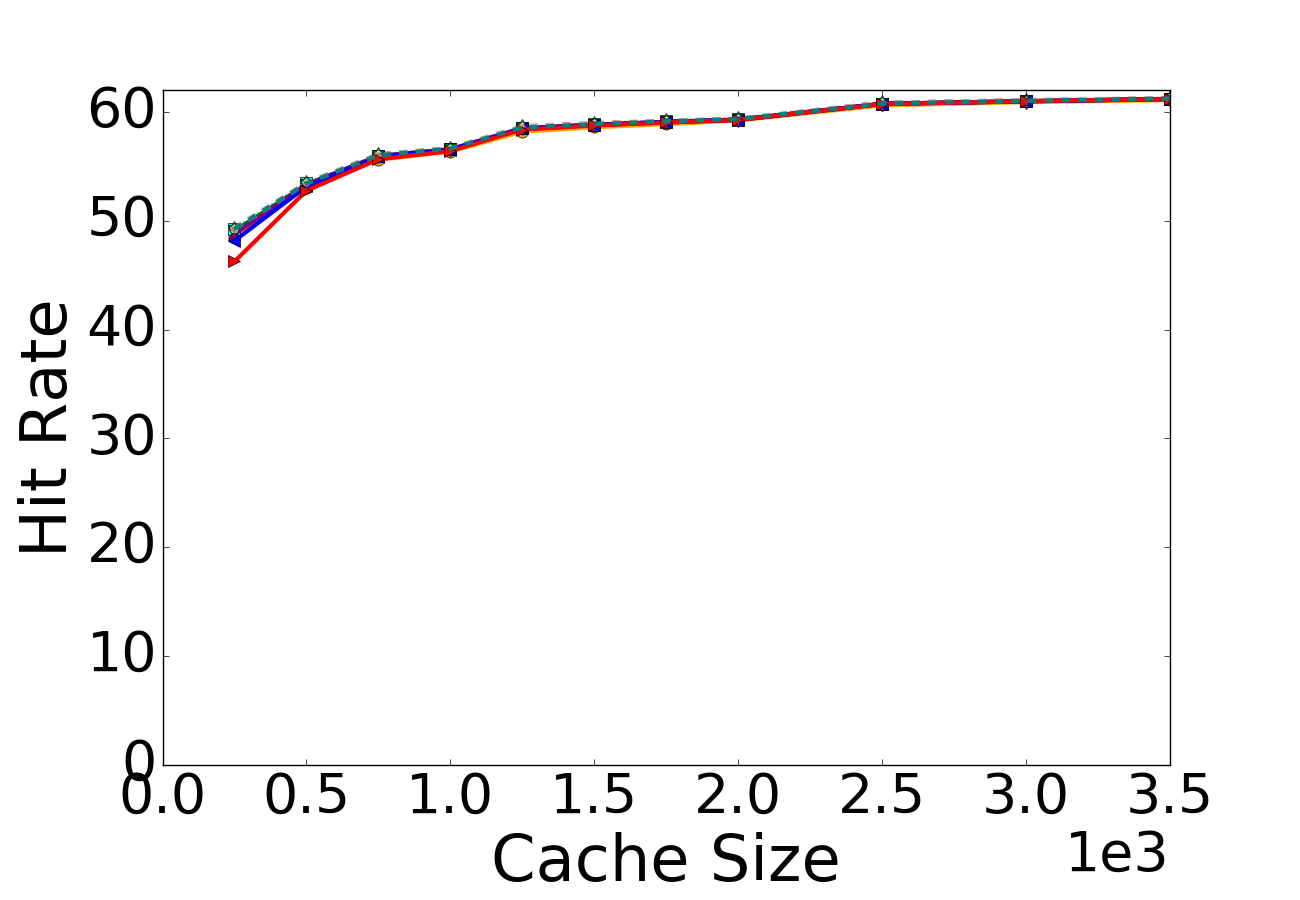}}
	\end{center}
\vspace{-0.5cm}
\caption{wiki1190322952.}
\label{fig:wiki1190322952}
\vspace{-0.5cm}
\end{figure*}
\nottoggle{MEDIUM}{
\begin{figure*}[t]
	\begin{center}
	\offinterlineskip
	\subfigure[LRU]{\includegraphics[width=0.45\columnwidth]{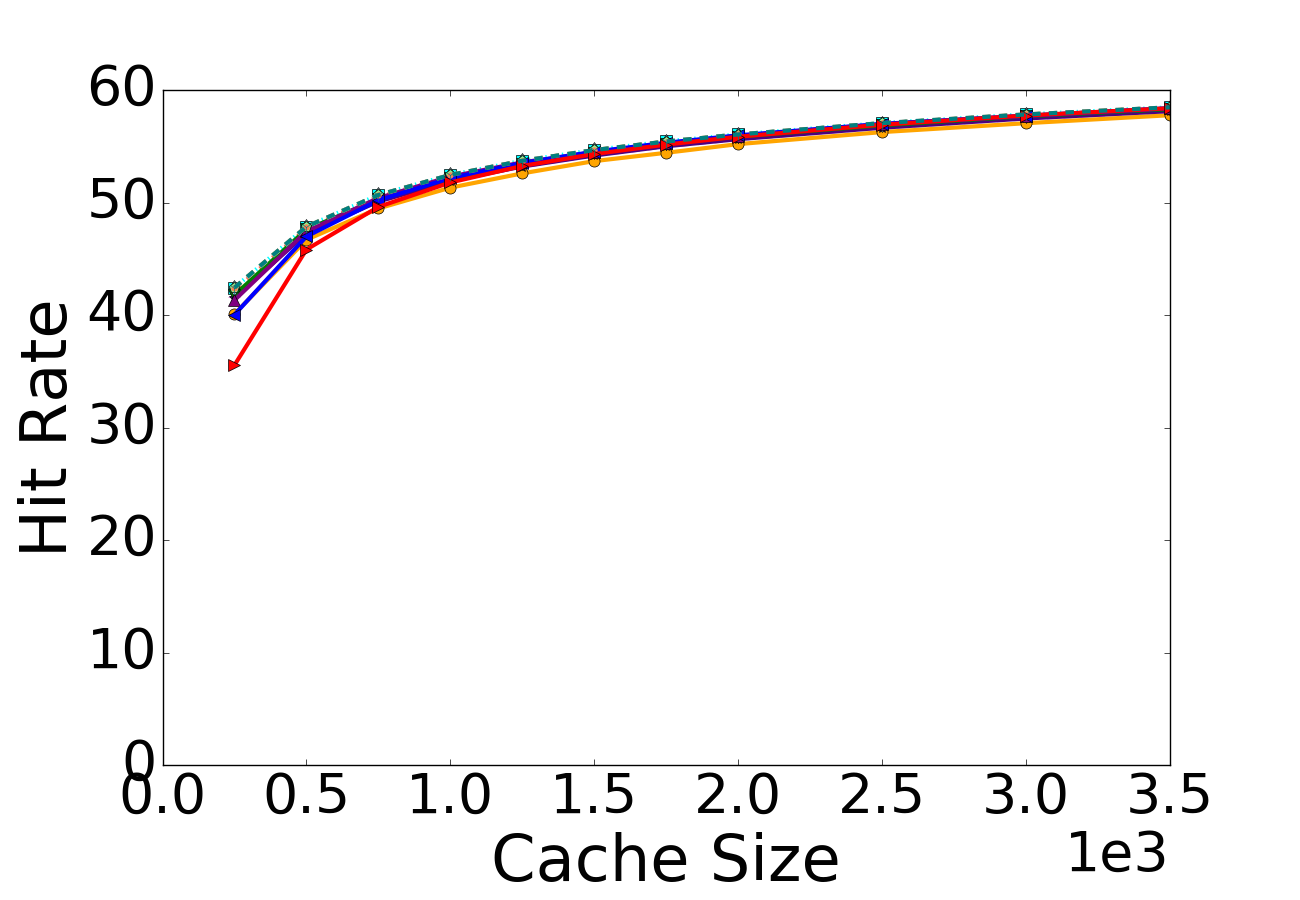}}
	\subfigure[LFU +TinyLFU]{\includegraphics[width=0.45\columnwidth]{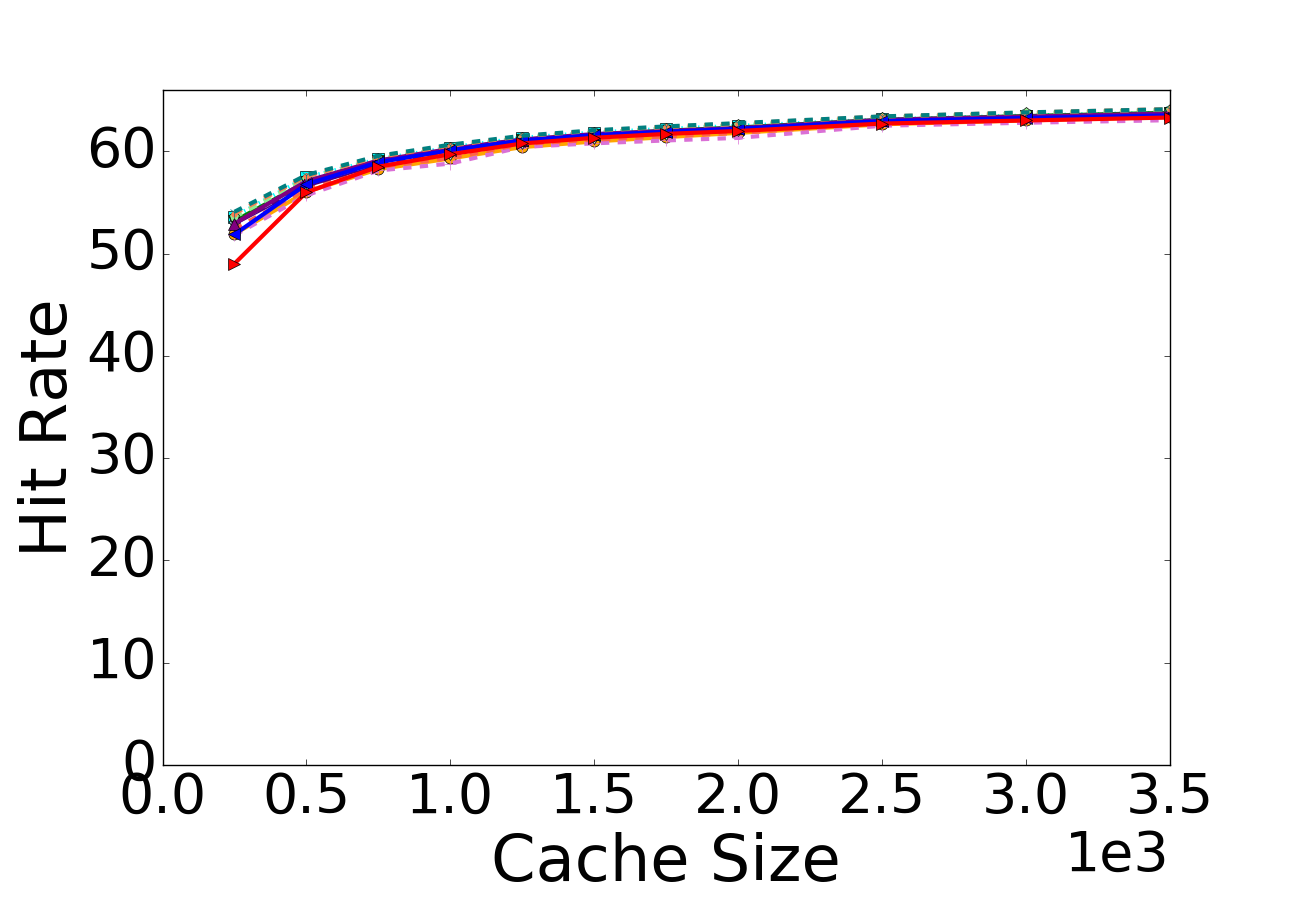}}
	\subfigure[Product]{\includegraphics[width=0.45\columnwidth]{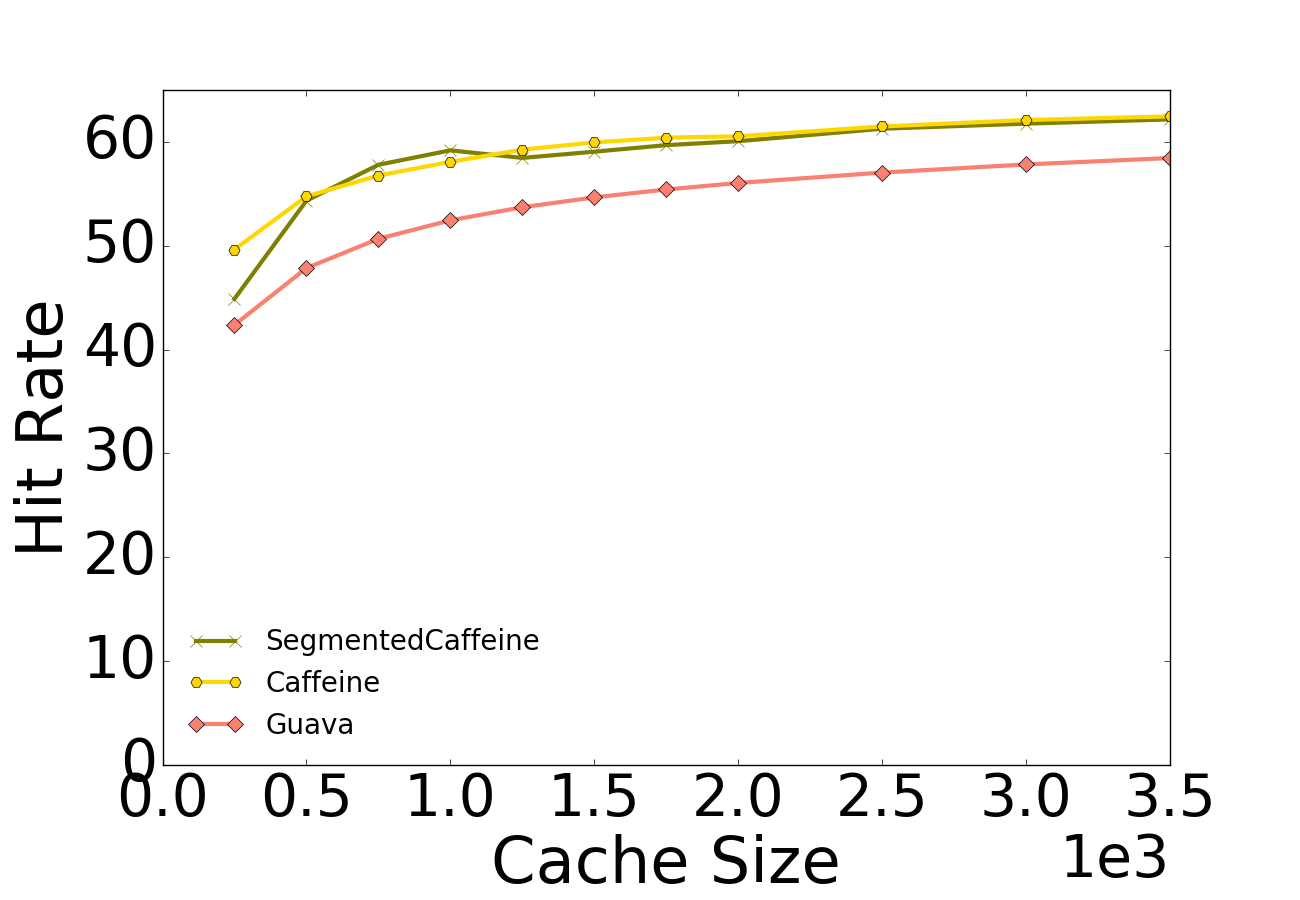}}	
	\subfigure[LRU +TinyLFU]{\includegraphics[width=0.45\columnwidth]{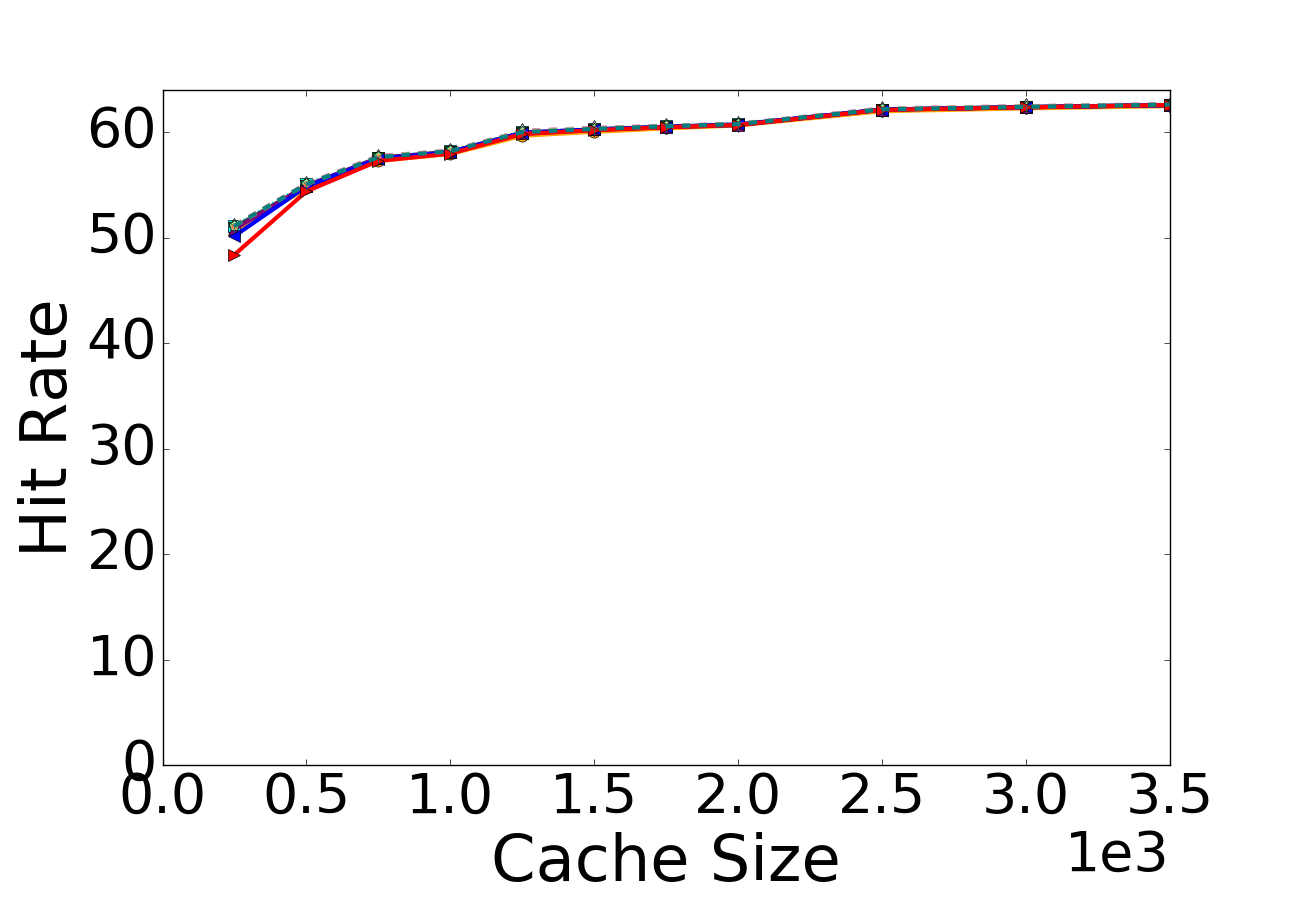}}	
	
		
	\end{center}
	\vspace{-0.5cm}
	\caption{wiki1191277217.}
	\label{fig:wiki1191277217}
	\vspace{-0.5cm}
\end{figure*}
}
\nottoggle{MEDIUM}{
\begin{figure*}[t]
	\begin{center}
		\offinterlineskip
		\subfigure[LRU]{\includegraphics[width=0.45\columnwidth]{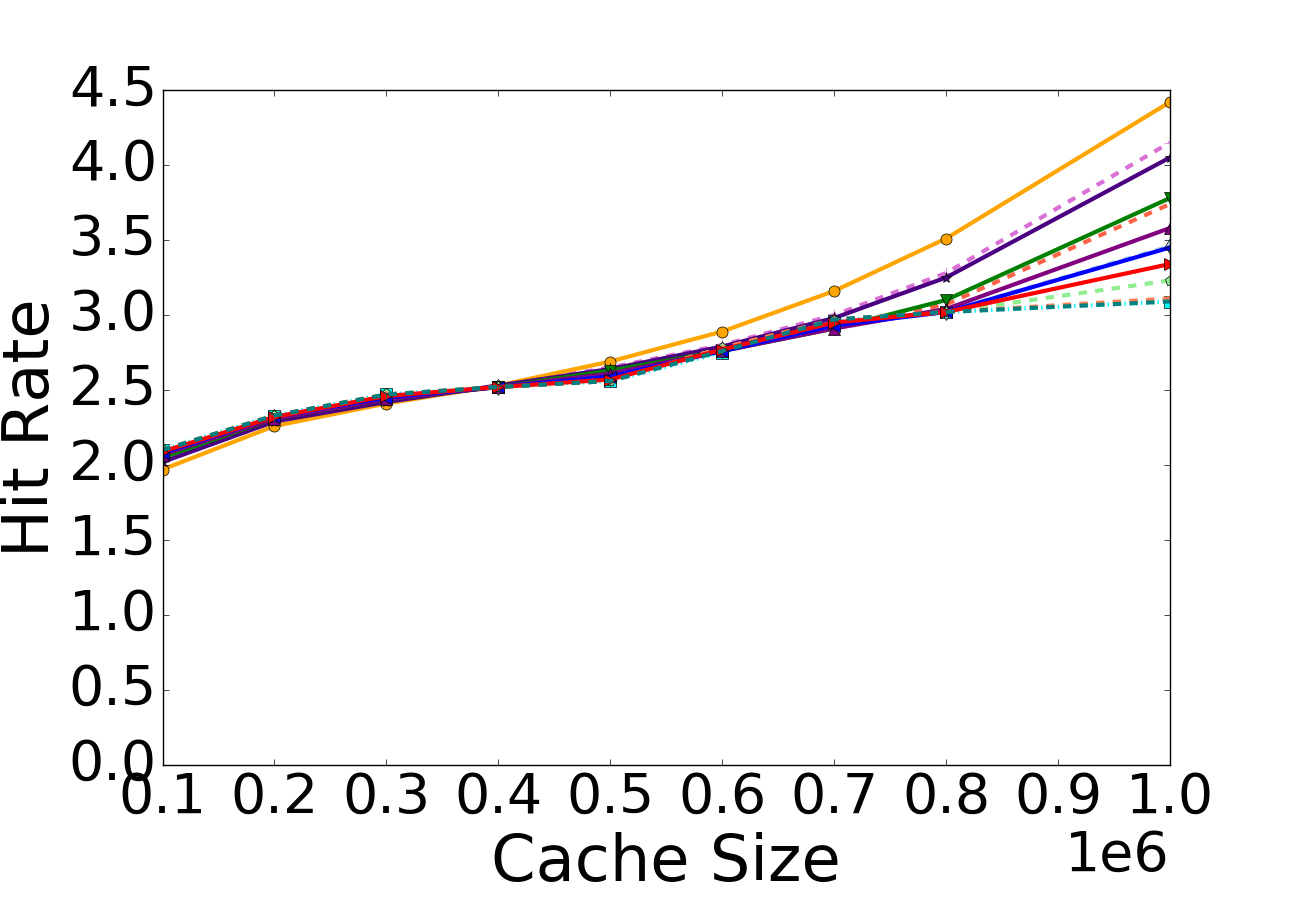}}
		\subfigure[LFU +TinyLFU]{\includegraphics[width=0.45\columnwidth]{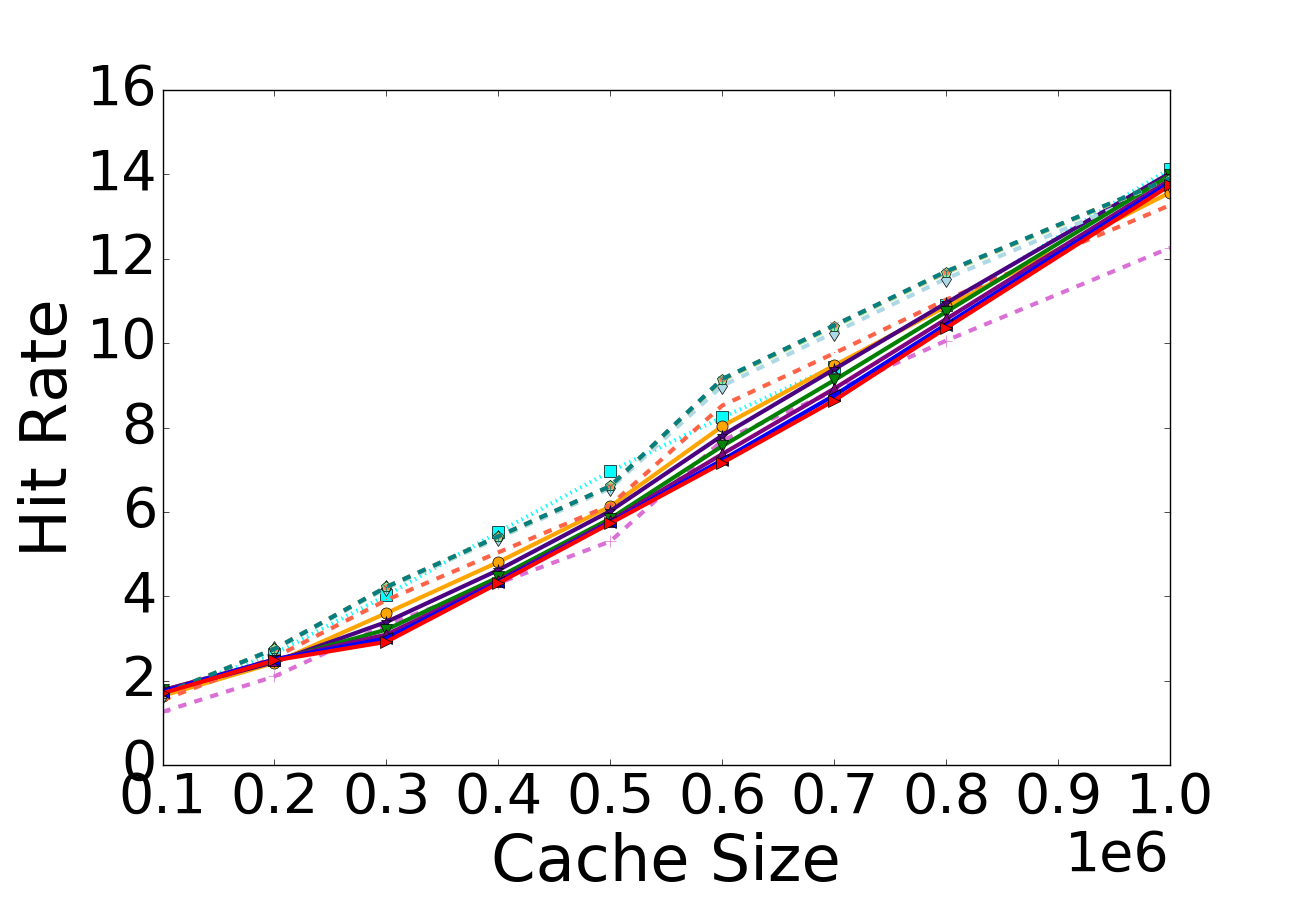}}
		\subfigure[Product]{\includegraphics[width=0.45\columnwidth]{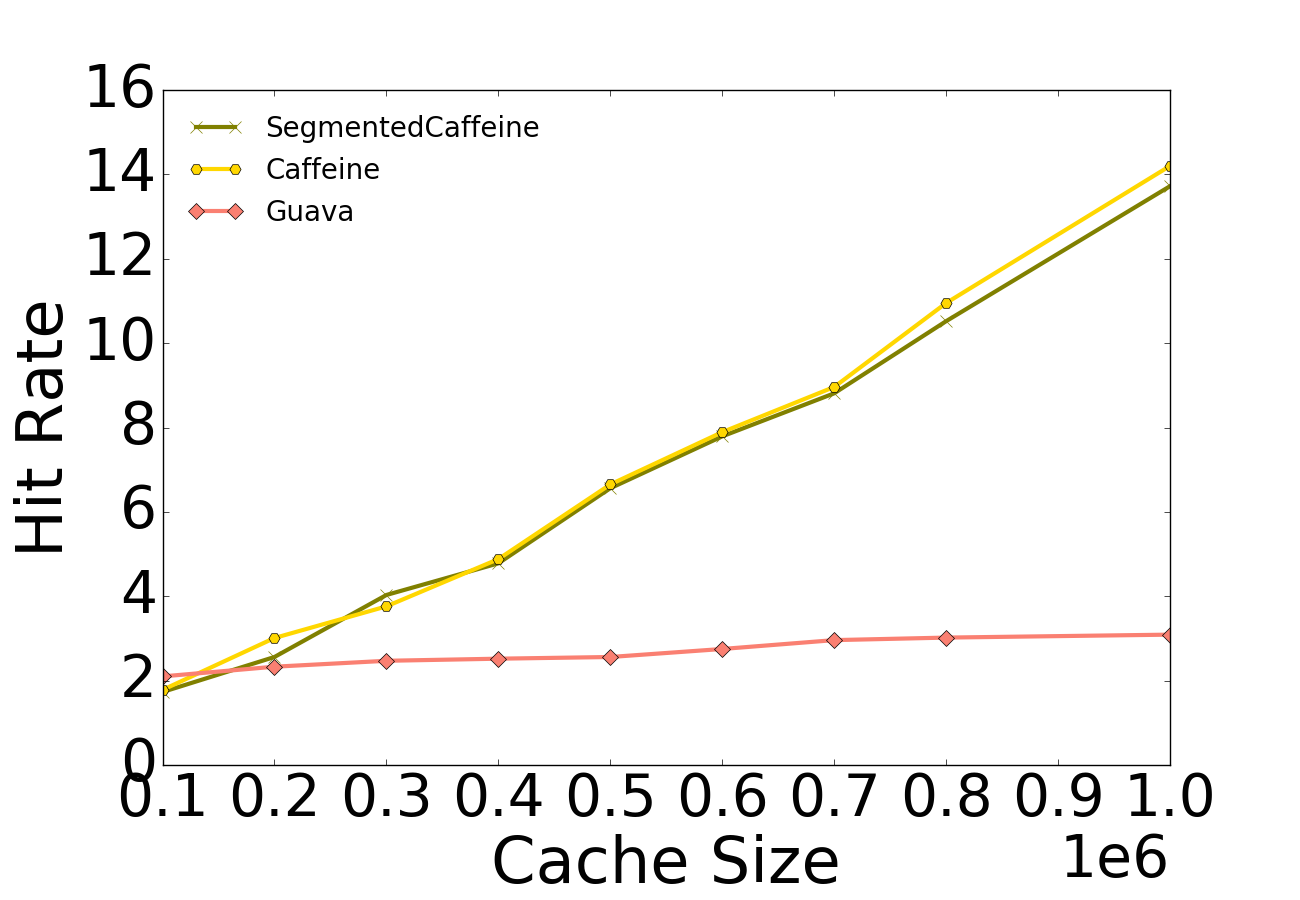}}			
		\subfigure[LRU +TinyLFU]{\includegraphics[width=0.45\columnwidth]{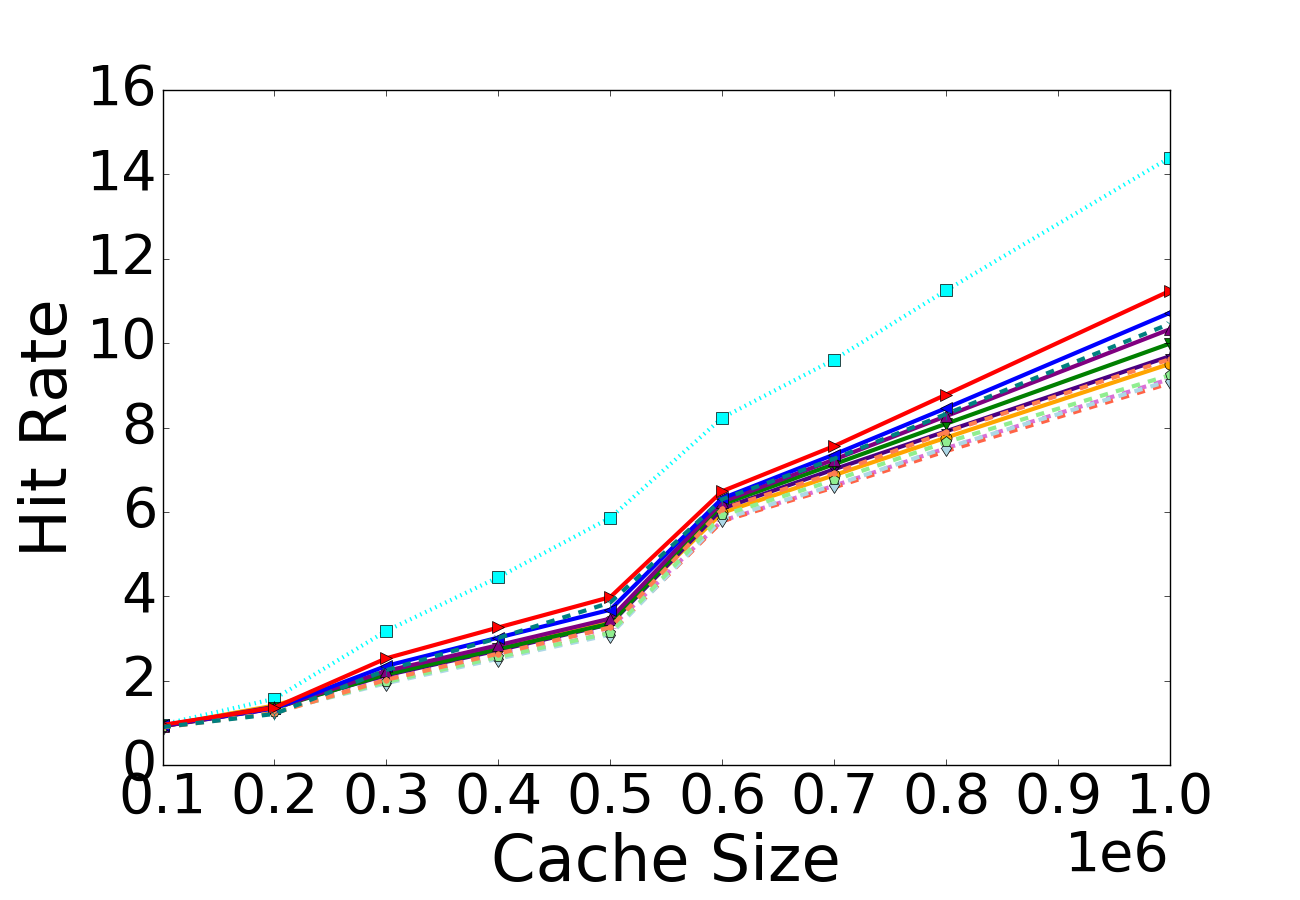}}	
	\end{center}
	\vspace{-0.5cm}
	\caption{DS1.}
	\label{figDS1}
	\vspace{-0.5cm}
\end{figure*}
}
\nottoggle{SMALL}{
\begin{figure*}[t]
	\begin{center}
		\offinterlineskip
		\subfigure[LRU]{\includegraphics[width=0.45\columnwidth]{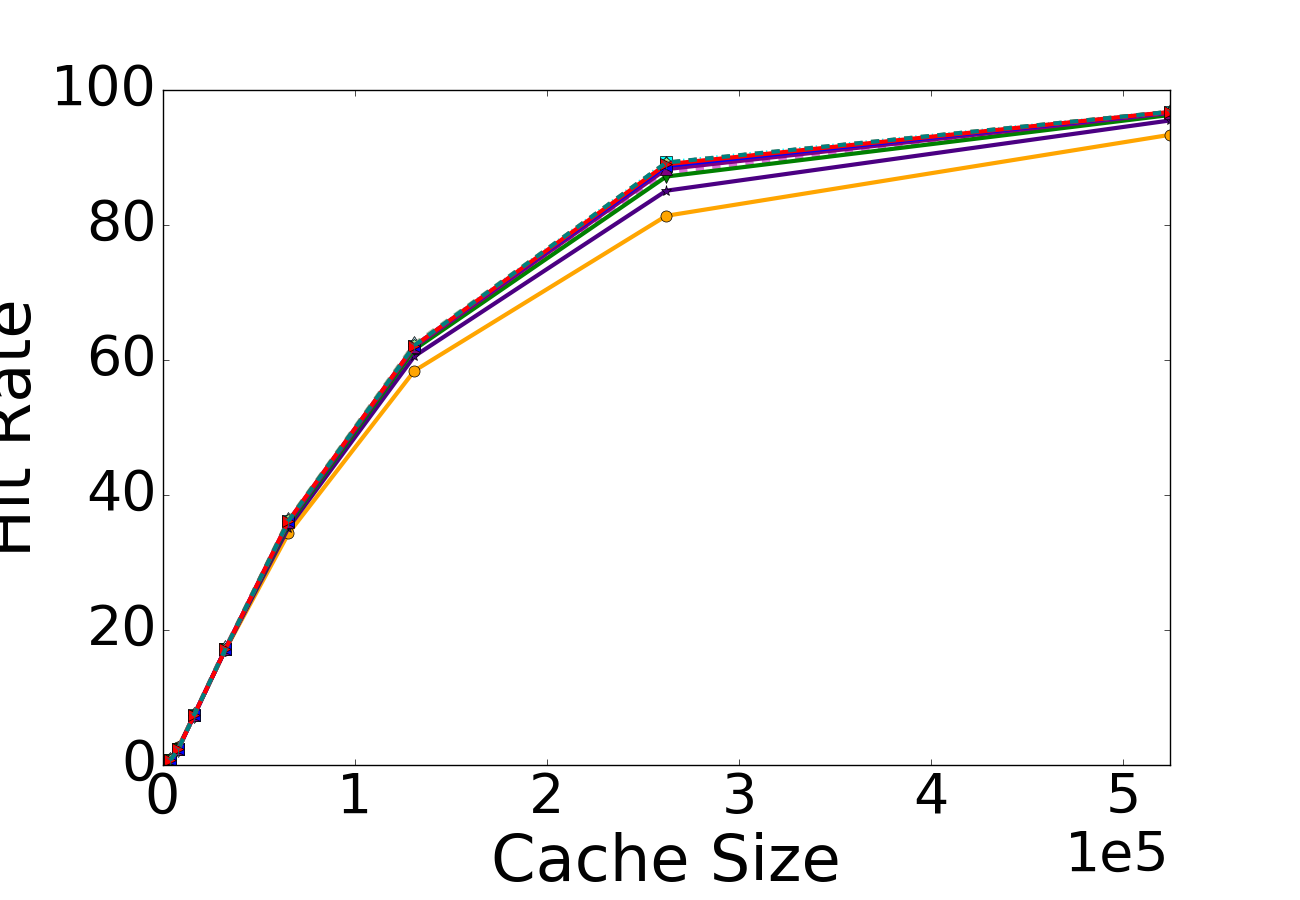}}
		\subfigure[LFU +TinyLFU]{\includegraphics[width=0.45\columnwidth]{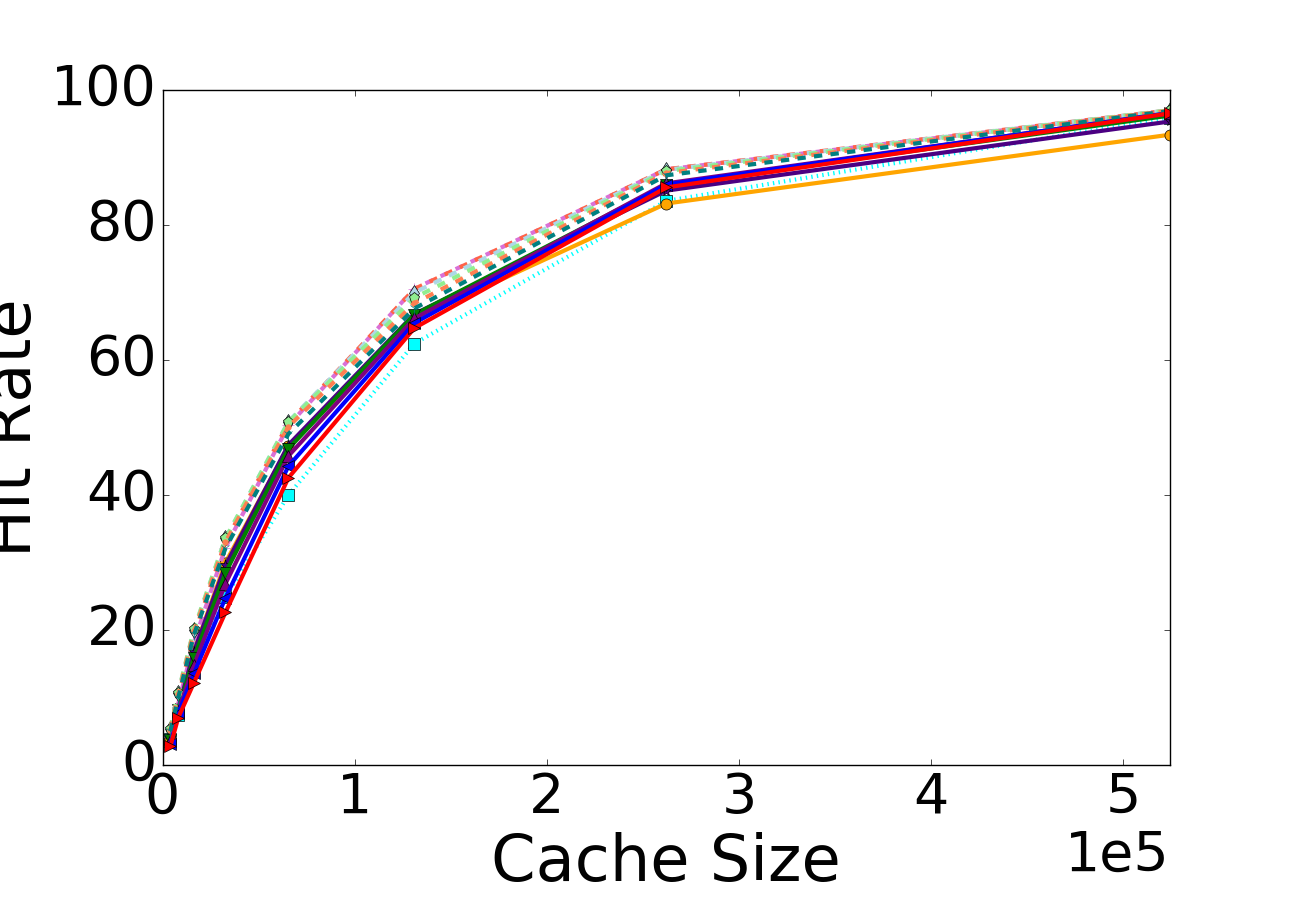}}	
		\subfigure[Product]{\includegraphics[width=0.45\columnwidth]{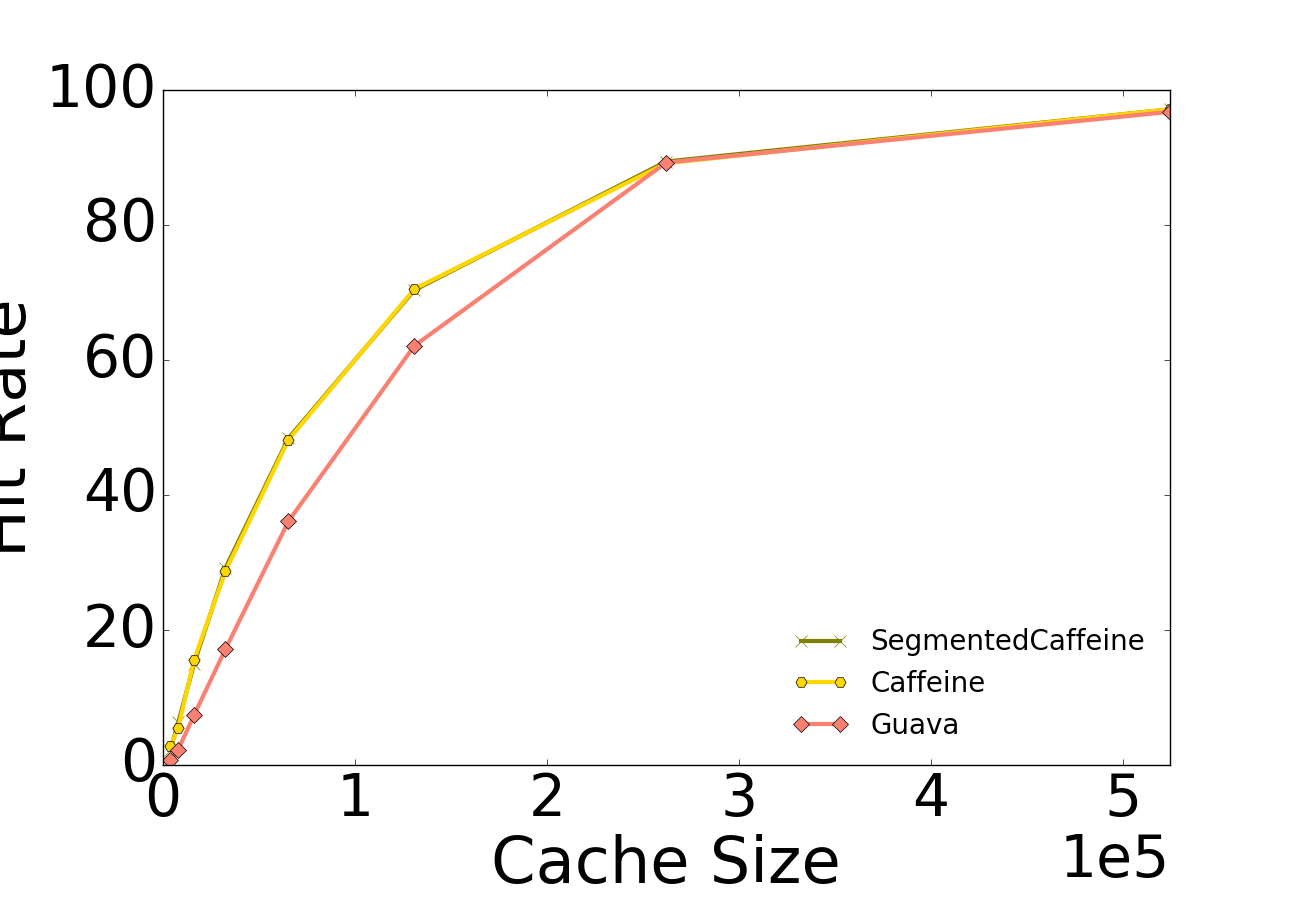}}	
		\subfigure[Hyperbolic]{\includegraphics[width=0.45\columnwidth]{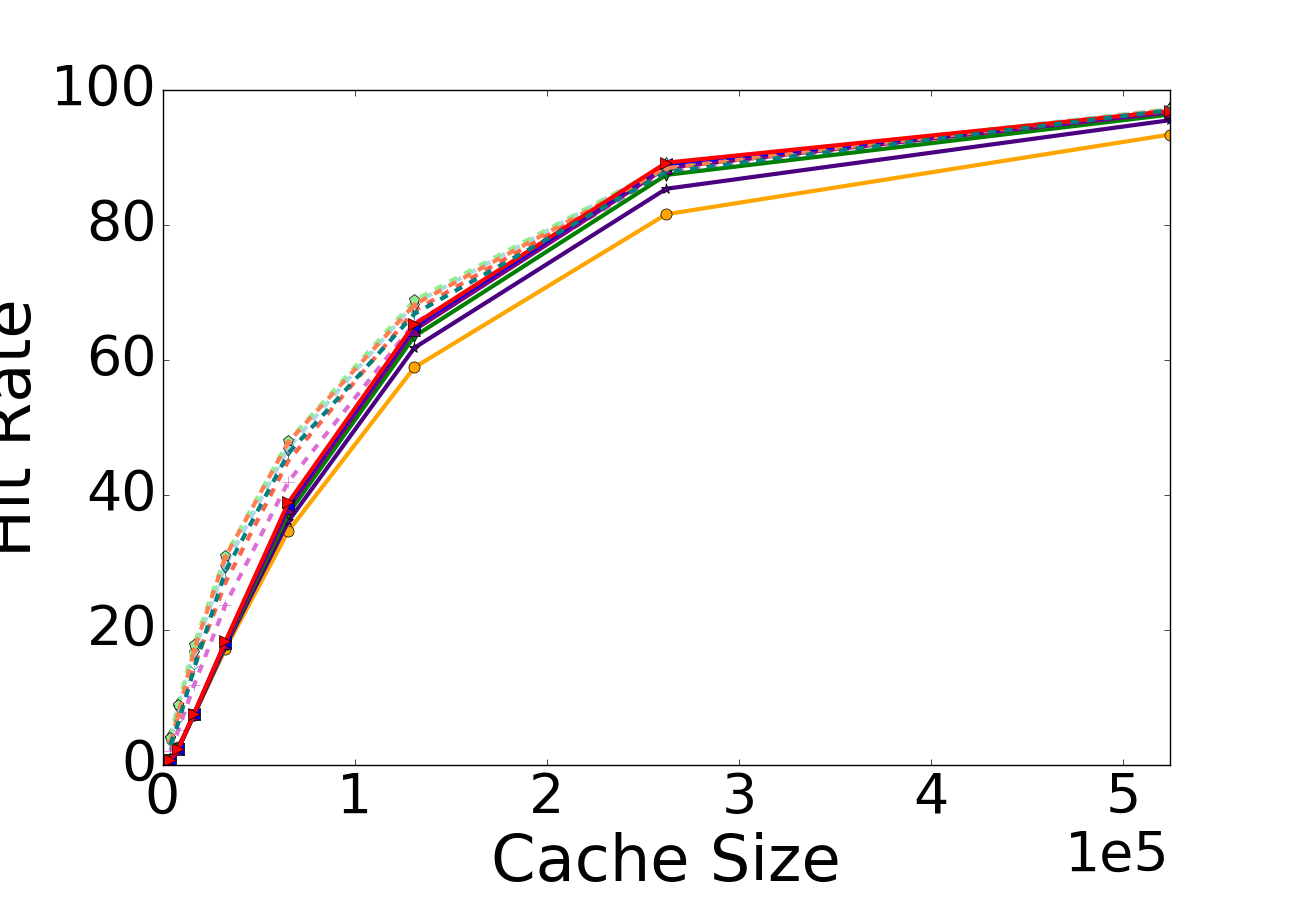}}
	
	\end{center}
	\vspace{-0.5cm}
	\caption{P8.}
	\label{figP8}
	\vspace{-0.5cm}
\end{figure*}
}{}

\begin{figure*}[t]
	\begin{center}
		\offinterlineskip
		\subfigure[LRU]{\includegraphics[width=0.45\columnwidth]{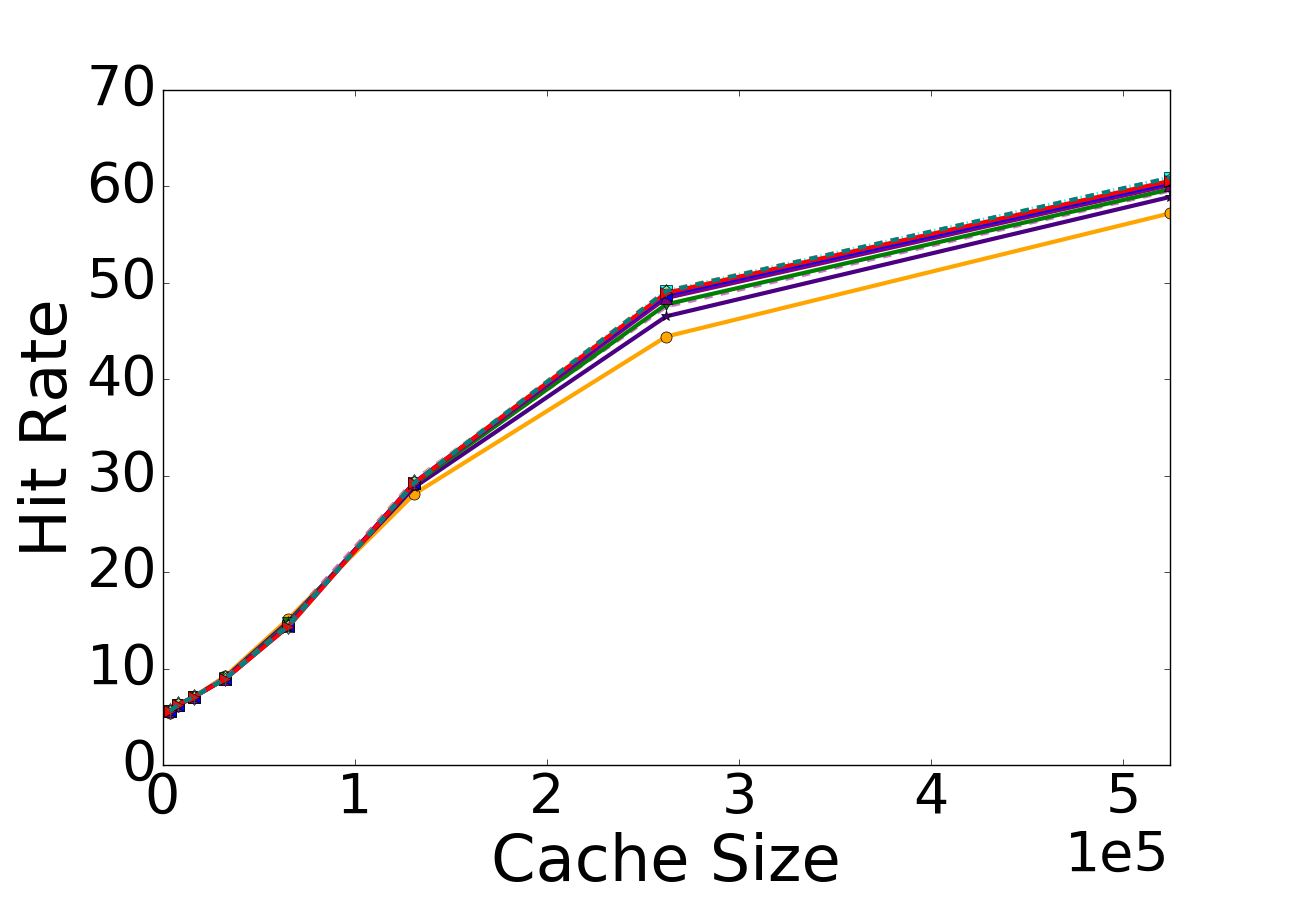}}
		\subfigure[LFU +TinyLFU]{\includegraphics[width=0.45\columnwidth]{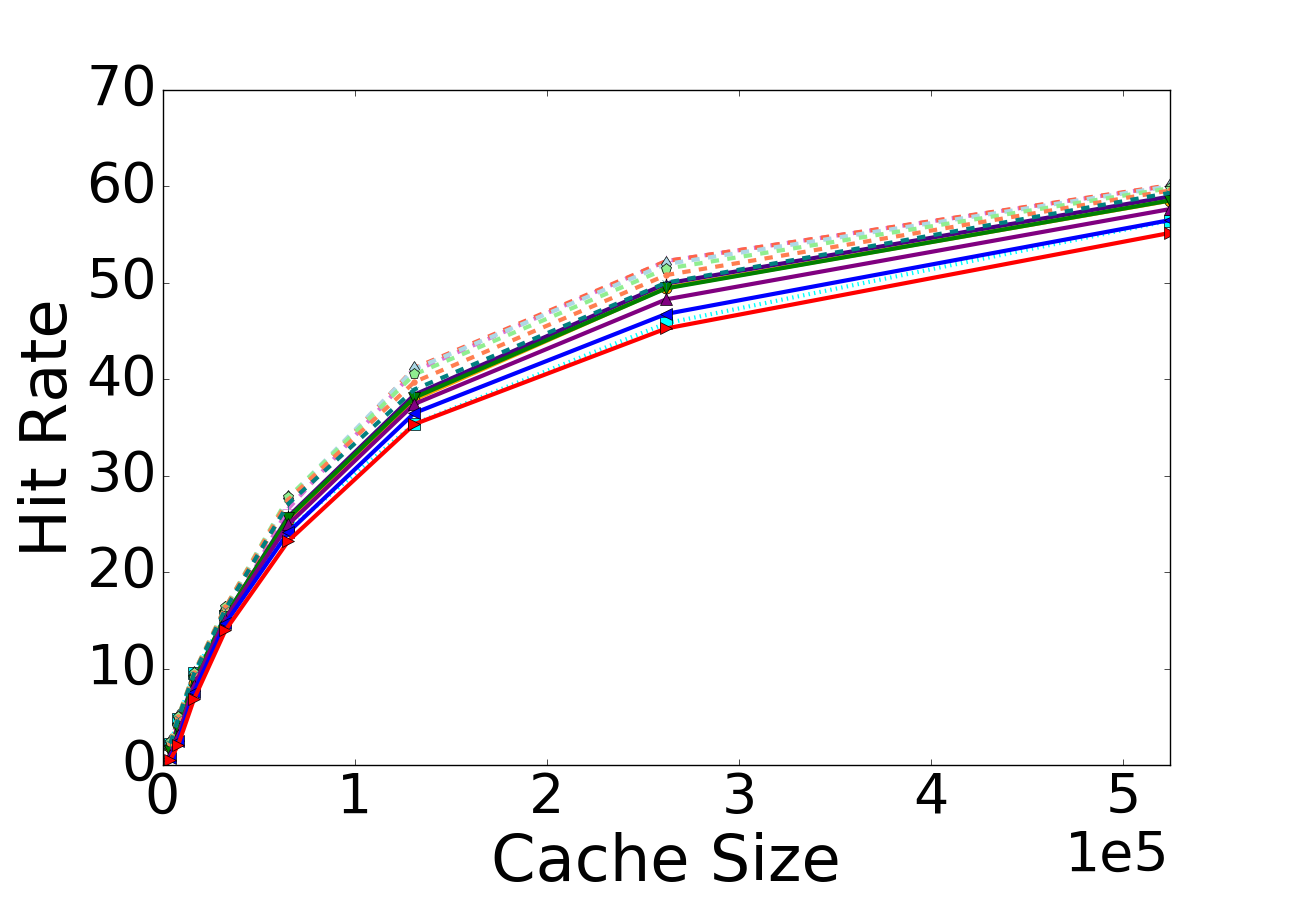}}	
		\subfigure[Product]{\includegraphics[width=0.45\columnwidth]{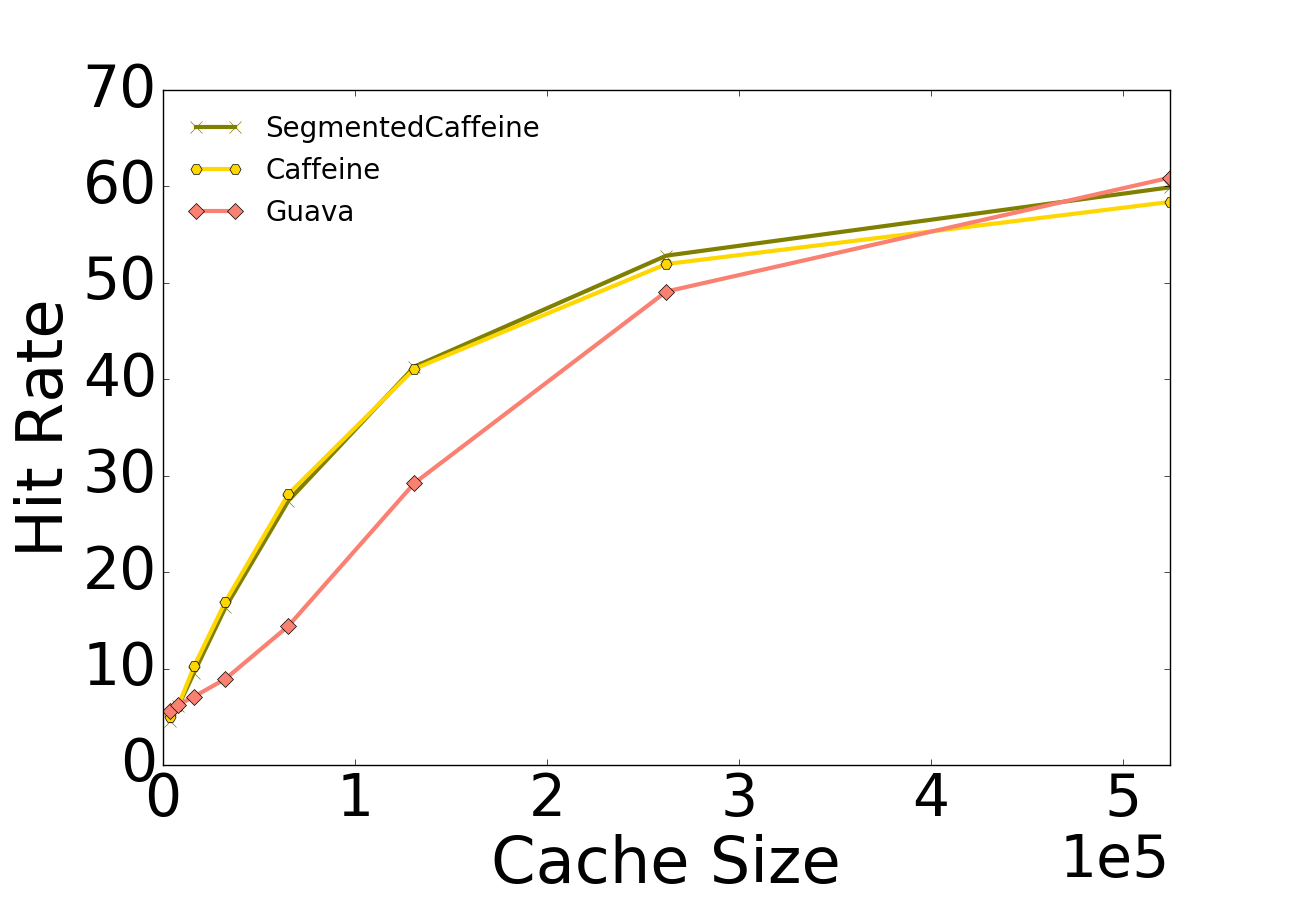}}
		\subfigure[Hyperbolic]{\includegraphics[width=0.45\columnwidth]{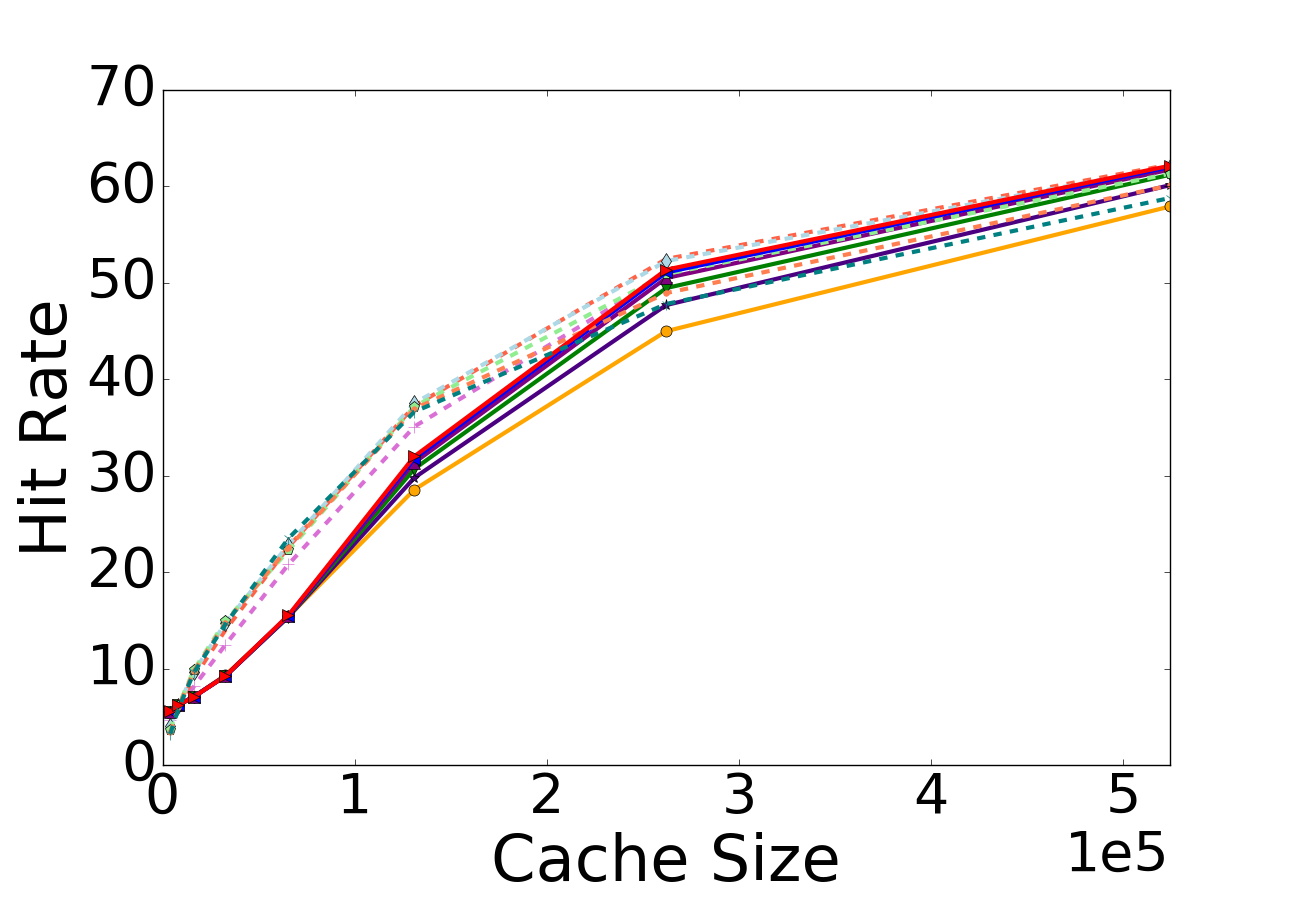}}		
	\end{center}
	\vspace{-0.5cm}
	\caption{P12.}
	\label{figP12}
	\vspace{-0.5cm}
\end{figure*}

\nottoggle{SMALL}{
\begin{figure*}[t]
	\begin{center}
		\offinterlineskip
		\subfigure[LRU]{\includegraphics[width=0.45\columnwidth]{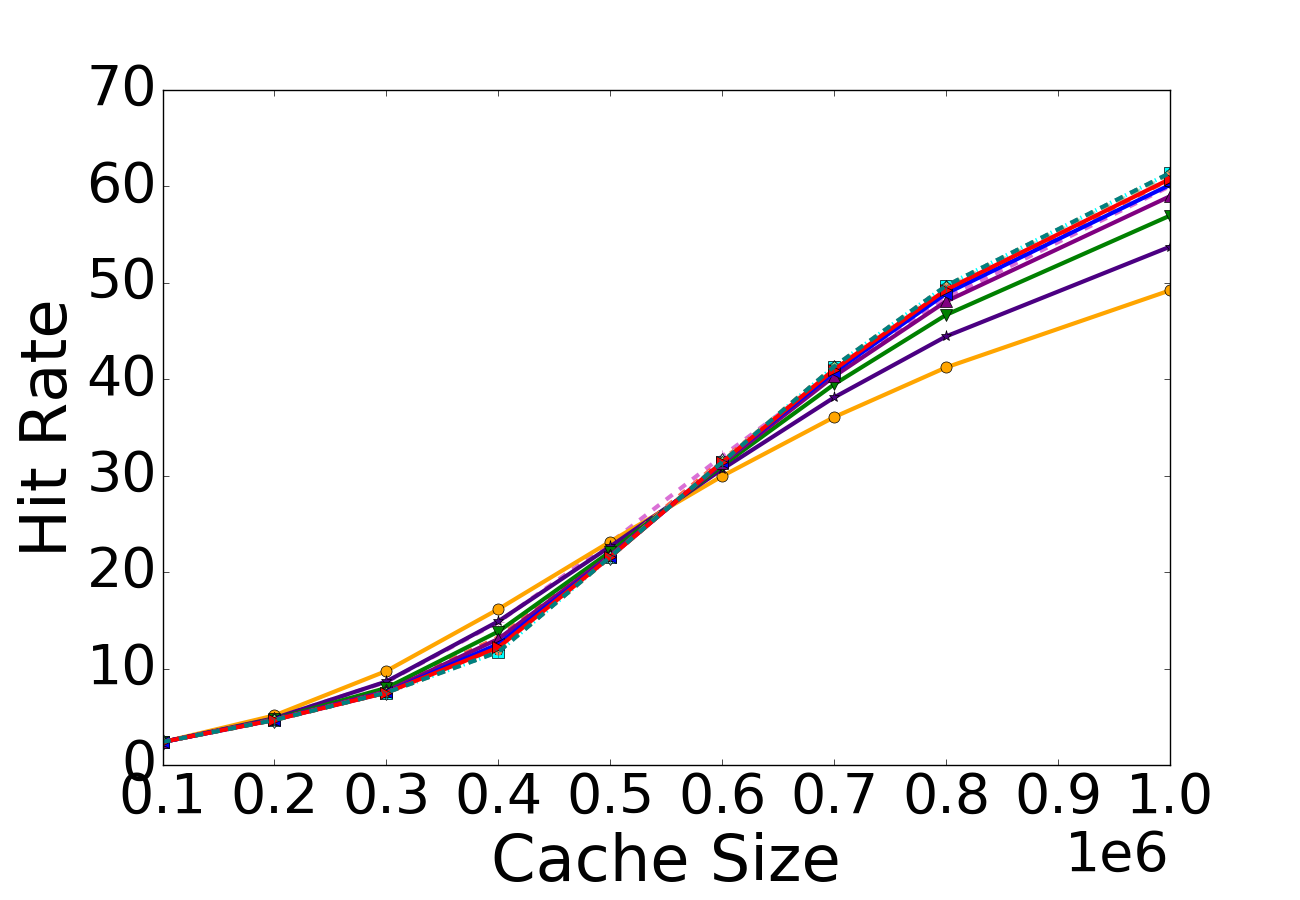}}
		\subfigure[LFU +TinyLFU]{\includegraphics[width=0.45\columnwidth]{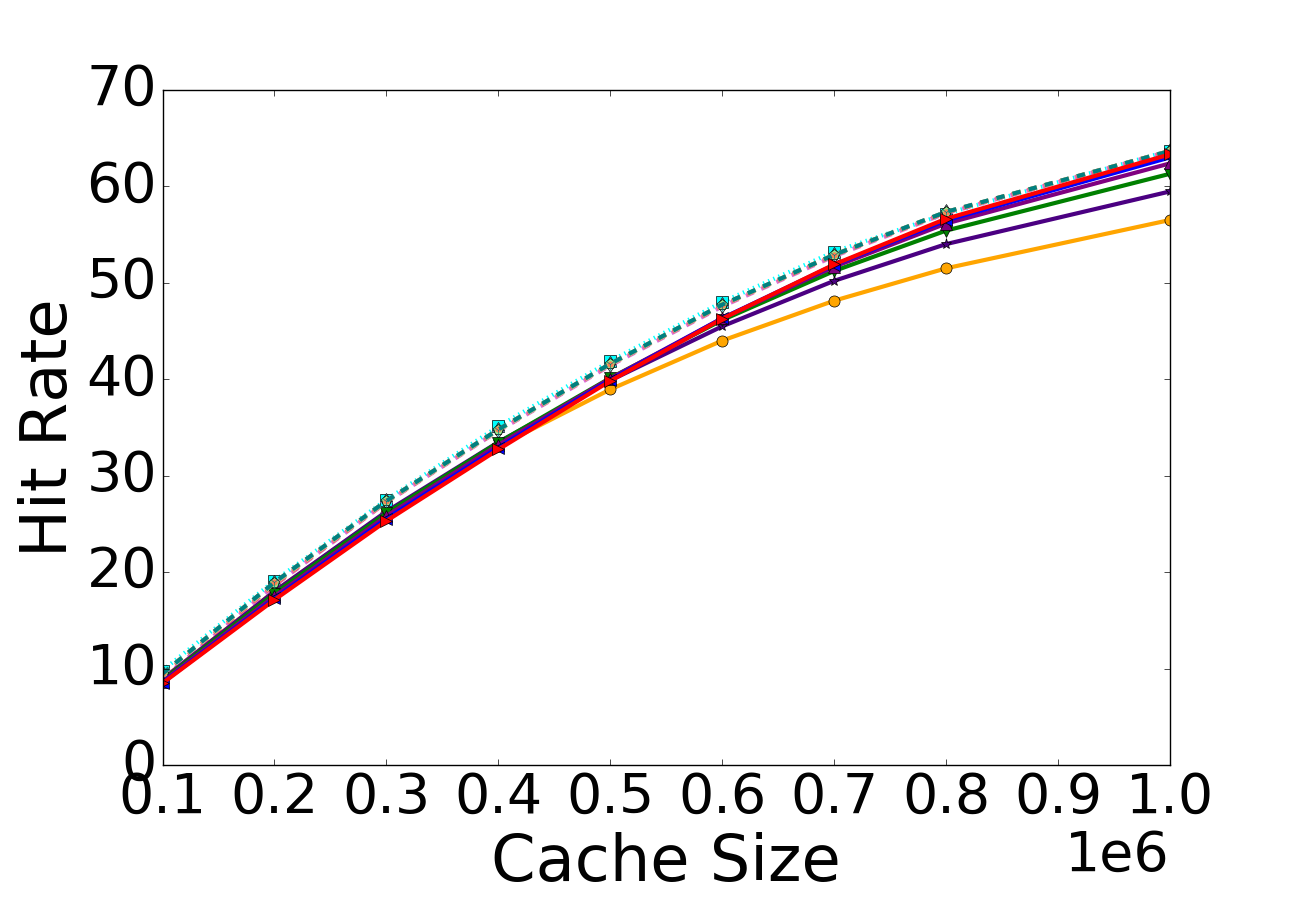}}	
		\subfigure[Product]{\includegraphics[width=0.45\columnwidth]{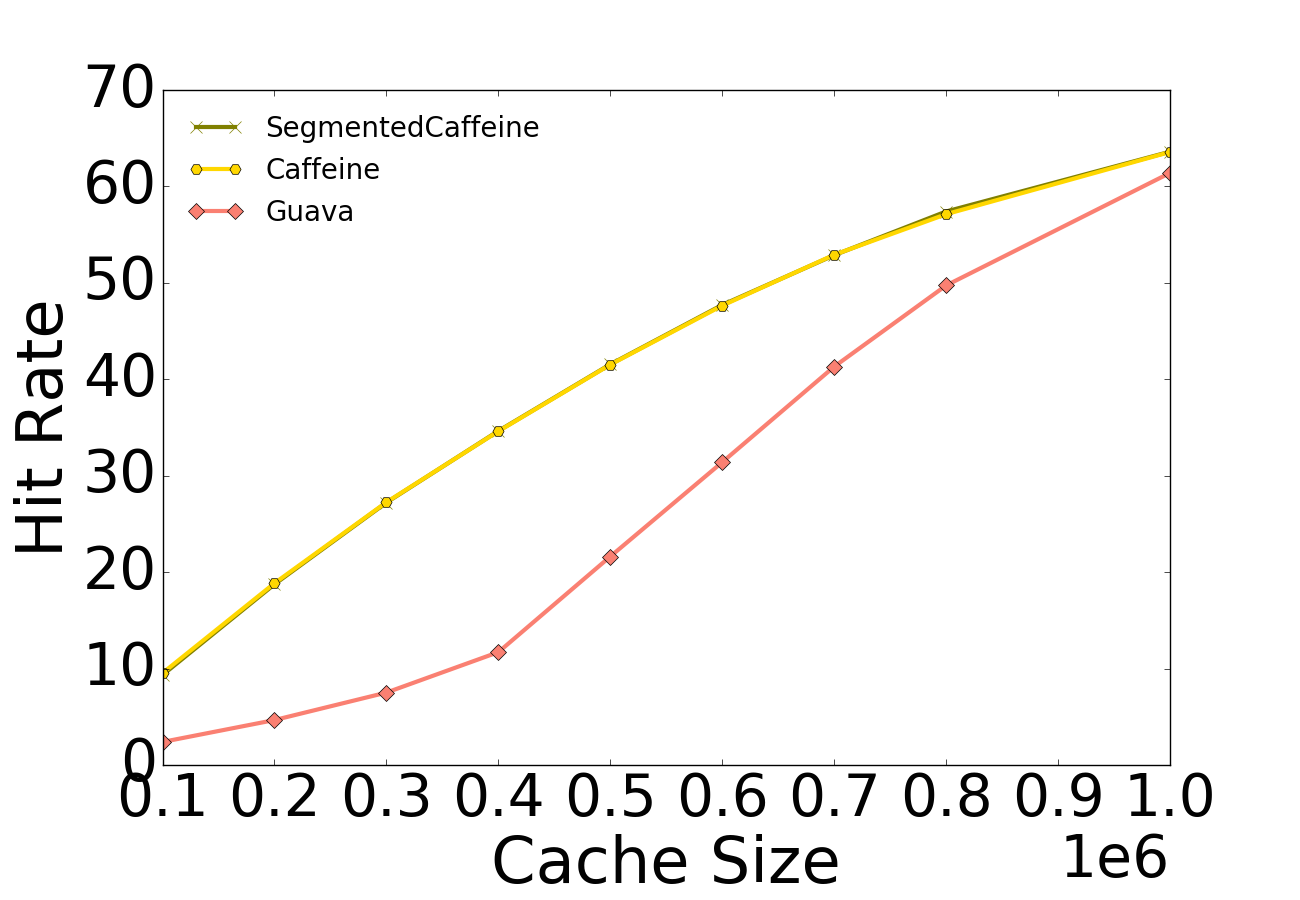}}	
		\subfigure[Hyperbolic+TinyLfu]{\includegraphics[width=0.45\columnwidth]{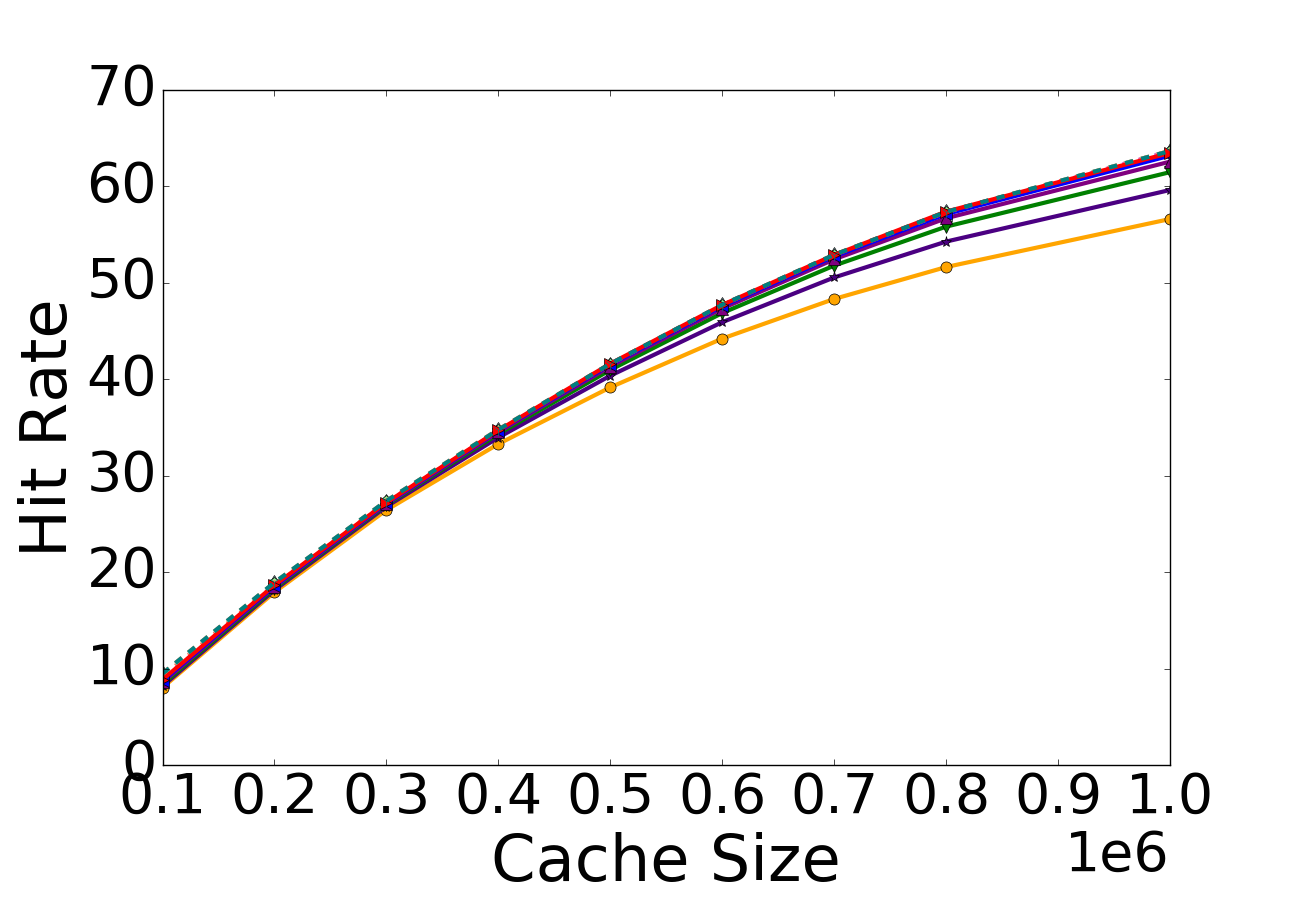}}			
	\end{center}

	\vspace{-0.5cm}
	\caption{S1.}
	\label{figS1}
	\vspace{-0.5cm}
\end{figure*}
}{}

\begin{figure*}[t]
	\begin{center}
	\offinterlineskip
		\subfigure[LRU]{\includegraphics[width=0.45\columnwidth]{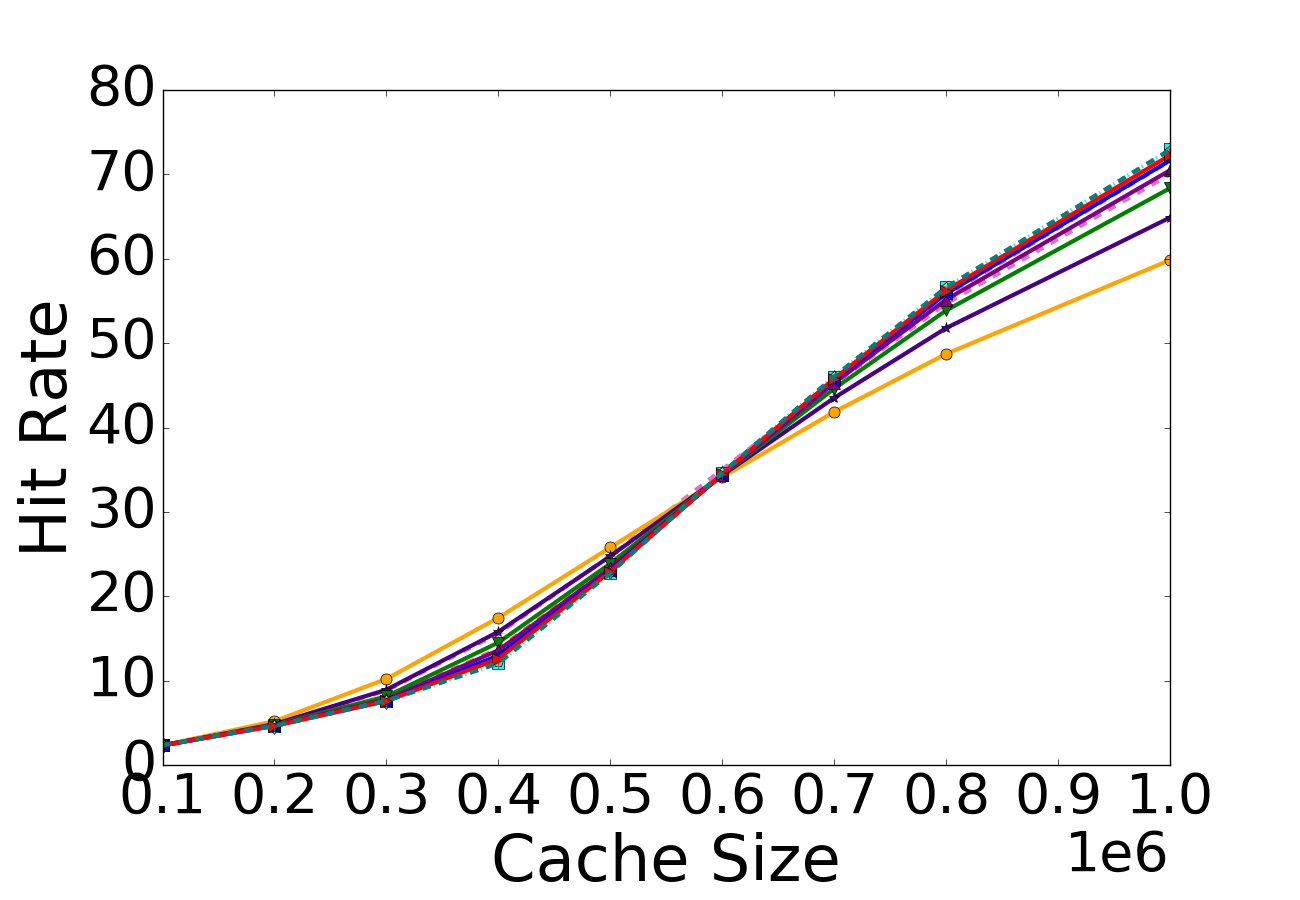}}
\subfigure[LFU +TinyLFU]{\includegraphics[width=0.45\columnwidth]{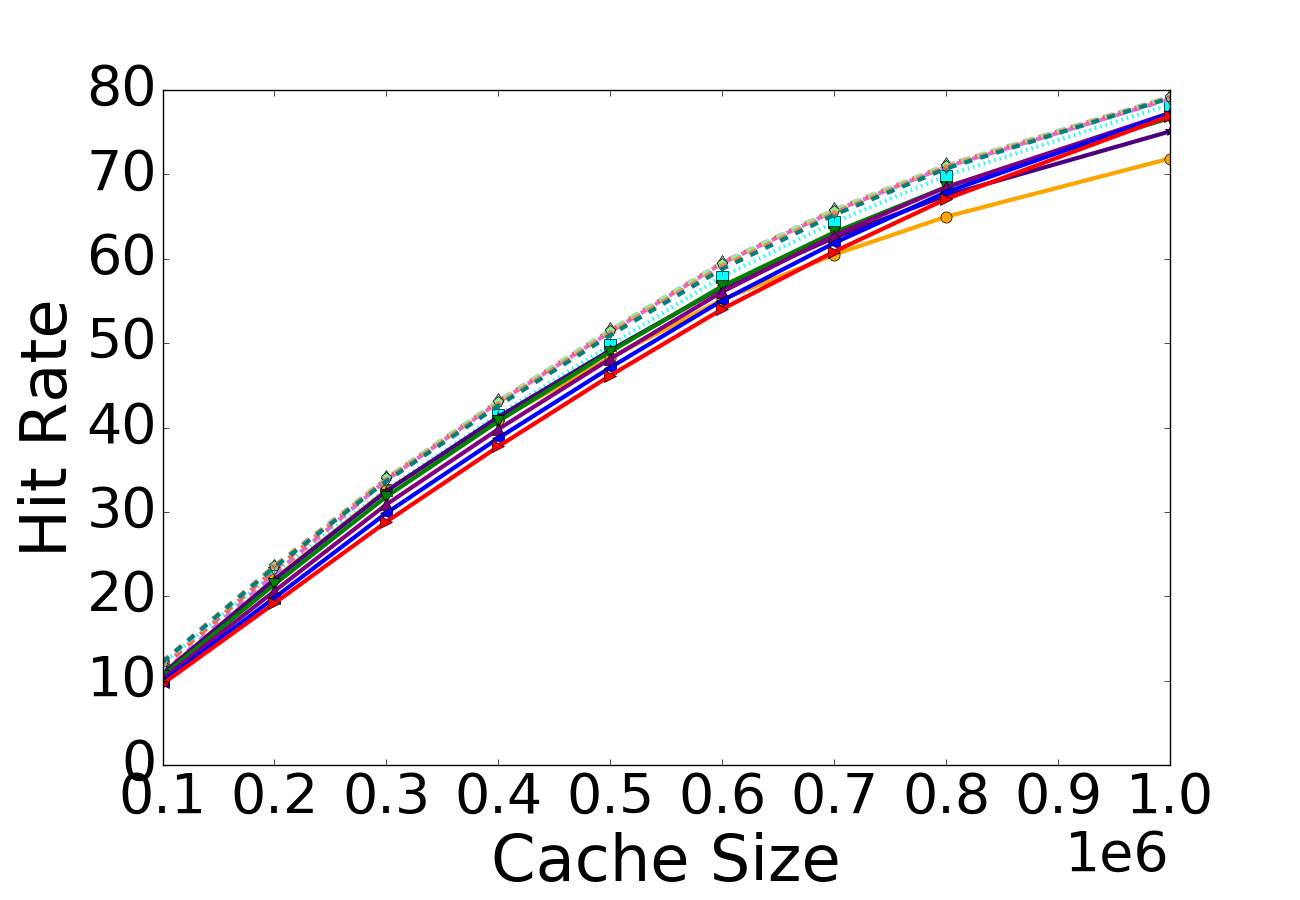}}	
	\subfigure[Product]{\includegraphics[width=0.45\columnwidth]{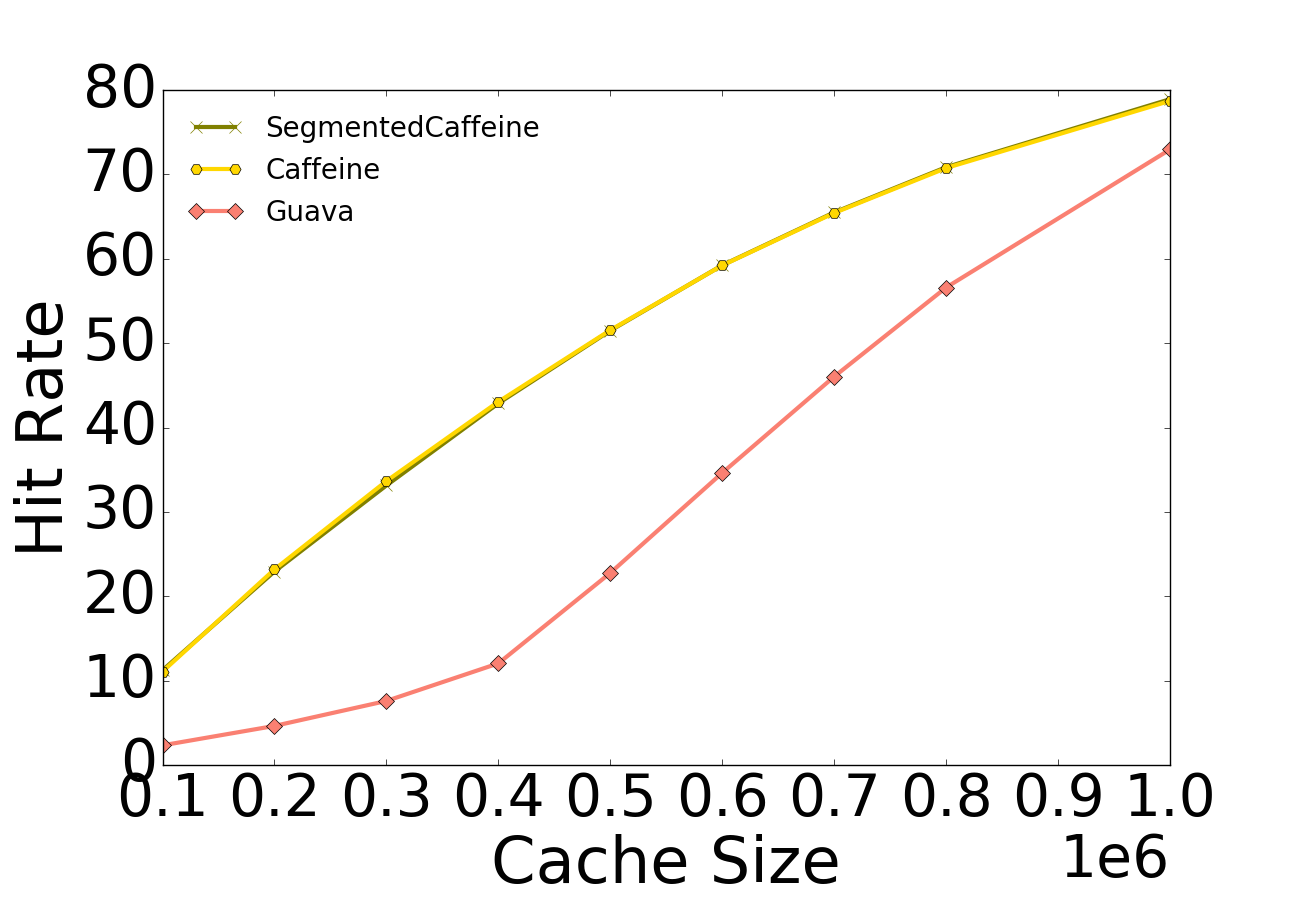}}	
	\subfigure[Hyperbolic+TinyLfu]{\includegraphics[width=0.45\columnwidth]{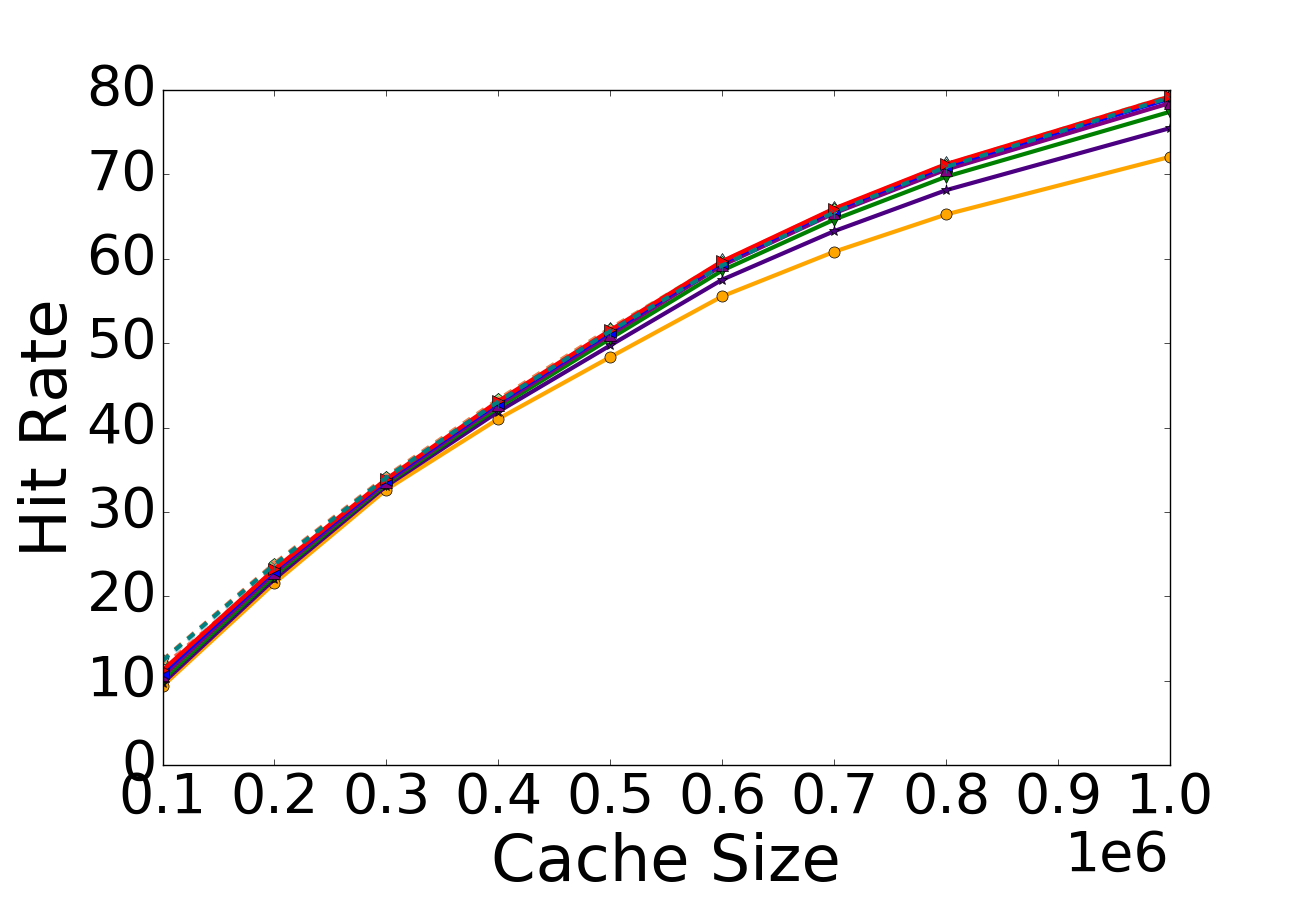}}			
	\end{center}
	\vspace{-0.5cm}
	\caption{S3.}
	\label{figS3}
	\vspace{-0.5cm}
\end{figure*}

\begin{figure*}[t]
	\begin{center}
		\offinterlineskip
		\subfigure[LRU]{\includegraphics[width=0.45\columnwidth]{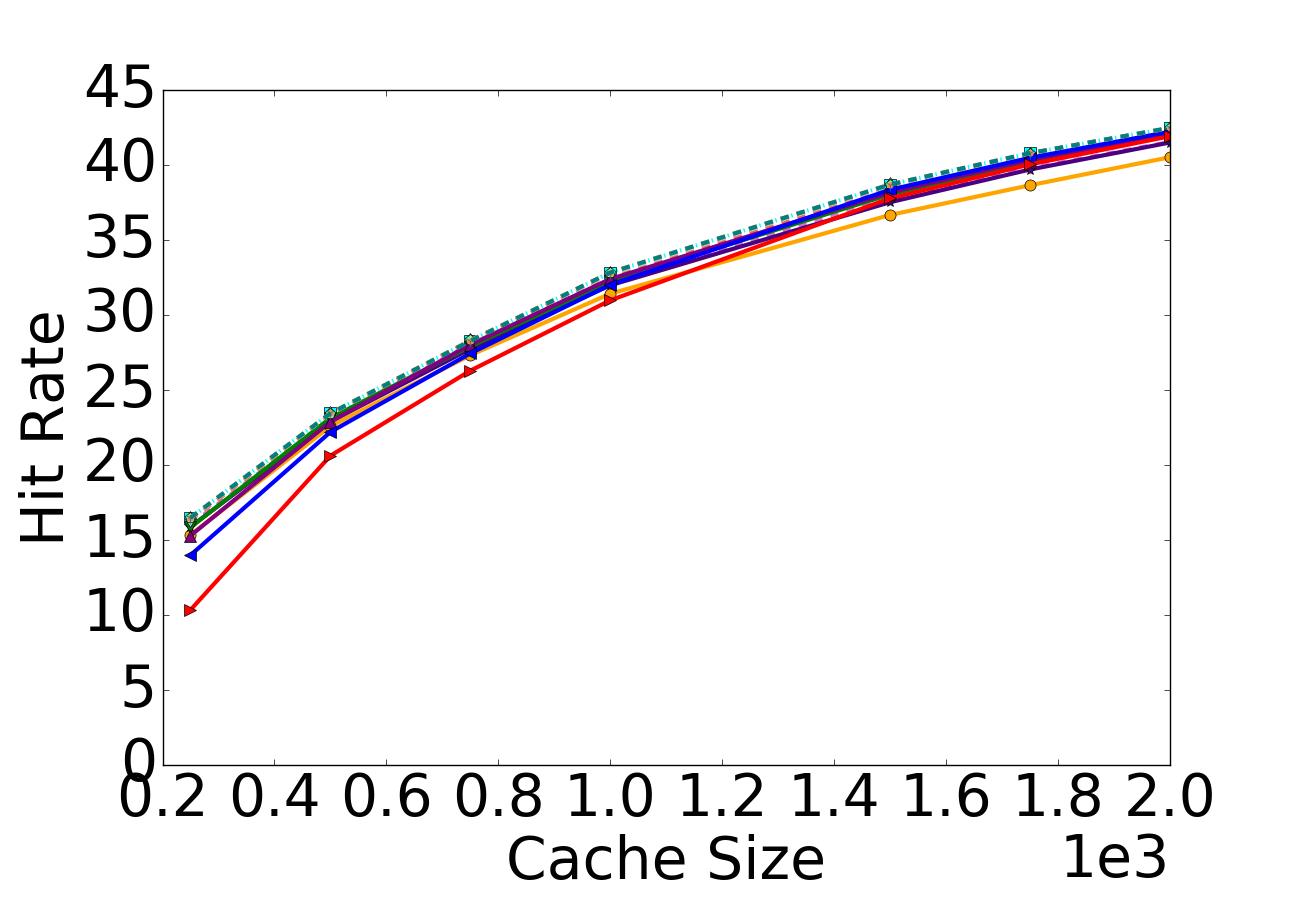}}
		\subfigure[LFU +TinyLFU]{\includegraphics[width=0.45\columnwidth]{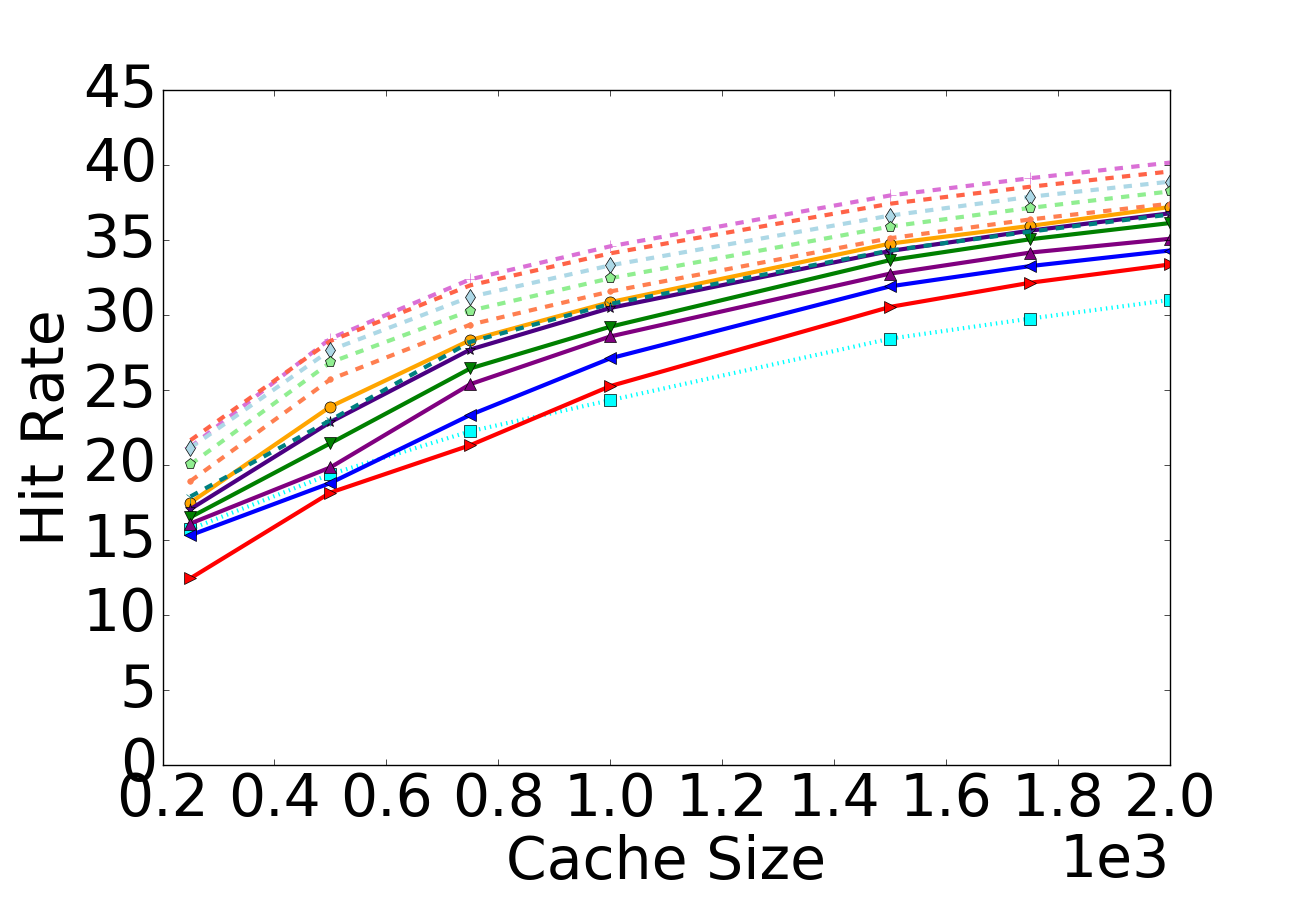}}	
		\subfigure[Product]{\includegraphics[width=0.45\columnwidth]{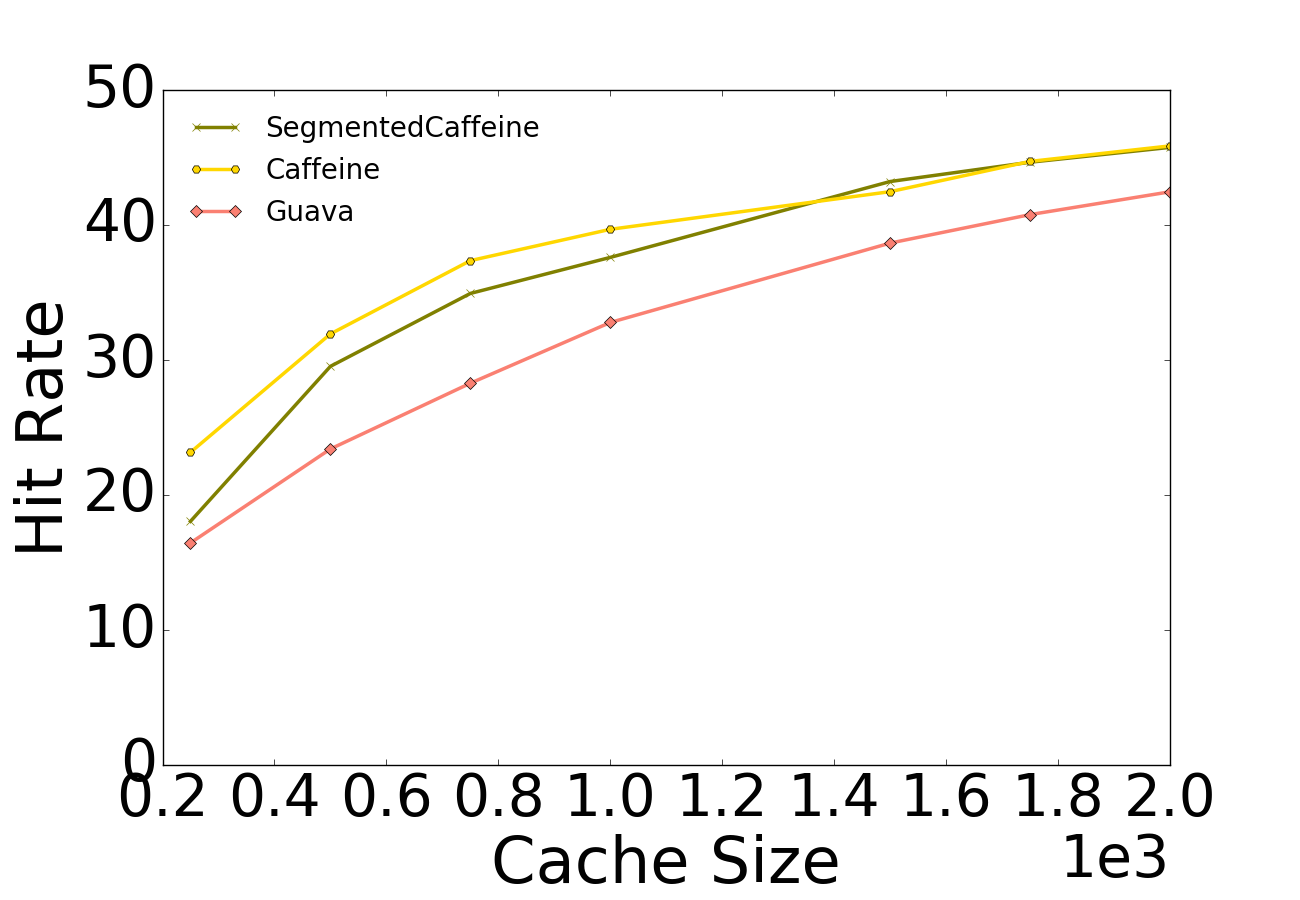}}	
		\subfigure[Hyperbolic+TinyLfu]{\includegraphics[width=0.45\columnwidth]{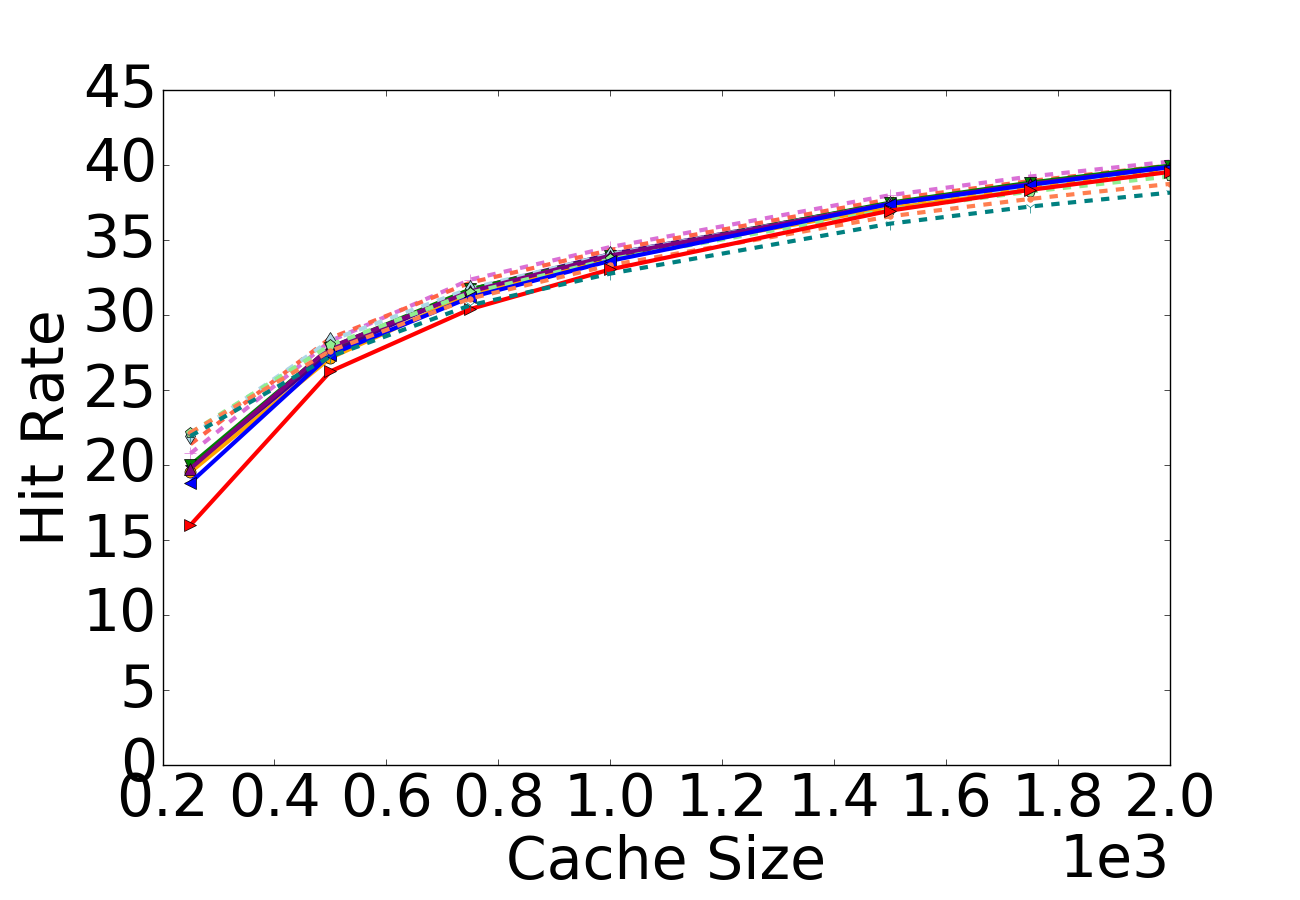}}			
	\end{center}
	\vspace{-0.5cm}
	\caption{OLTP.}
	\label{figOLTP}
	\vspace{-0.5cm}
\end{figure*}


\nottoggle{MEDIUM}{
\begin{figure*}[t]
	\begin{center}
		\offinterlineskip
		\subfigure[LRU]{\includegraphics[width=0.45\columnwidth]{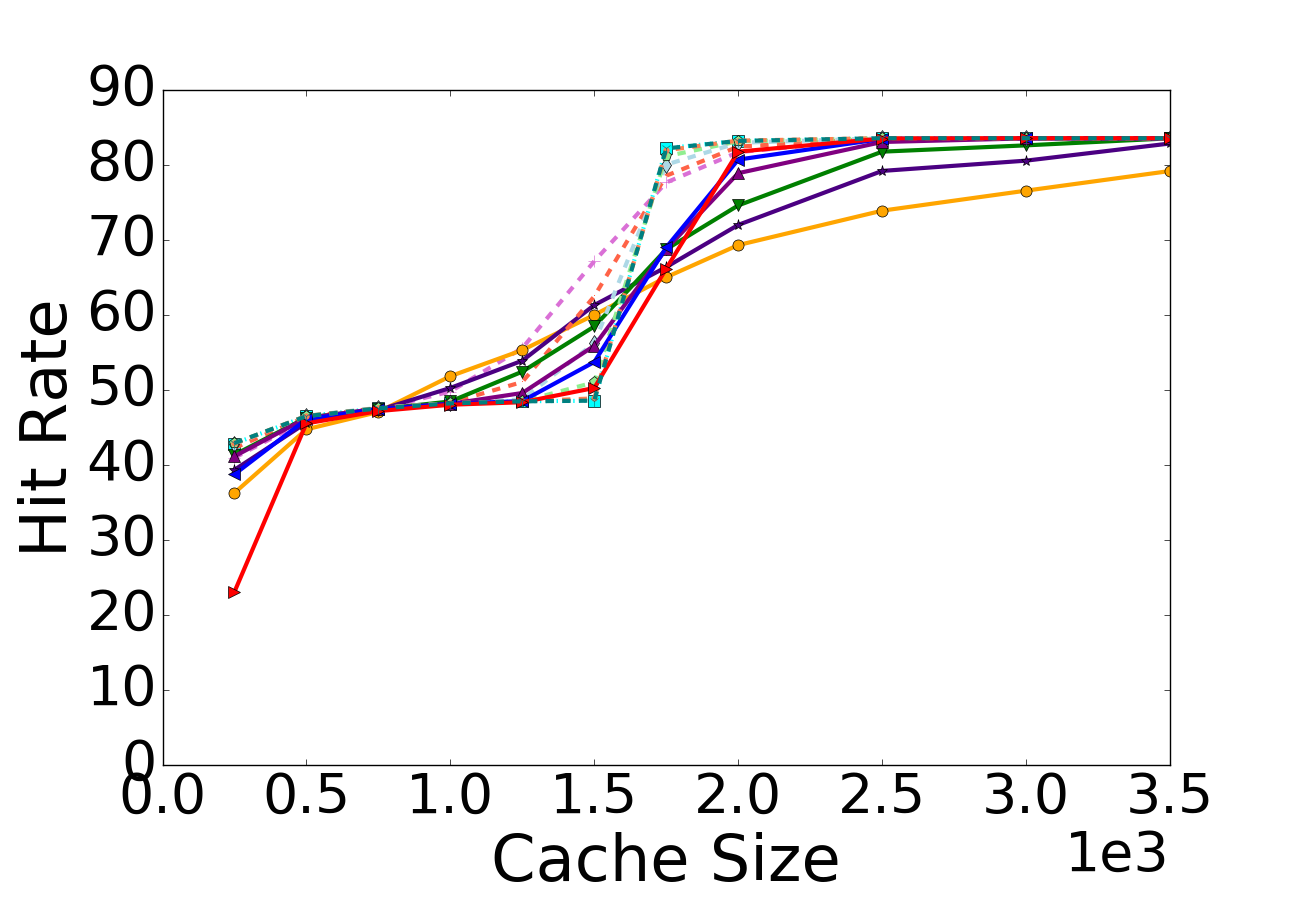}}
		\subfigure[LFU +TinyLFU]{\includegraphics[width=0.45\columnwidth]{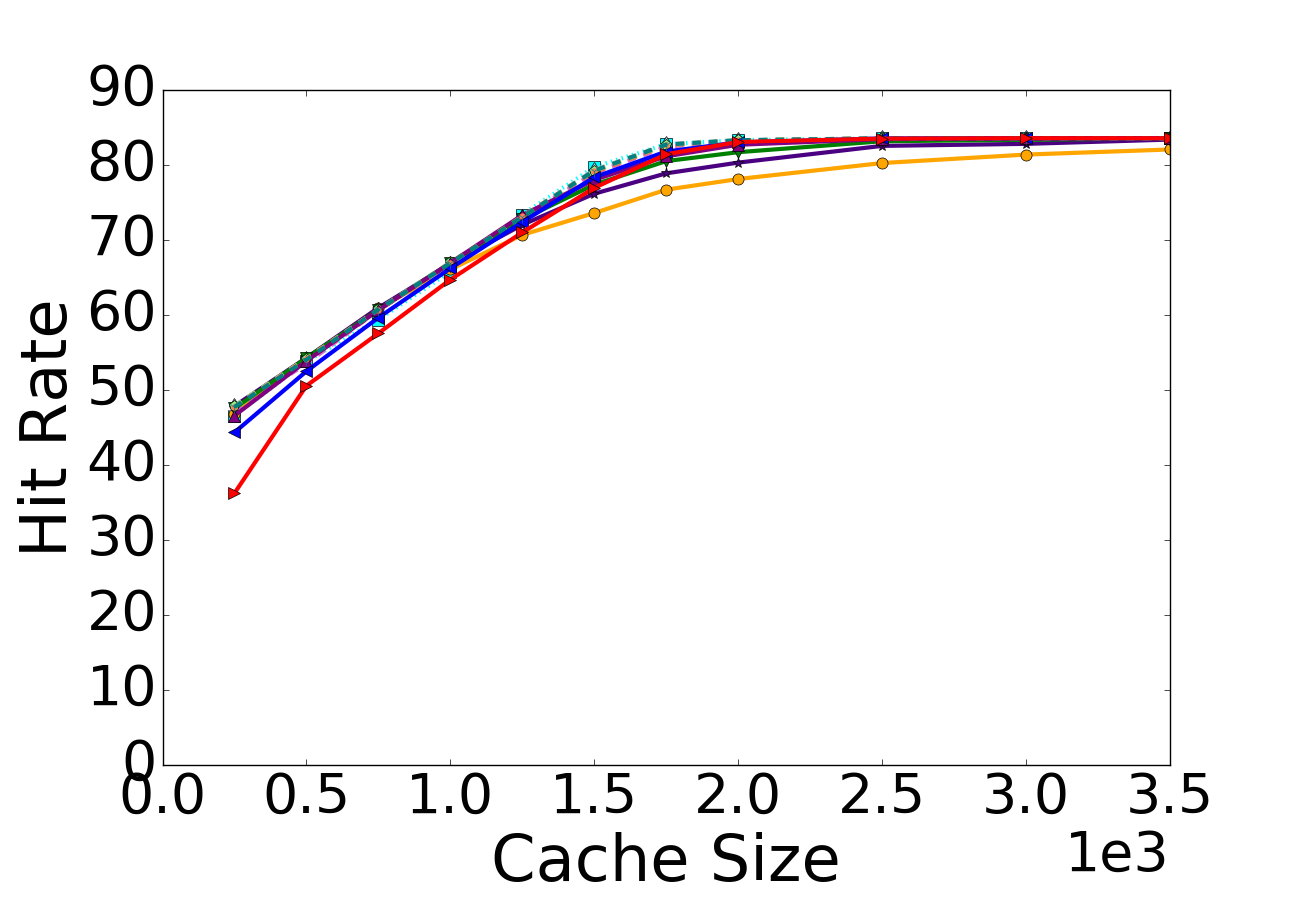}}
		\subfigure[Product]{\includegraphics[width=0.45\columnwidth]{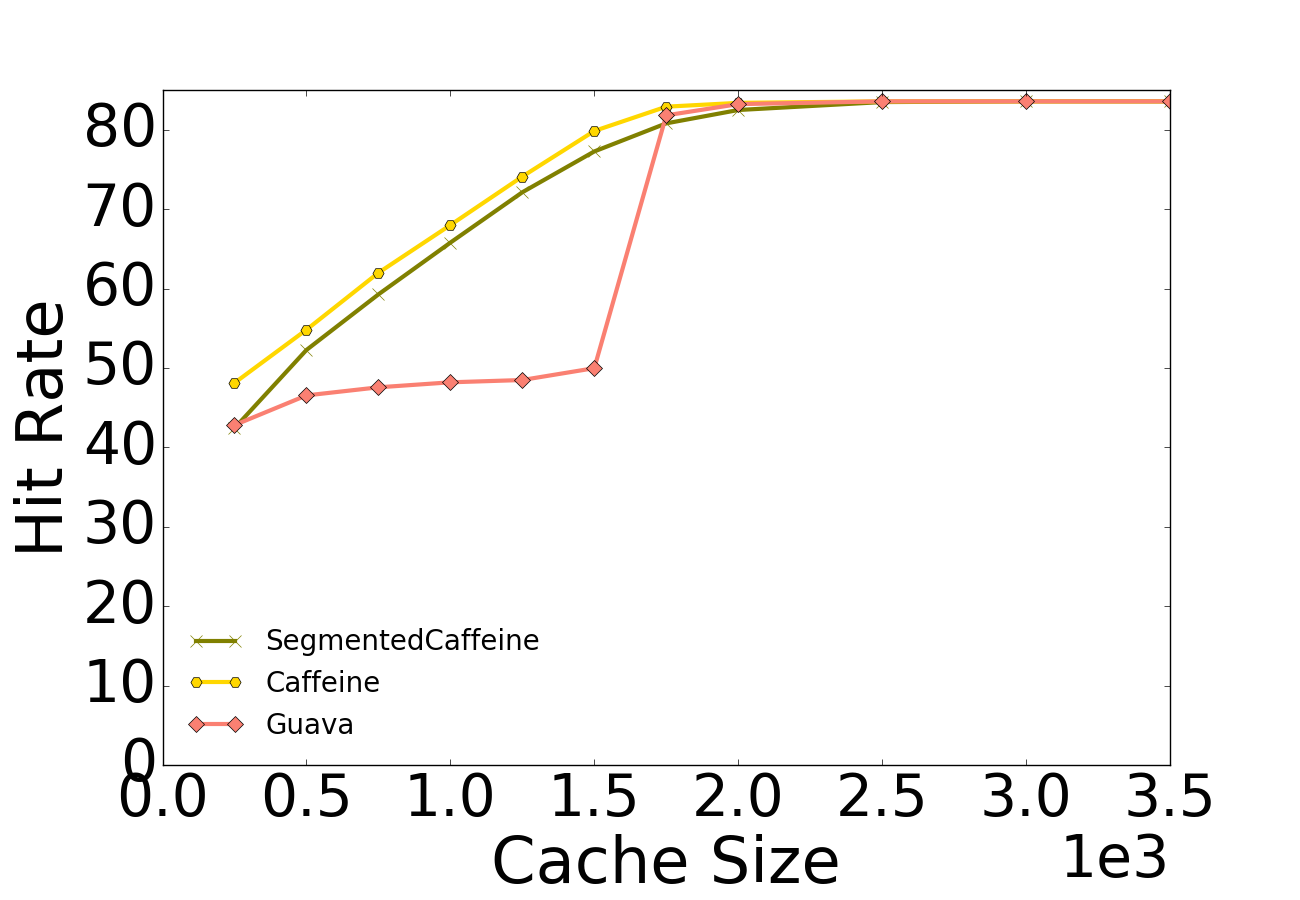}}	
		\subfigure[FIFO+TinyLFU]{\includegraphics[width=0.45\columnwidth]{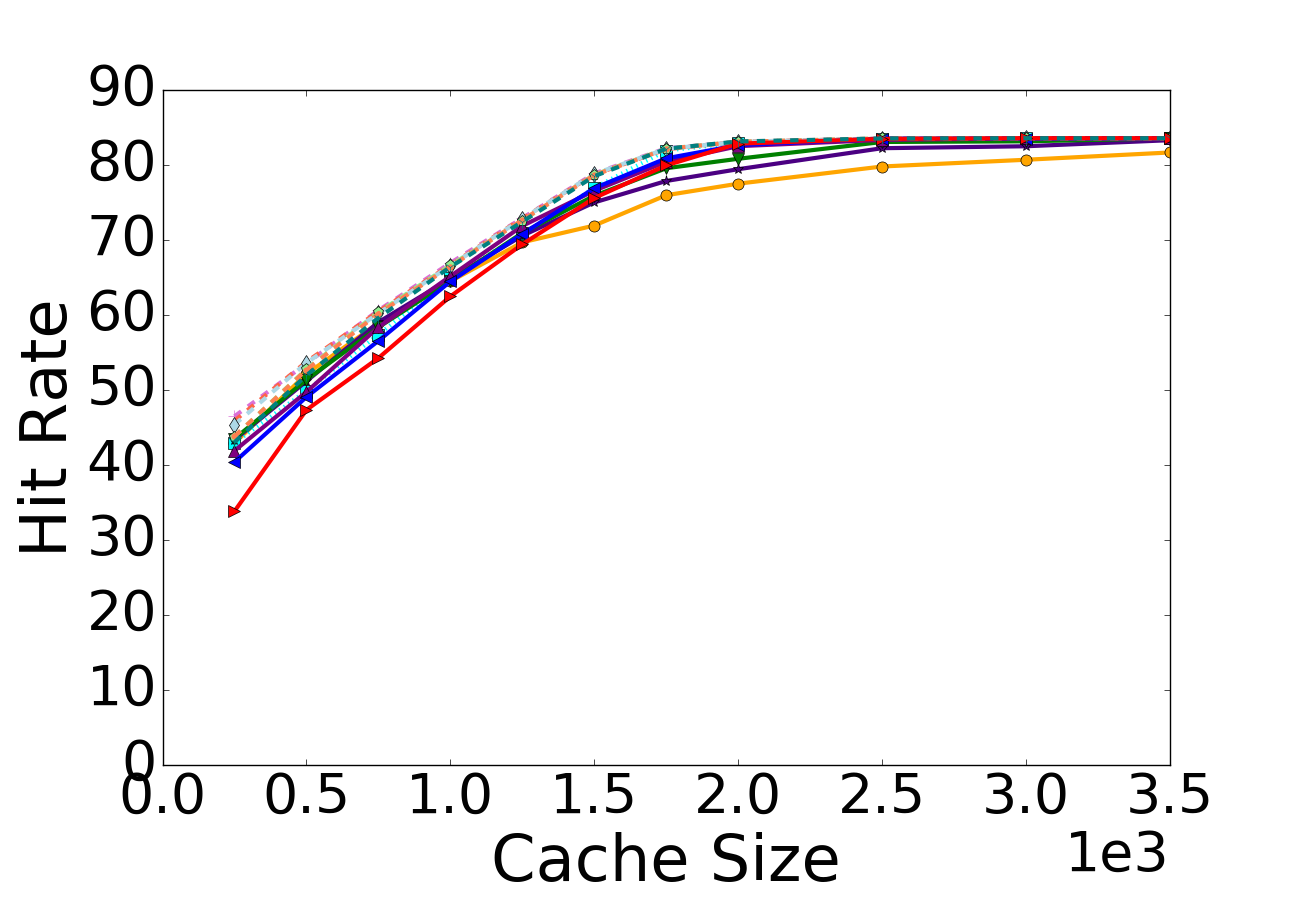}}	
	\end{center}
	\vspace{-0.5cm}
	\caption{multi1.}
	\label{figmulti1}
	\vspace{-0.5cm}
\end{figure*}
}
\nottoggle{SMALL}{
\begin{figure*}[t]
	\begin{center}
		\offinterlineskip
		\subfigure[LRU]{\includegraphics[width=0.45\columnwidth]{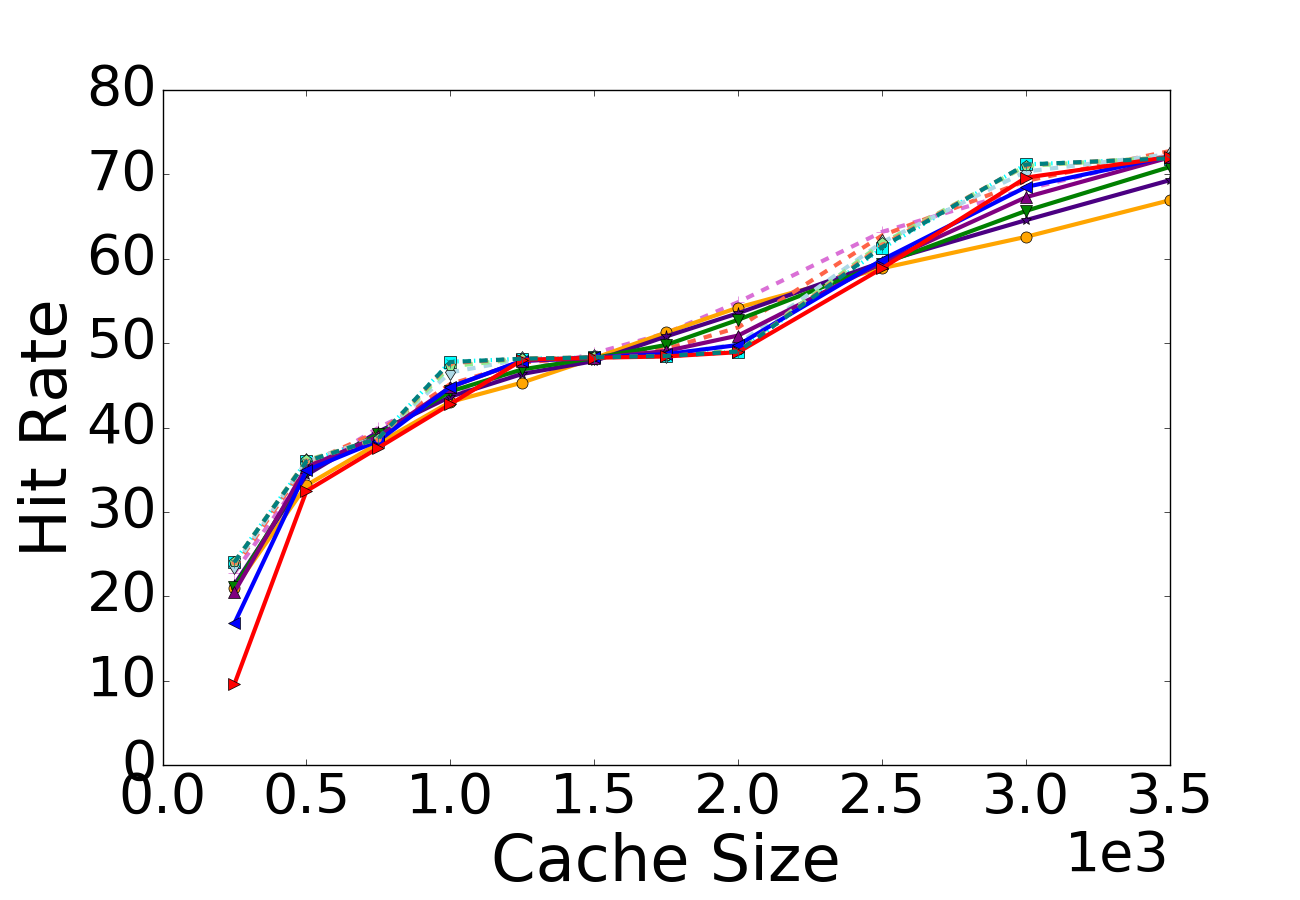}}
		\subfigure[LFU +TinyLFU]{\includegraphics[width=0.45\columnwidth]{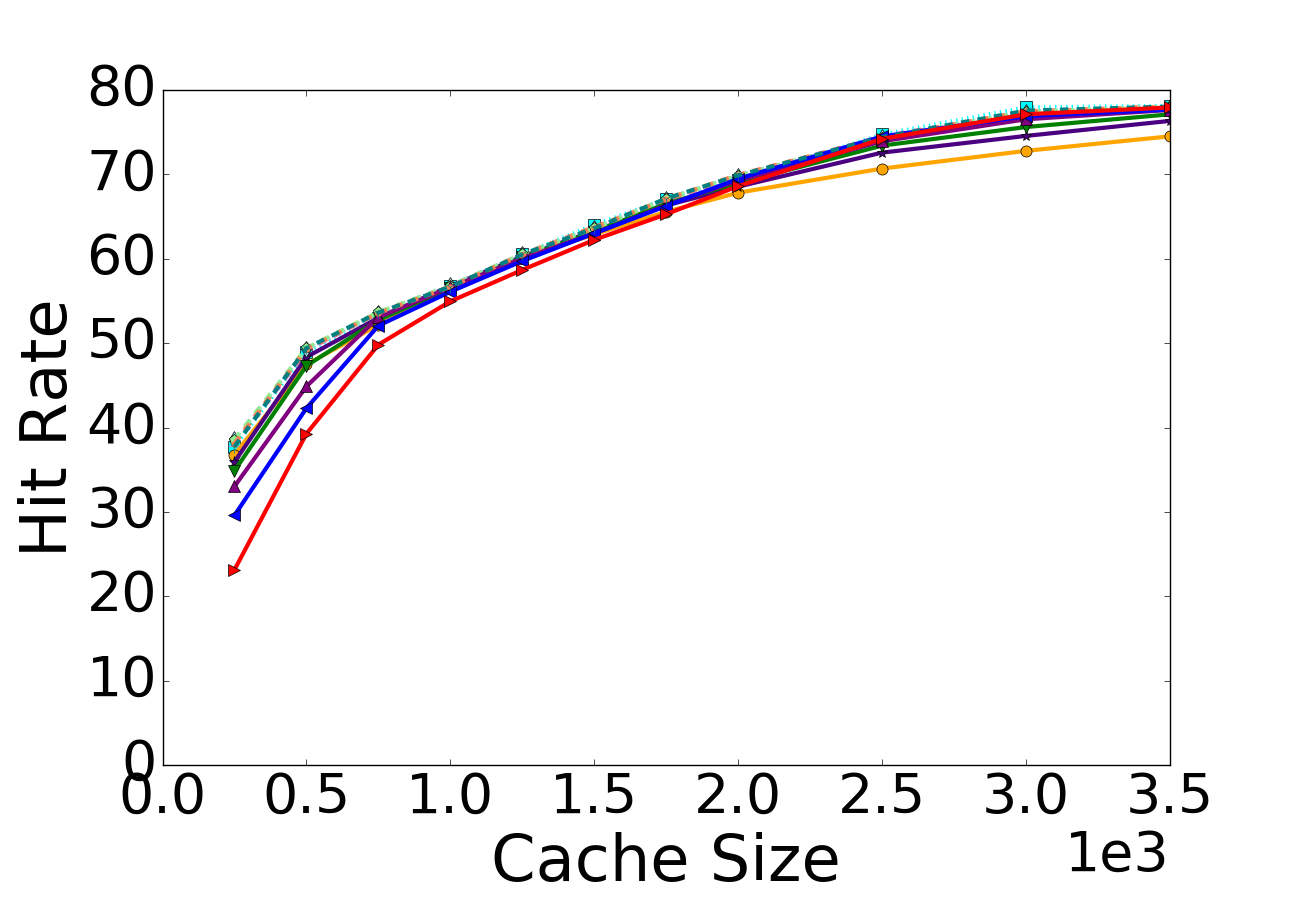}}
		\subfigure[Product]{\includegraphics[width=0.45\columnwidth]{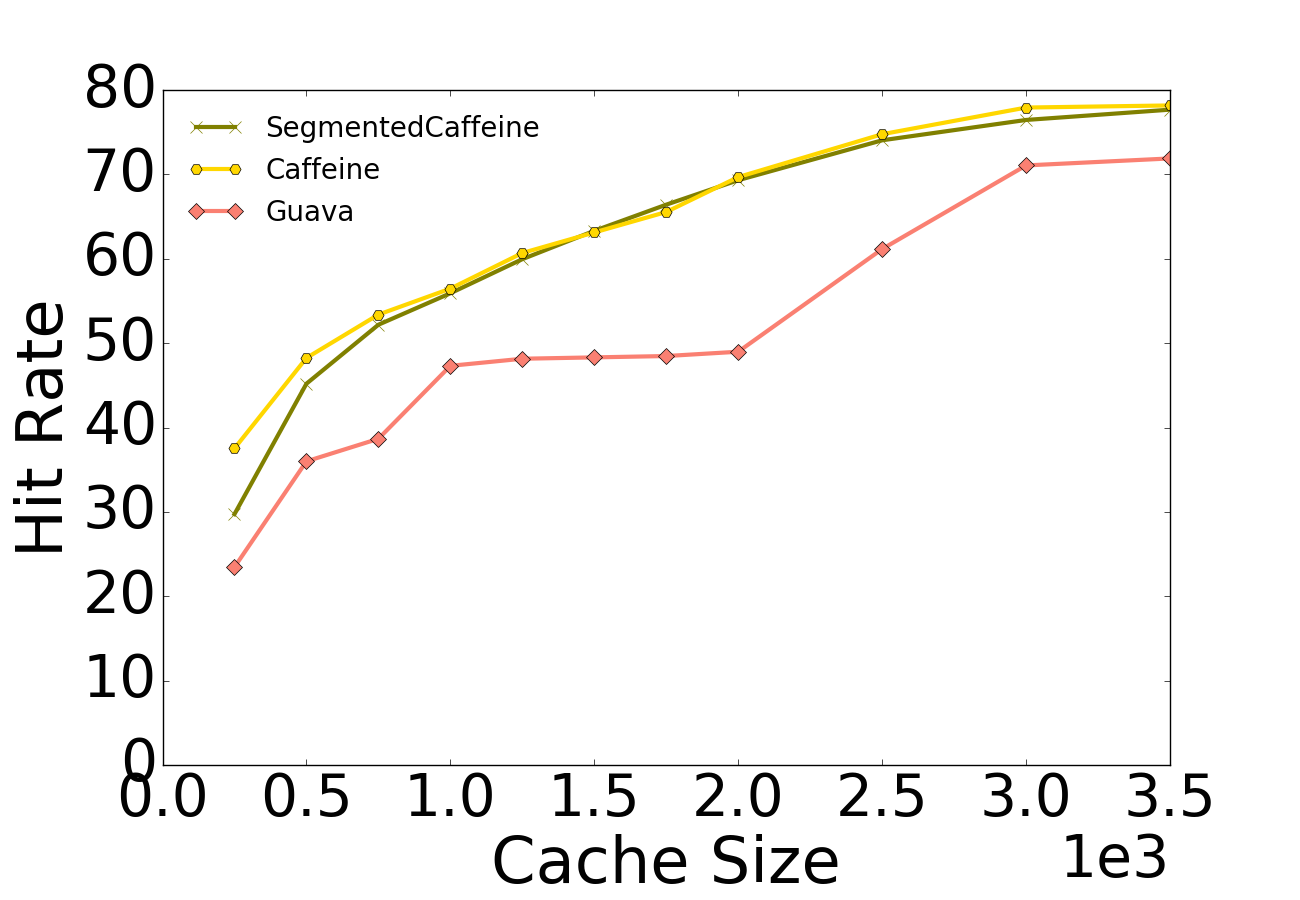}}	
		\subfigure[FIFO+TinyLFU]{\includegraphics[width=0.45\columnwidth]{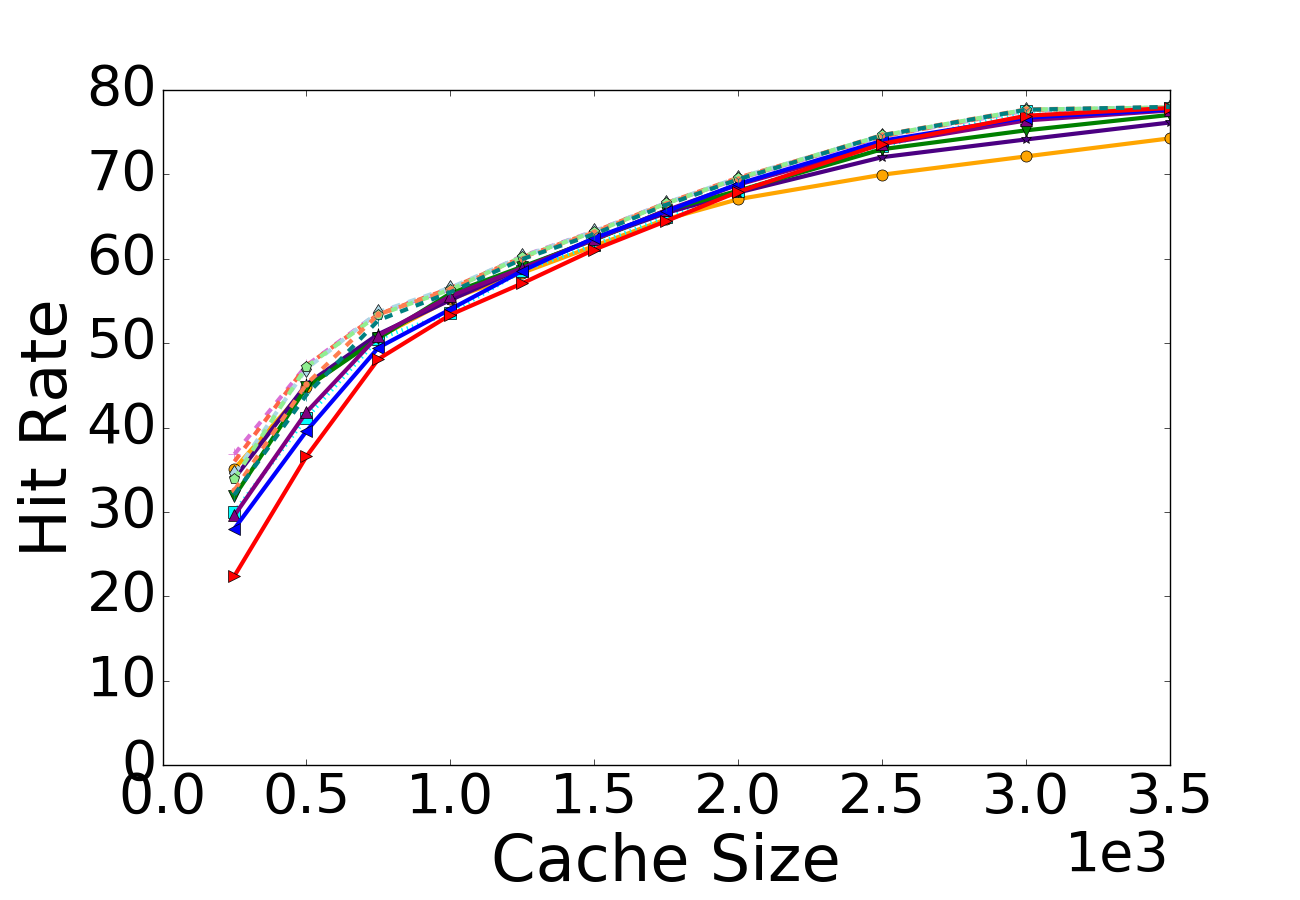}}	
	\end{center}
	\vspace{-0.5cm}
	\caption{LIRS - multi2.}
	\label{figmulti2}
	\vspace{-0.5cm}
\end{figure*}
}{}
\begin{figure*}[t]
	\begin{center}
		\offinterlineskip
		\subfigure[LRU]{\includegraphics[width=0.45\columnwidth]{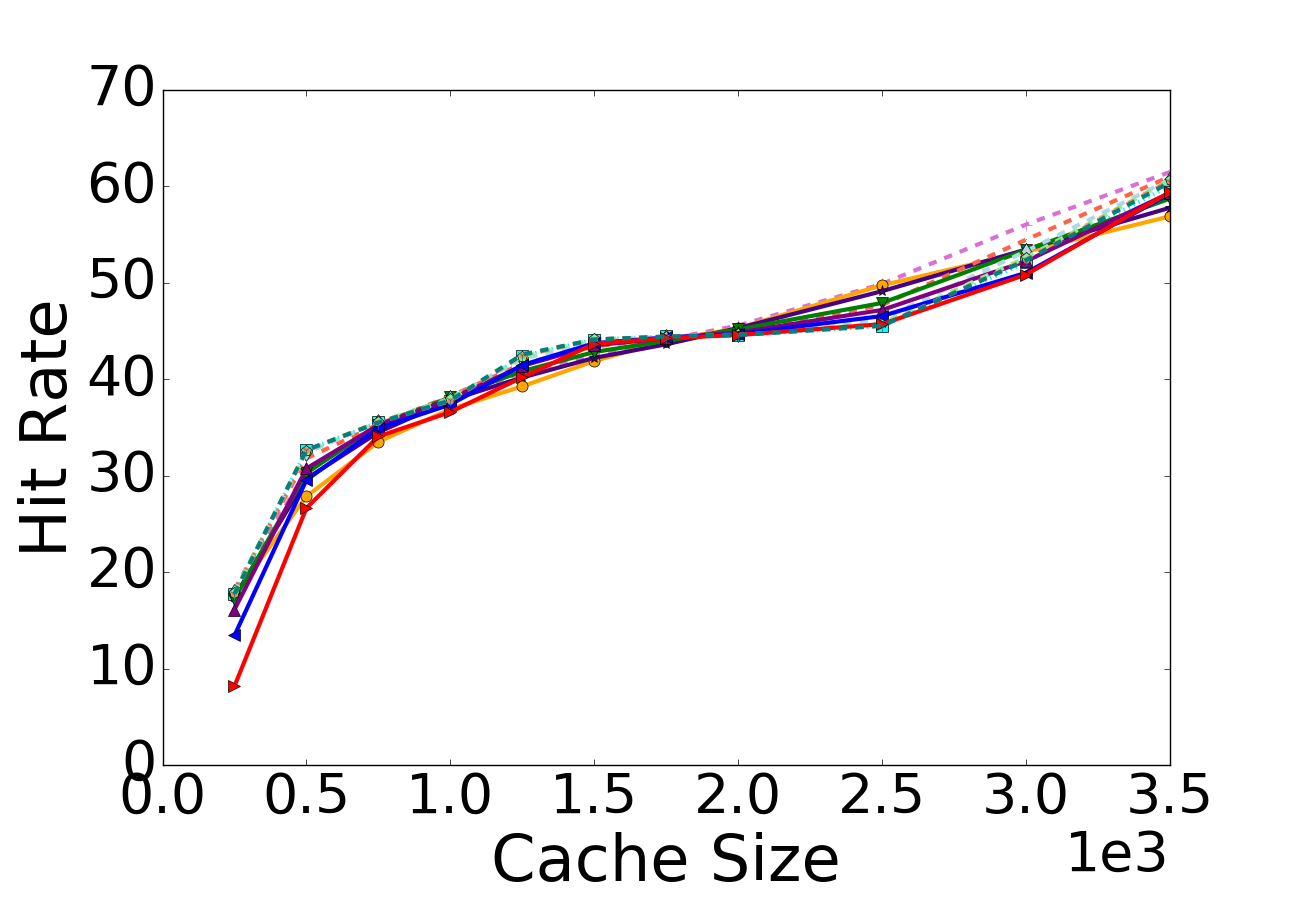}}
		\subfigure[LFU +TinyLFU]{\includegraphics[width=0.45\columnwidth]{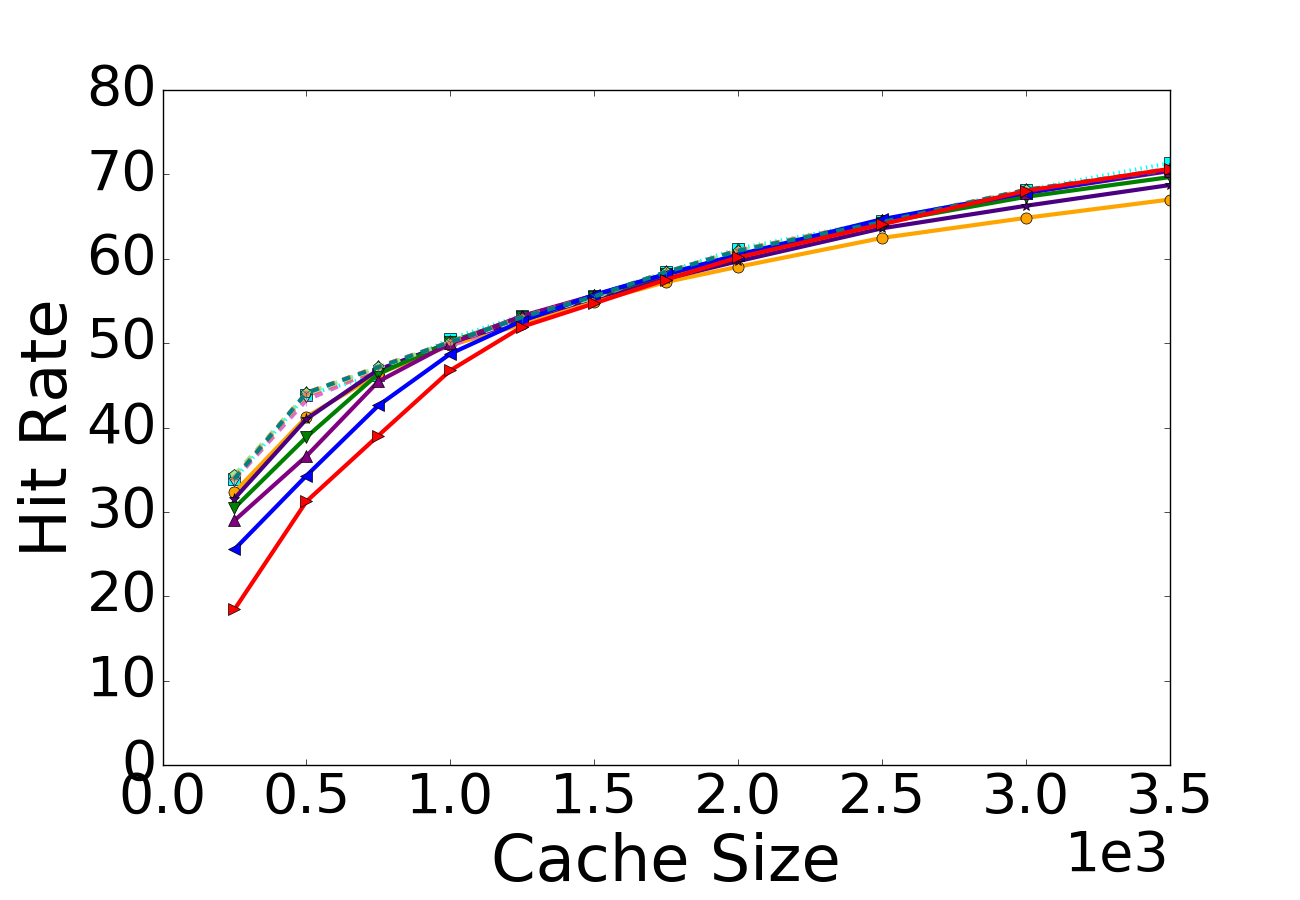}}
		\subfigure[Product]{\includegraphics[width=0.45\columnwidth]{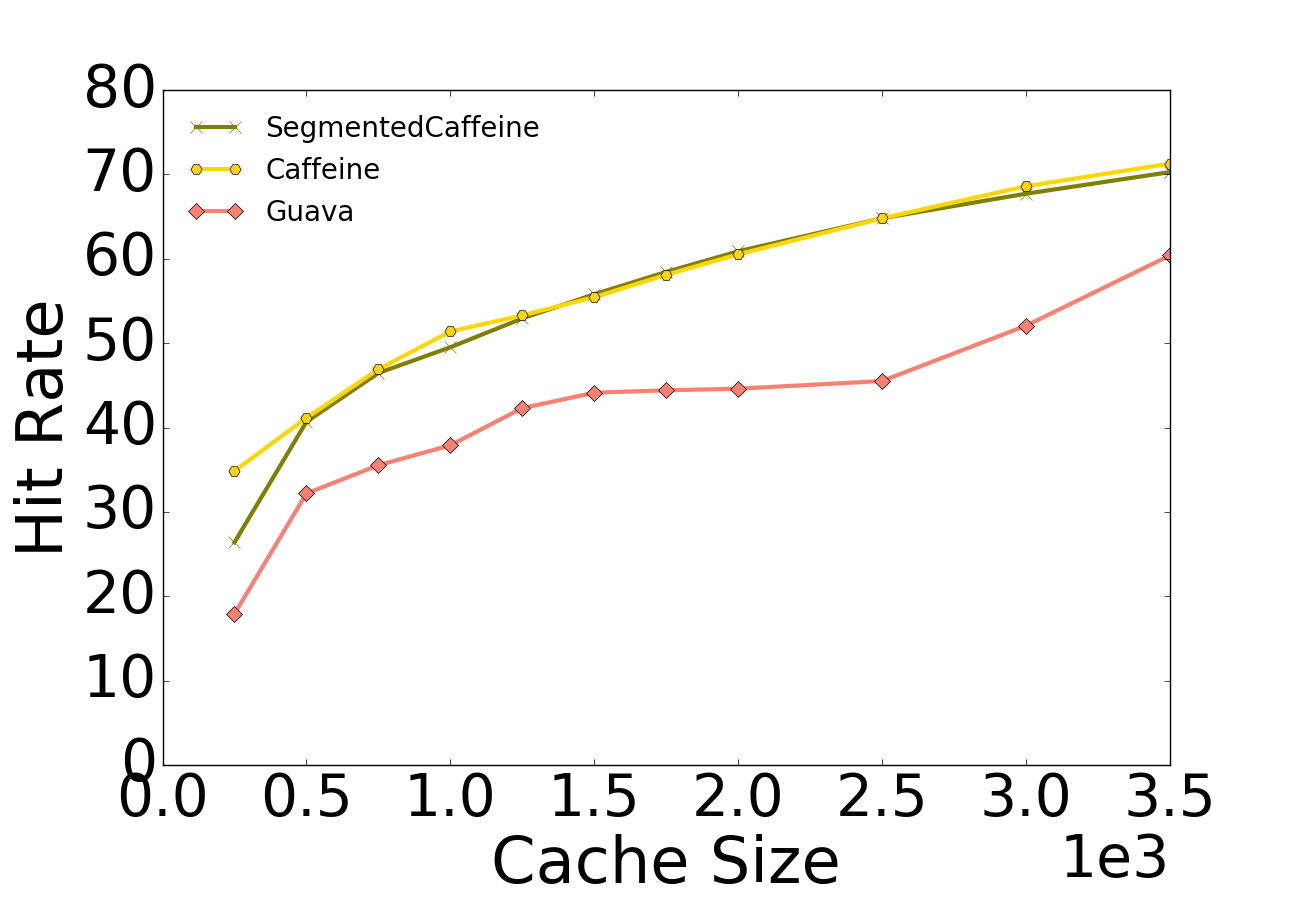}}	
		\subfigure[FIFO+TinyLFU]{\includegraphics[width=0.45\columnwidth]{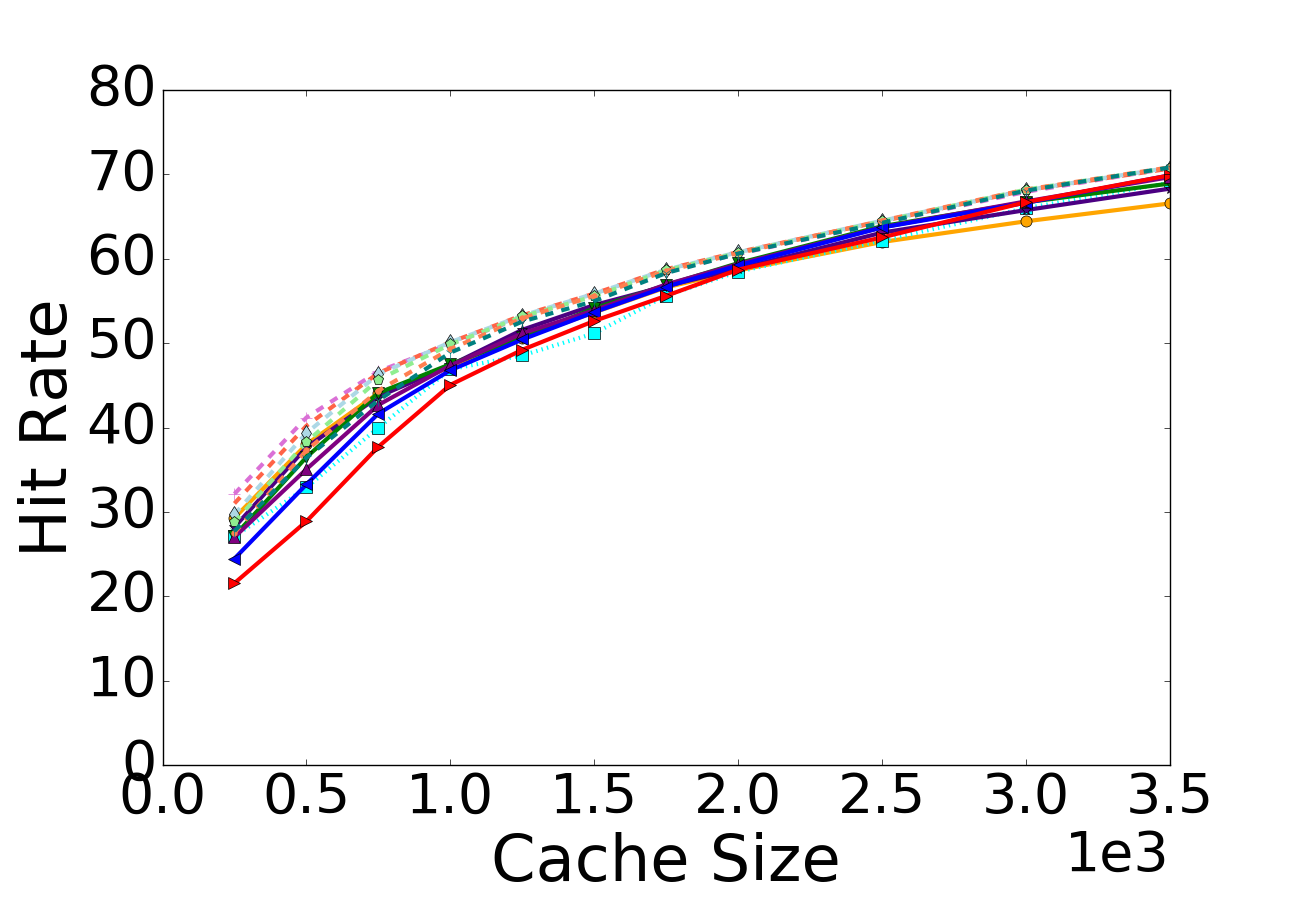}}	
	\end{center}
	\vspace{-0.5cm}
	\caption{multi3.}
	\label{figmulti3}
	\vspace{-0.5cm}
\end{figure*}

\begin{figure*}[t]
	\begin{center}
		\offinterlineskip
		\subfigure[LRU]{\includegraphics[width=0.45\columnwidth]{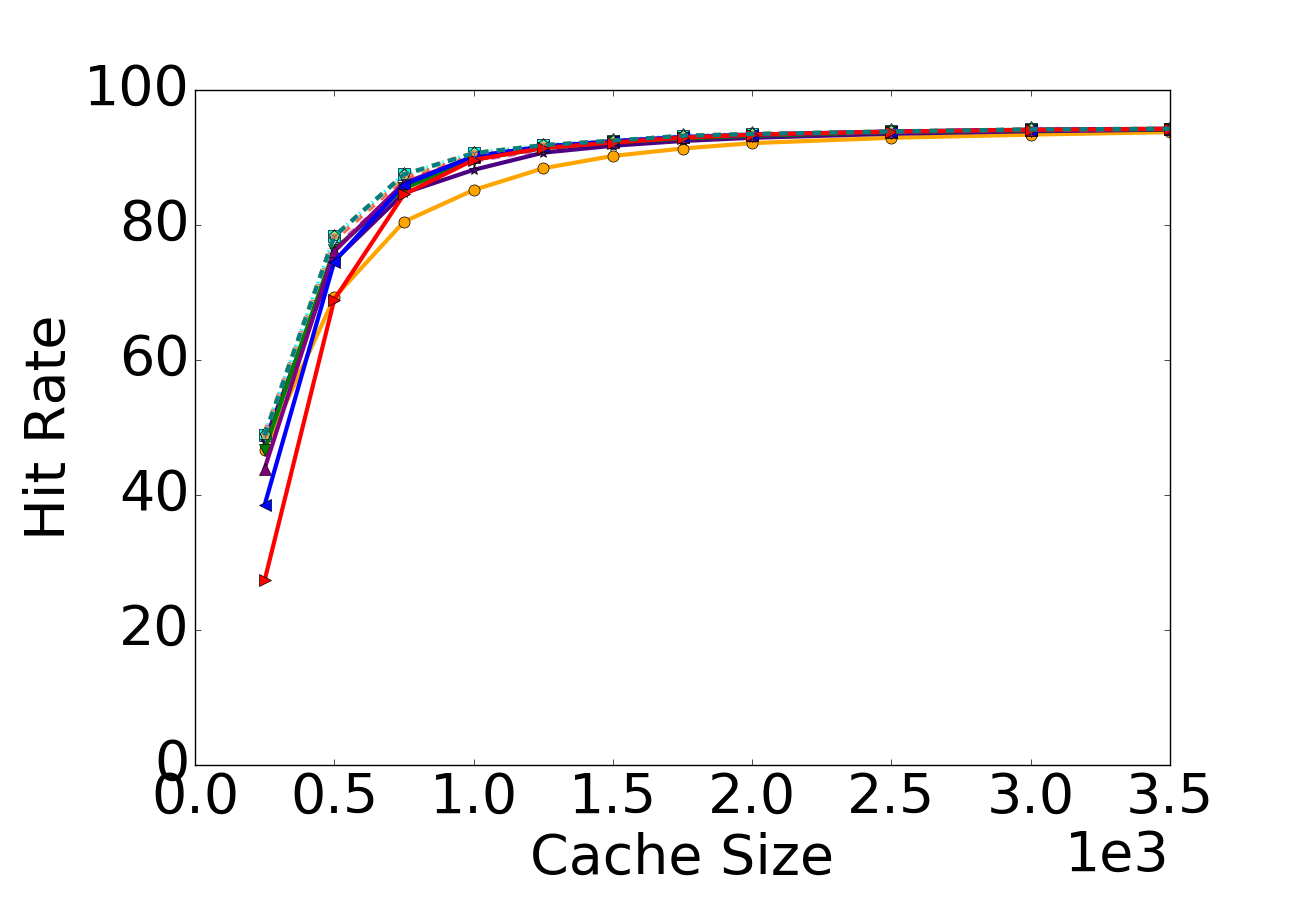}}
		\subfigure[LFU +TinyLFU]{\includegraphics[width=0.45\columnwidth]{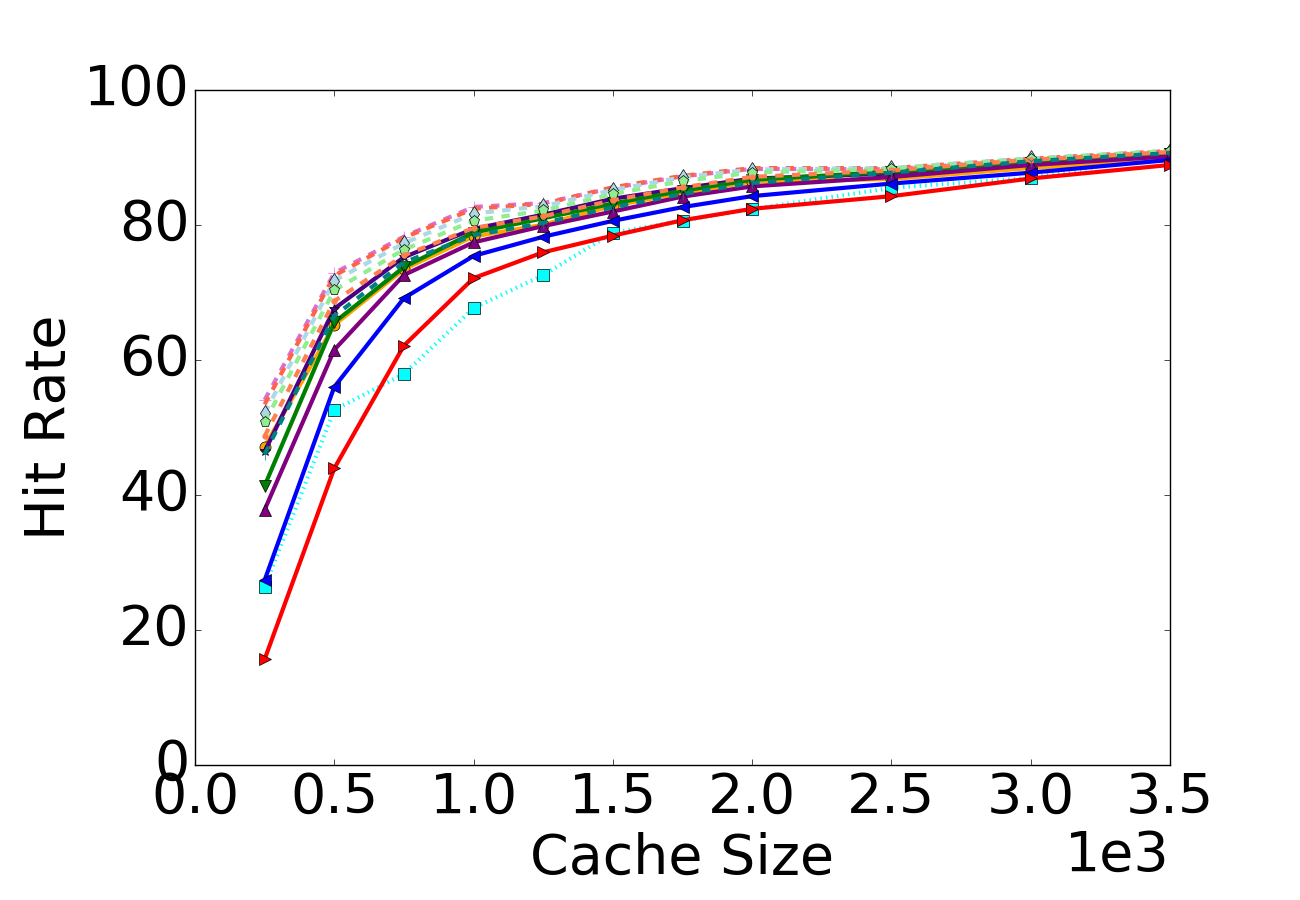}}
		\subfigure[Product]{\includegraphics[width=0.45\columnwidth]{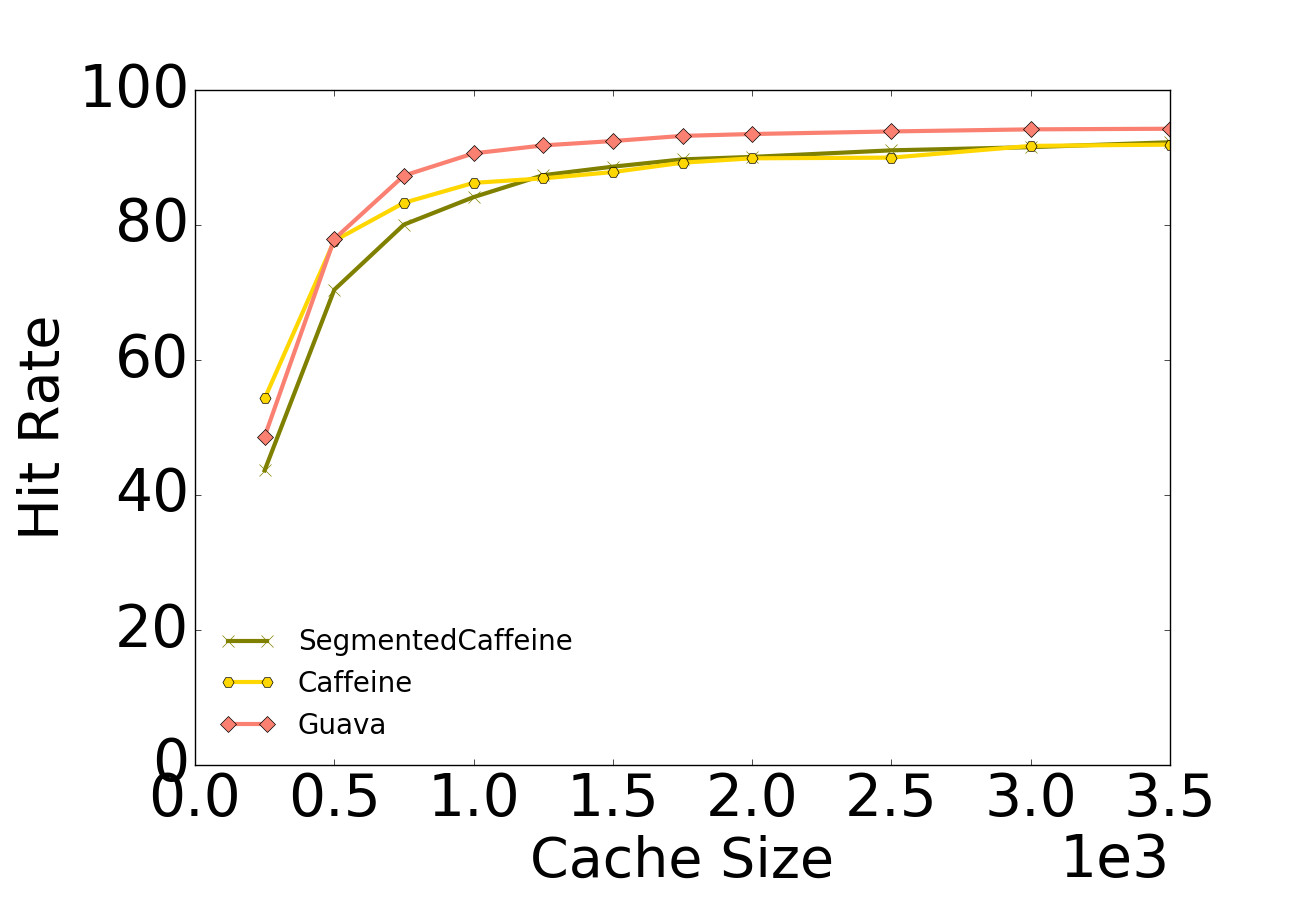}}			
		\subfigure[Hyperbolic]{\includegraphics[width=0.45\columnwidth]{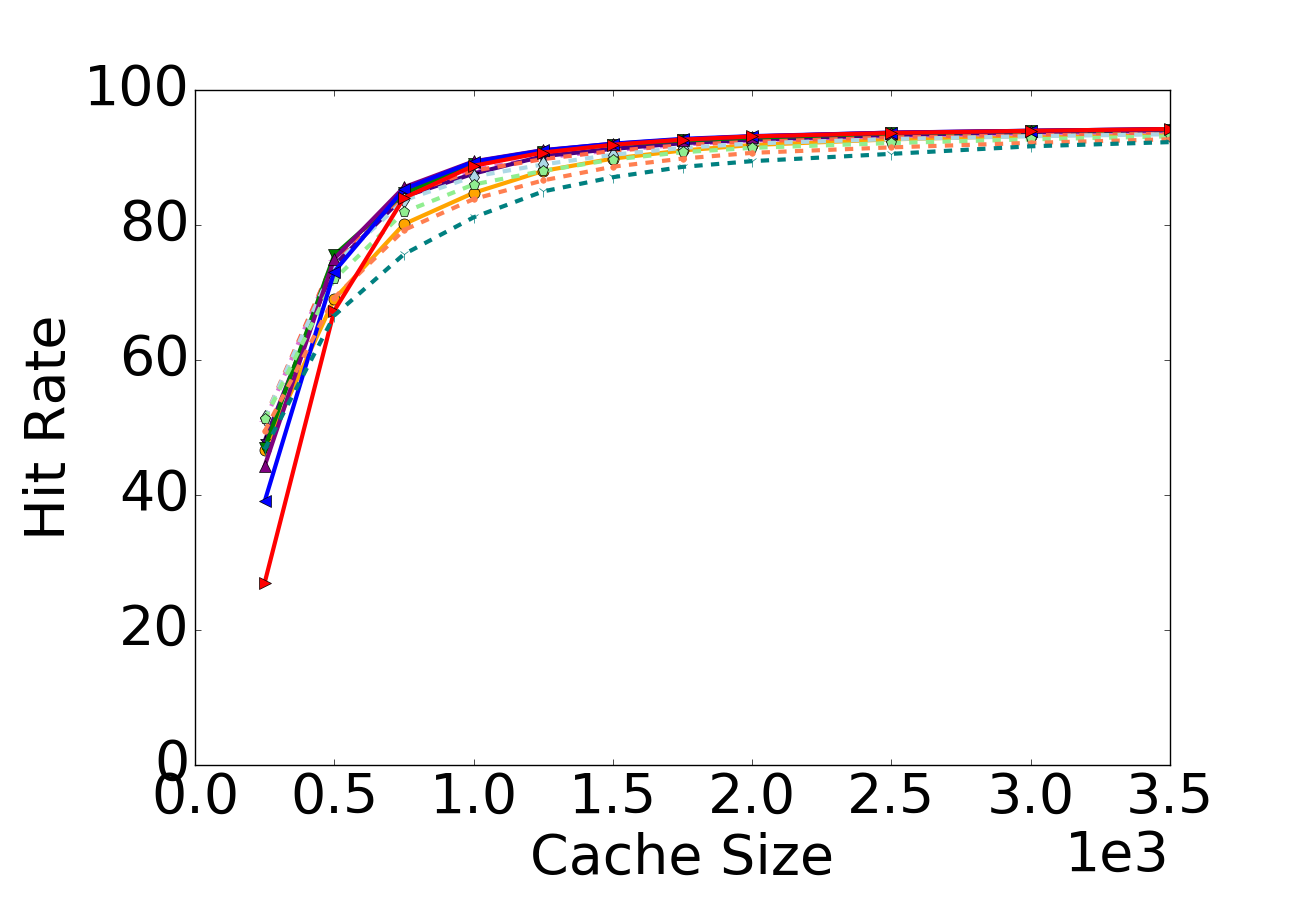}}
	\end{center}
	\vspace{-0.5cm}
	\caption{sprite.}
	\label{figsprite}
	\vspace{-0.5cm}
\end{figure*}

\nottoggle{MEDIUM}{
\begin{figure*}[t]
	\begin{center}
		\offinterlineskip
		\subfigure[LRU]{\includegraphics[width=0.45\columnwidth]{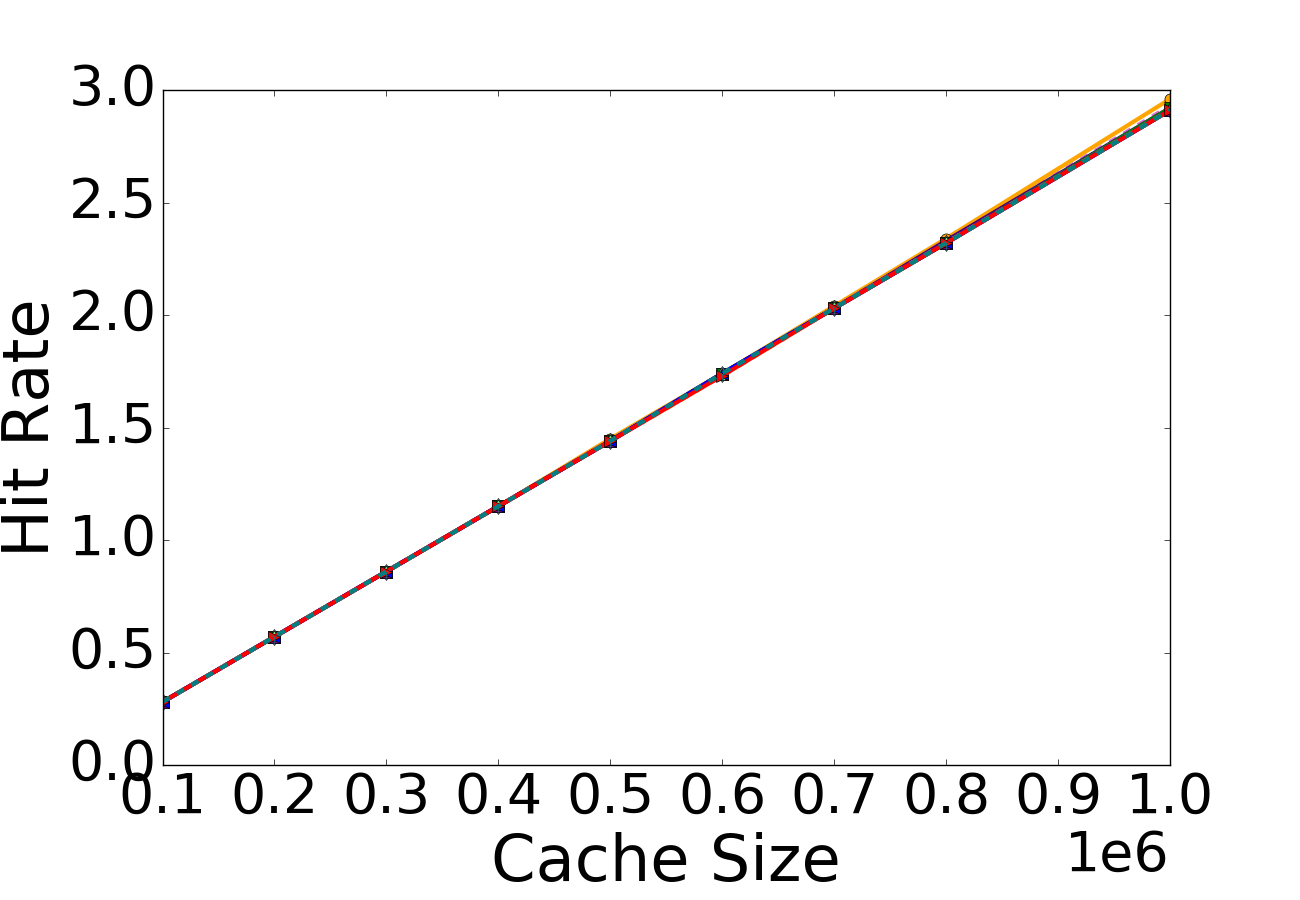}}
		\subfigure[[LFU +TinyLFU]{\includegraphics[width=0.45\columnwidth]{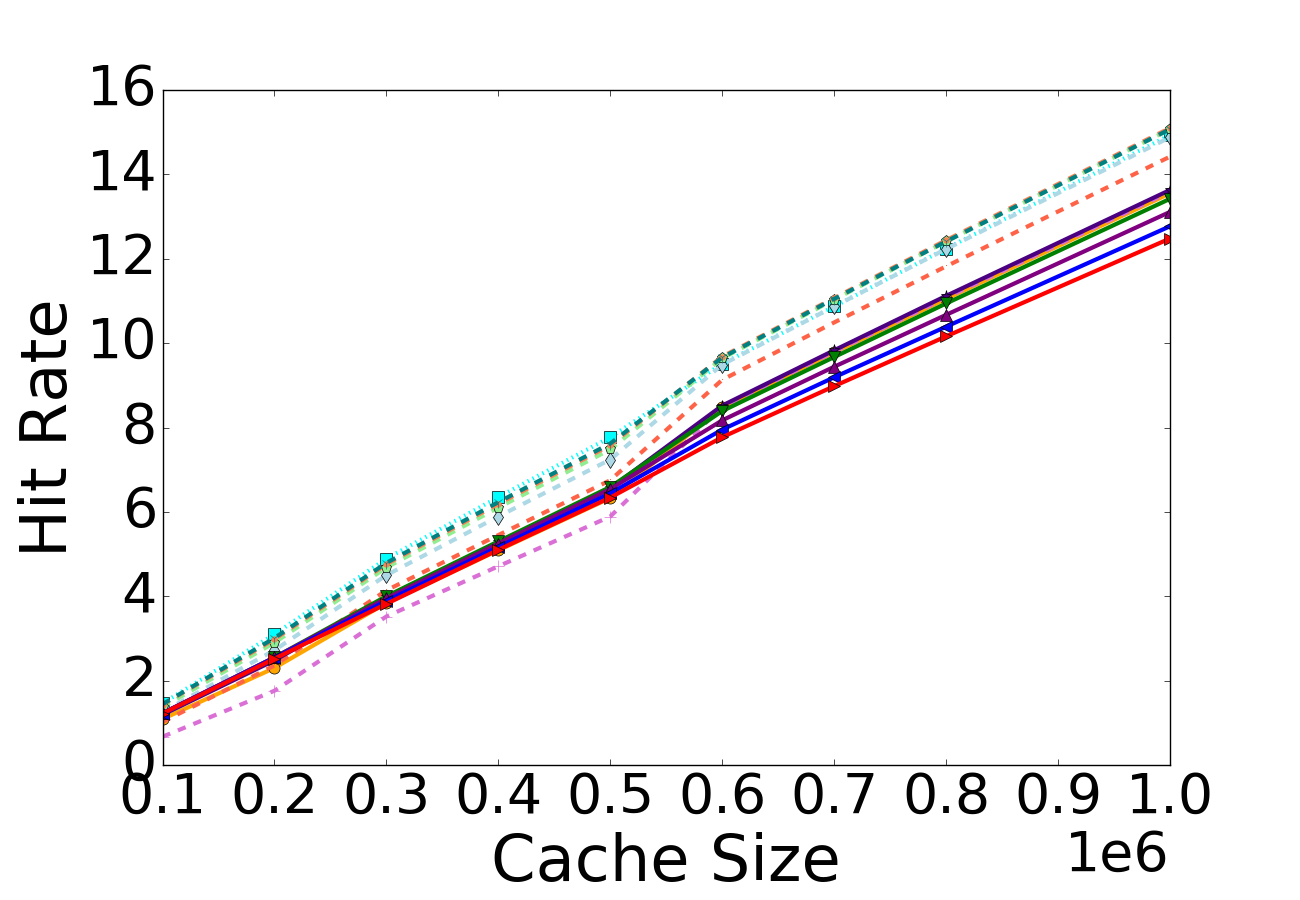}}
		\subfigure[Product]{\includegraphics[width=0.45\columnwidth]{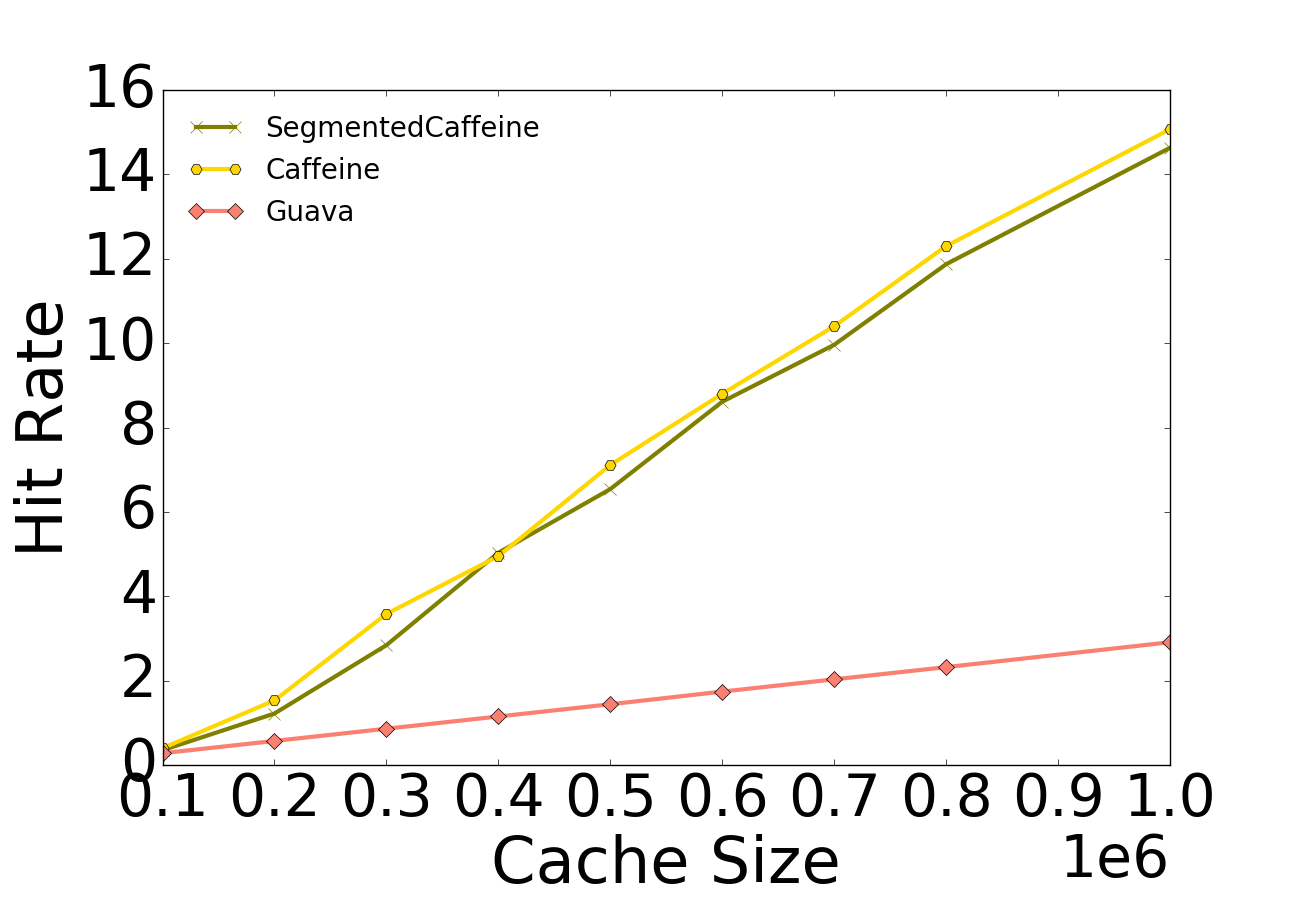}}	
		\subfigure[LFU]{\includegraphics[width=0.45\columnwidth]{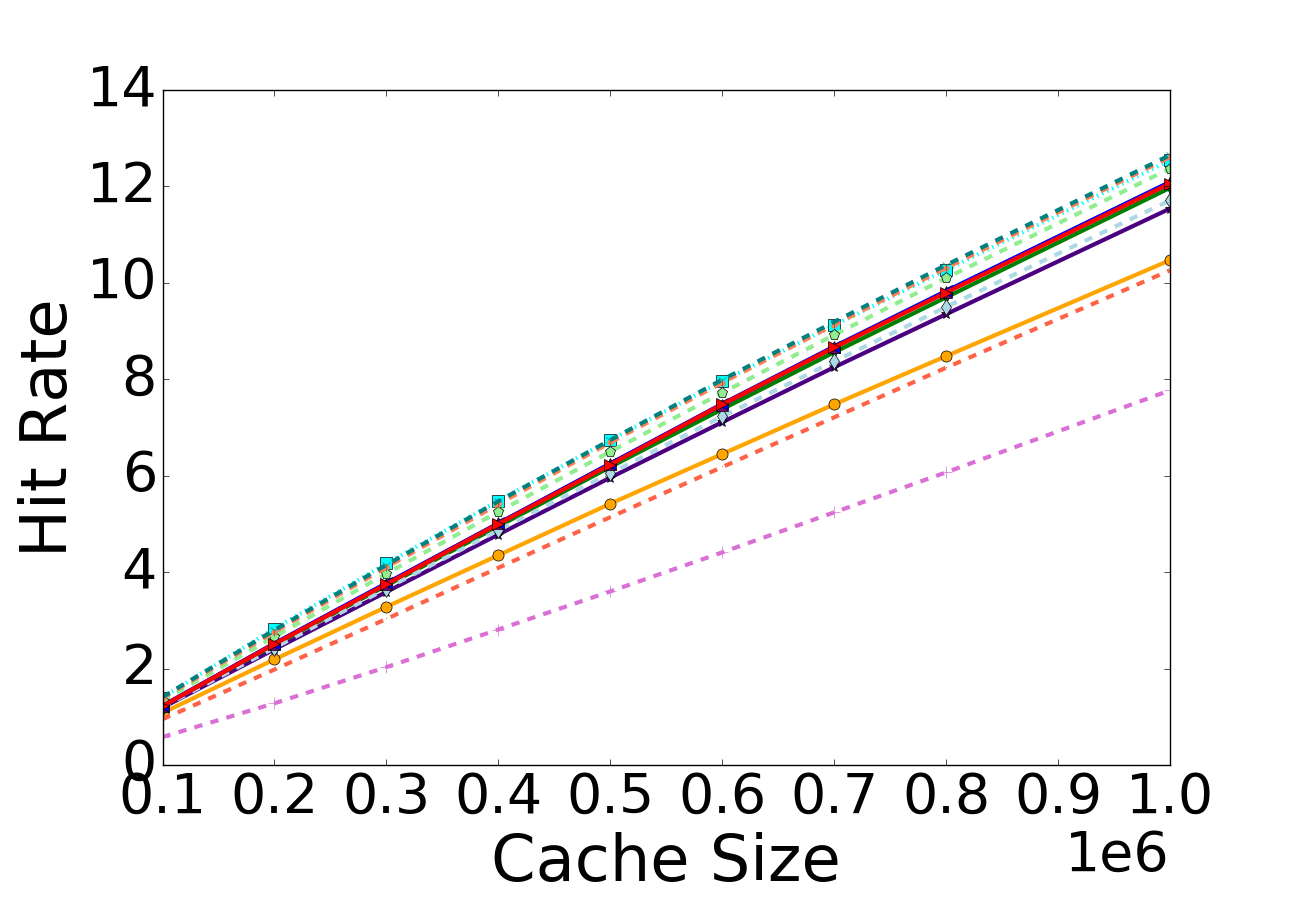}}
	\end{center}
	\vspace{-0.5cm}
	\caption{WebSearch2.}
	\label{figWebSearch2}
	\vspace{-0.5cm}
\end{figure*}
}{}

\nottoggle{SMALL}{
\begin{figure*}[t]
	\begin{center}
		\offinterlineskip
		\subfigure[LRU]{\includegraphics[width=0.45\columnwidth]{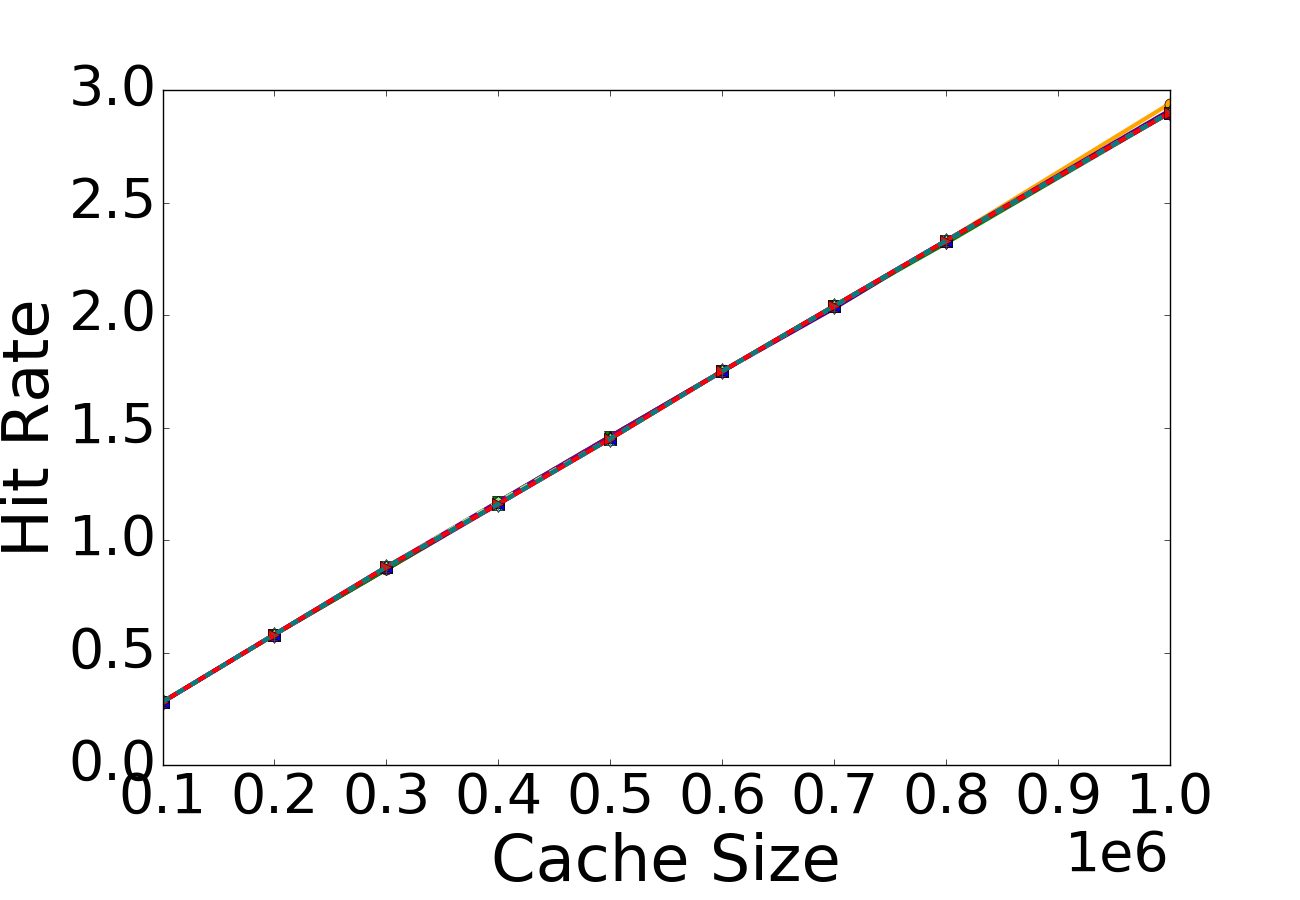}}
		\subfigure[[LFU +TinyLFU]{\includegraphics[width=0.45\columnwidth]{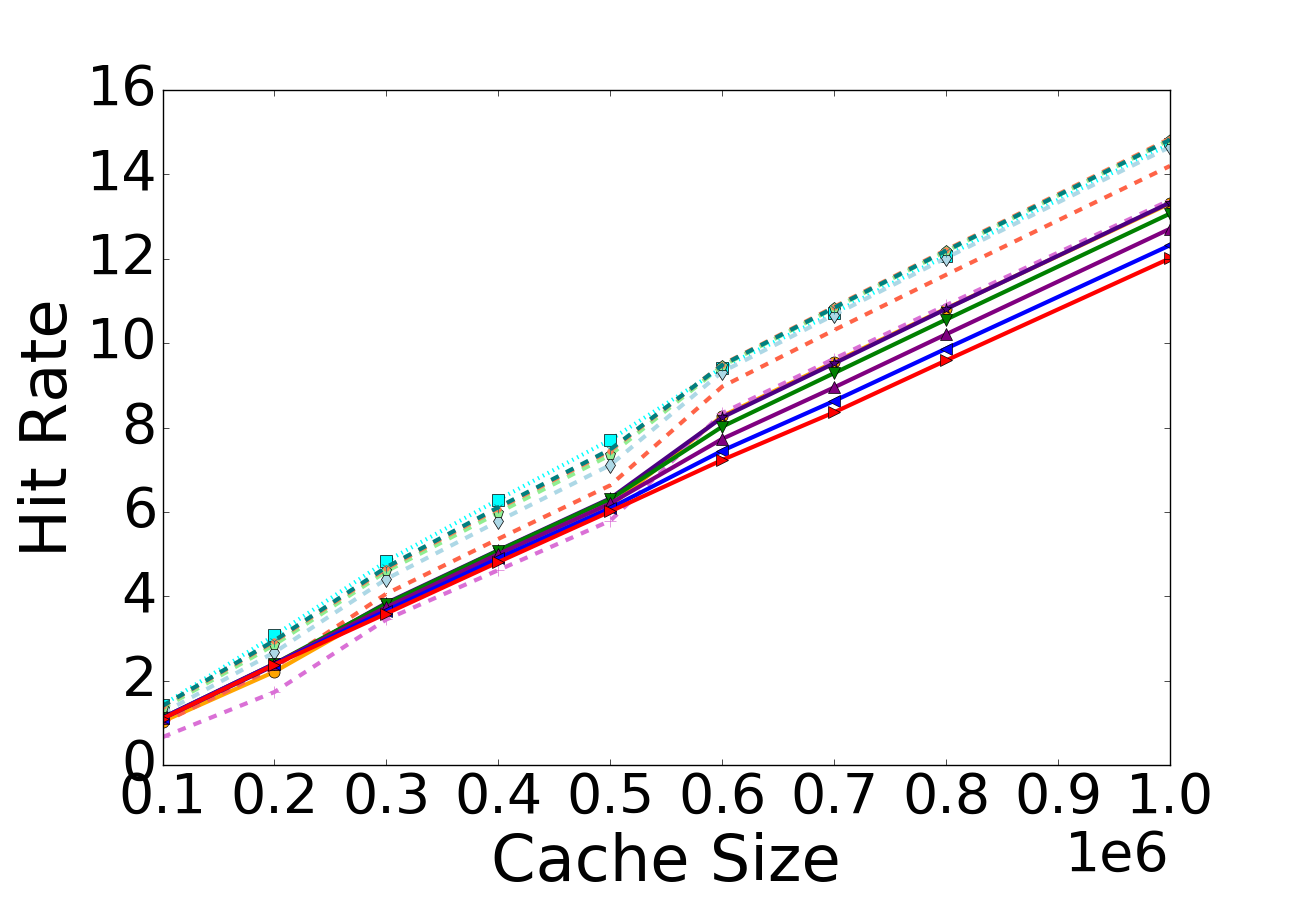}}
		\subfigure[Product]{\includegraphics[width=0.45\columnwidth]{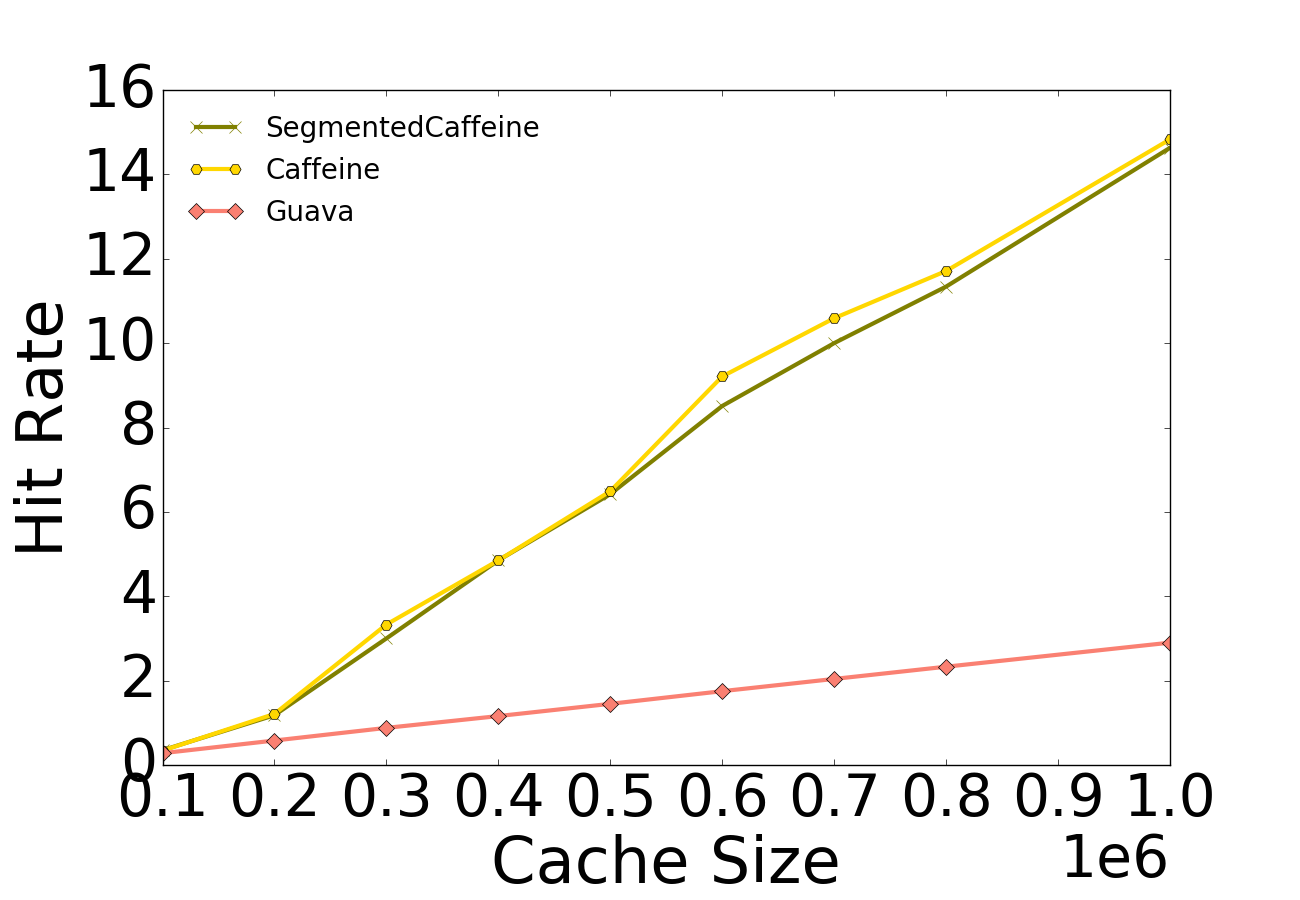}}	
		\subfigure[LFU]{\includegraphics[width=0.45\columnwidth]{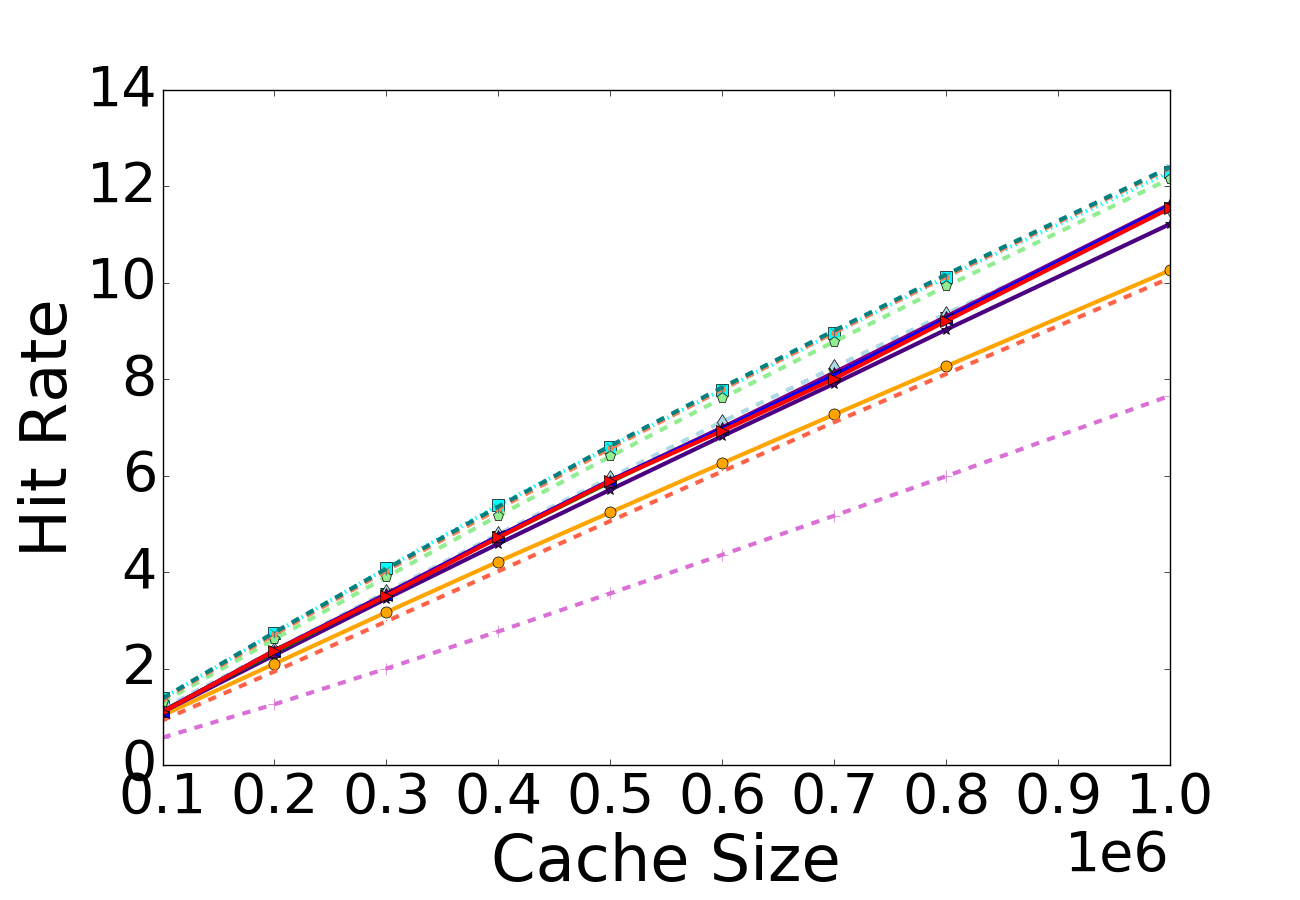}}
	\end{center}
	\vspace{-0.5cm}
	\caption{WebSearch3.}
	\label{figWebSearch3}
	\vspace{-0.5cm}
\end{figure*}
}{}
\subsection{Workloads and Setup}

We utilized the following real world workloads: 
\begin{description}[style=unboxed]
\item[wiki1190322952 and wiki1191277217:] two traces from Wikipedia containing 10\% of the traffic to Wikipedia during three months starting in September
2007~\cite{wikipedia}.

\item[sprite:] From the Sprite network file system, which contains requests to a file server from client workstations for a two-day period by~\cite{LIRS}.
\item[multi1:] This trace is obtained by executing two workloads, cs, an interactive C source program examination tool trace, and cpp, a GNU C compiler pre-processor trace  by~\cite{LIRS}.
\item[multi2:] This is obtained by executing three workloads, cs, cpp and postgres, a trace of join queries among four relations in a relational database system from the UC Berkeley, provided by~\cite{LIRS}.
\item[multi3:] This is obtained by executing four workloads, cs, cpp, glimpse, a text information retrieval utility, and postgres provided by~\cite{LIRS}.
\item[OLTP:] A file system trace of an OLTP server~\cite{ARC}.
\nottoggle{SMALL}{\item[DS1:] A database trace provided by~\cite{ARC}.}
\iftoggle{SMALL}{\item[S3:] Search engine trace provided by~\cite{ARC}.}{\item[S1 and S3:] Search engine traces provided by~\cite{ARC}.}
\iftoggle{SMALL}{\item[P12:] Windows server disc accesses~\cite{ARC}.}{\item[P8, P12 and P14:] Windows server disc accesses~\cite{ARC}.}
\item[F1 and F2:] Traces of transaction processing taken from large financial institution. The trace is provided by the UMass trace repository~\cite{umasstrace}.
\iftoggle{SMALL}{\item[W3:] A search engine trace provided by~\cite{umasstrace}.}{\item[W2 and W3:] Search engine traces provided by~\cite{umasstrace}.}
\end{description}

In all figures below, the ``$k$ ways'' and ``sampled'' lines are \NAMNECACHE{} with set/sample sizes of $4,8,16,32,64$ \& $128$. The  ``fully associative'' line stands for a  linked-list based fully associative implementation.
The ``product'' figures for each trace present Guava cache by Google~\cite{guava-cache}, Caffeine~\cite{CaffeineProject} and segmented Caffeine~\cite{Seg-Caffeine}, which is Caffeine divided into 64 segments to enable Caffeine to benefit from parallelism at the possible expense of reduced hit-ratios (plain Caffeine serves updates by a single thread).

In all throughput evaluations, we compare \NAMNECACHE{} with $k = 8$ to sampled implementations of corresponding methods with sample size = 8, Guava cache, Caffeine and segmented Caffeine.
All evaluated cache implementations are written in Java.
Our \NAMNECACHE{} cache uses xxHash~\cite{xxHash} as the hash function that distributes items to sets.

\subsubsection{Platform}
We ran on the following platforms:
\begin{itemize}
    \item AMD PowerEdge R7425 server with two AMD EPYC 7551 Processors, each with $32$ 2.00GHz/2.55GHz cores that multiplex $2$ hardware threads, so in total, this system supports $128$ hardware threads.
	The hardware caches include 32K L1 cache, 512K L2 cache and 8192K L3 cache.
	It has $128$GB DRAM organized in 8 NUMA nodes, 4 per processor.
	
	\item Intel Xeon E5-2667 v4 Processor including $8$ 3.20GHz cores with 2 hardware threads, so this system supports $16$ hardware threads.
	It has $128$GB DRAM.
	The hardware caches include 32K L1 cache, 256K L2 cache and 25600K L3 cache.
\end{itemize}

\subsubsection{Methodology}
During throughput evaluation, we employed the following typical cache behavior:
For each element in the trace, we first performed a read operation.
If the element was not in the cache, we initiated a write operation of that element into the cache.

We started with a warm-up for filling the cache with elements not in the trace. 
We did the warm-up first in the main thread that inserted elements up to the cache size and then in each thread by inserting $ size/\#threads$ non cached elements that are not in the trace.
All threads started the test after the warm-up simultaneously by waiting on a barrier.

We ran each test for a fixed period, counting the number of operations.
The time was calculated in comparison to the size of the trace and is between 1 and 4 seconds.
Each trace was run with a different cache size taking into account the size of the trace and the hit rate.
Each point in the graphs is the mean taken over 11 runs.

\subsection{Hit Ratio Evaluation}
Figure~\ref{fig:legend} shows the legend for all the hit ratio and throughput graphs shown in this work.
\Cref{fig:wiki1190322952,figP12,figS1,figS3,figOLTP,figmulti2,figmulti3,figsprite,figWebSearch3}
depict the results for the traces described above with different eviction and admission policies.
For brevity, not all traces and combinations are shown.
The ones removed essentially exhibit similar results.

The first graphs for all traces is LRU due to its popularity.
The second graph is LFU eviction with TinyLFU admission, first proposed in~\cite{EFM17}.
LFU is also a popular policy, but LFU alone is not a very good management policy as it lacks aging mechanisms and lacks data about non-cached objects, which are added by the TinyLFU admission mechanism.
The third graph includes the real caching products: Guava~\cite{guava-cache}, Caffeine~\cite{CaffeineProject} and the proof-of-concept segmented Caffeine (with number of  independent segments that match the number of threads tested)~\cite{Seg-Caffeine}.
For brevity, the fourth graph presents different (additional) policies for different traces.
We summarise the results below:

\paragraph{LRU eviction policy}
Recall that LRU evicts the least recently used item (and admits all items).
In the limited associativity model, we first assign a new item to one of the small-sized sets based on hashing its key and then evict the least recently used item of that set.
As can be seen from the graphs, for most traces, the level of associativity has minimal impact on the obtained hit ratio.

\paragraph{LFU eviction with LFU admission}
The results for this policy (in subfigures b) show that associativity has a minor impact on the hit ratio.
Also, observe that the performance of the sampled approach is similar to that of limited associativity.
In some cases one is slightly better than the other and in others it is the opposite.

\paragraph{Products}
In these graphs (subfigures c), we compared Guava cache~\cite{guava-cache}, Caffeine~\cite{CaffeineProject} and segmented Caffeine~\cite{Seg-Caffeine}.
In the latter, we used a hash function to map each element to one of the corresponding Caffeine managed segments.
Each such Caffeine instance was constructed with MAX SIZE/$\#threads$, to hold MAX SIZE elements in total.  
As can be observed, the hit-ratio of segmented Caffeine is nearly identical to that of Caffeine. 
Also, notice that Caffeine attains higher hit-ratio than Guava as it uses the climber policy from~\cite{EEFM18} compared to some version of LRU in Guava. 

\paragraph{Hyperbolic Caching}
We also show that limited associativity has little effect on Hyperbolic
\iftoggle{SMALL}{caching~(\cref{figP12}) and on Hyperbolic caching with the TinyLFU admission policy~(\cref{figS3}).}
{caching~(\cref{figP12,figsprite}) and on Hyperbolic caching with the TinyLFU admission policy~(\cref{figS3}).}
Notice that for Hyperbolic caching, there are slight differences between sampling and limited associativity.
When combining Hyperbolic with TinyLFU admission, these differences become barely noticeable.

\paragraph{Conclusions:}
This motivational study concludes that limited associativity only has a minor impact on the obtained hit ratio. 


\subsection{Throughput Evaluation on Real Traces}

\begin{figure}[t]
	\center{
		\includegraphics[width=0.75\columnwidth]{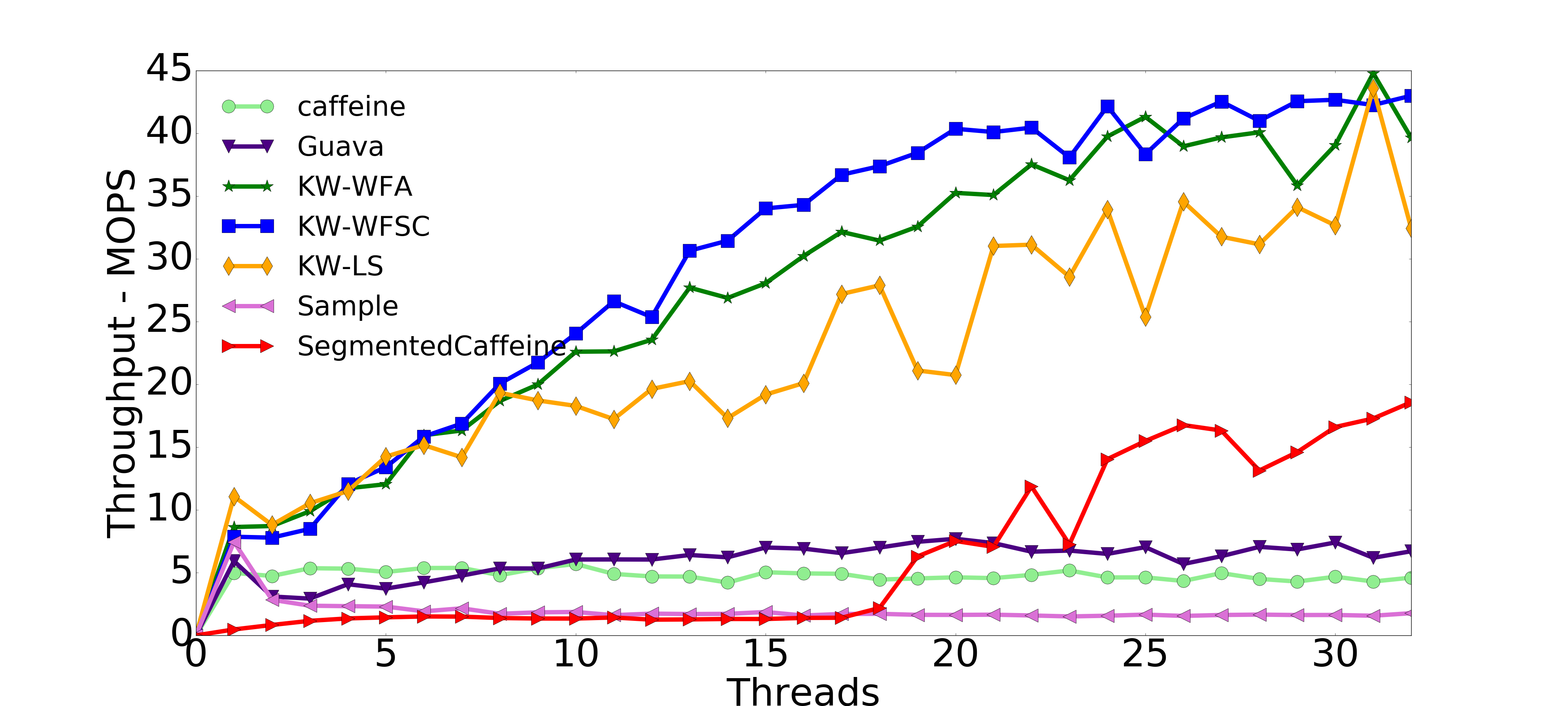}}
	\caption {F1 trace with cache size of $2^{11}$ elements and duration of run of 1 second, run on AMD PowerEdge R7425. }
	\label{fig:f1Throughput}
\end{figure}
\begin{figure}[t]
	\center{
		\includegraphics[width=0.75\columnwidth]{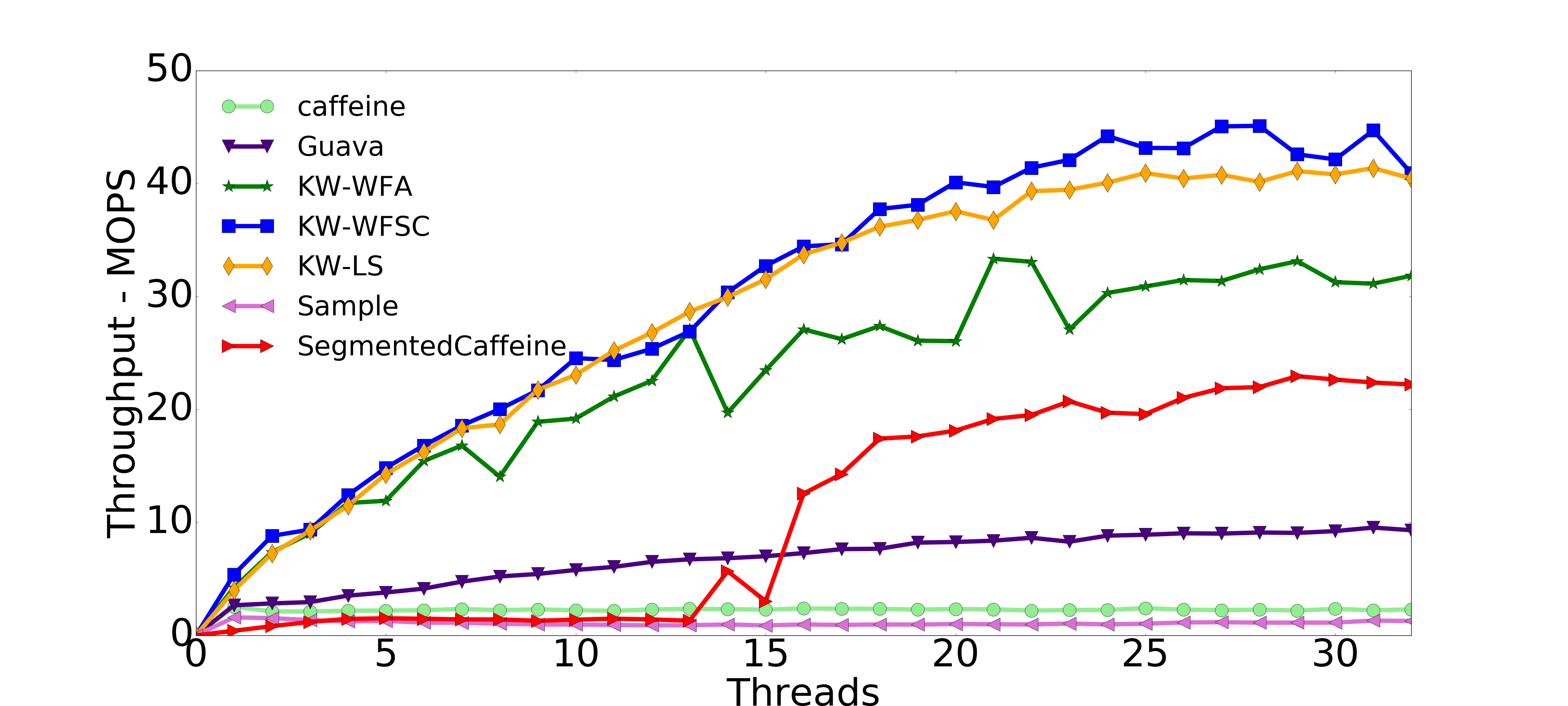}}
	\caption {S3 trace with cache size of $2^{19}$ elements and duration of run of 4 second, run on AMD PowerEdge R7425. }
	\label{fig:S3Throughput}
\end{figure}
\nottoggle{SMALL}{
\begin{figure}[t]
	\center{
		\includegraphics[width=0.75\columnwidth]{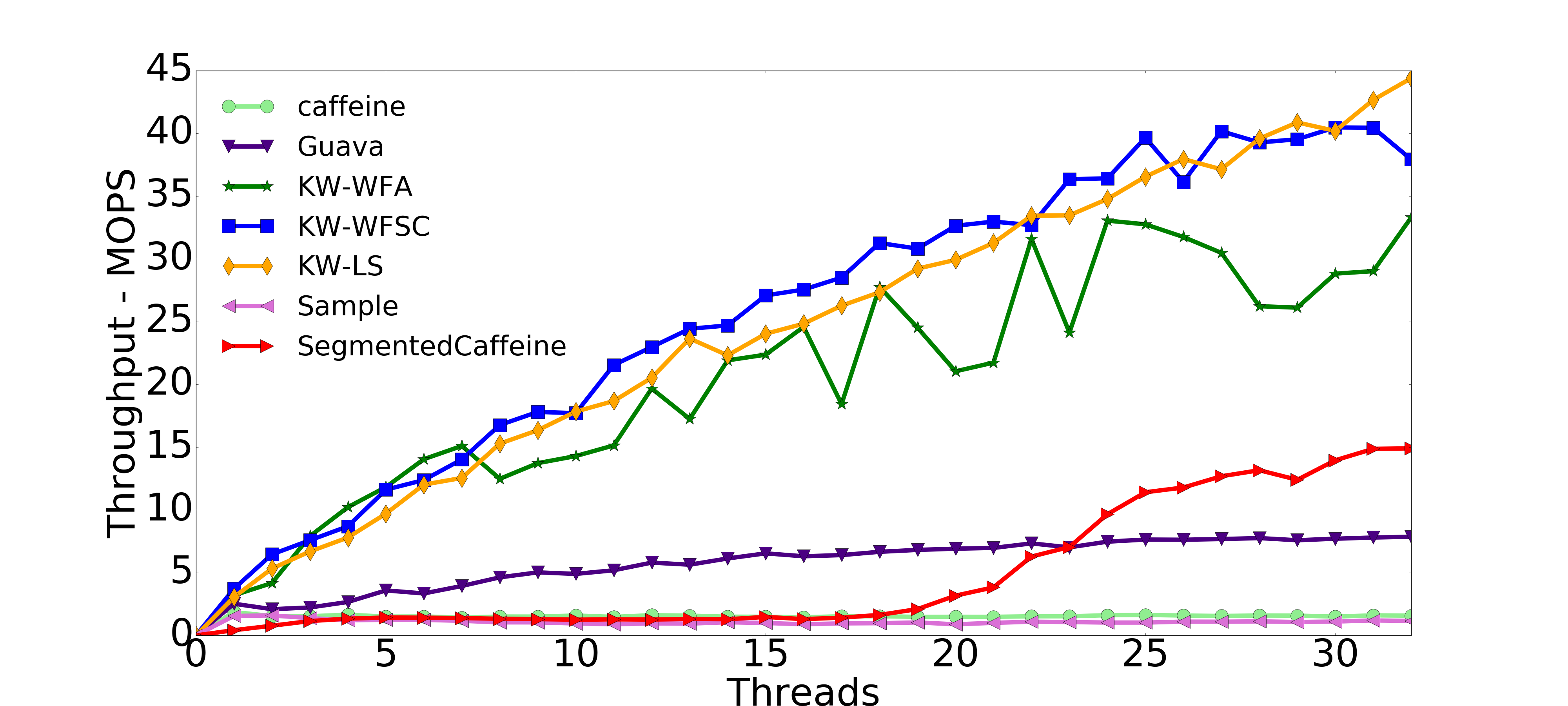}}
	\caption {S1 trace with cache size of $2^{19}$ elements and duration of run of 4 second, run on AMD PowerEdge R7425. }
	\label{fig:S1Throughput}
\end{figure}
}{}
\begin{figure}[t]
	\center{
		\includegraphics[width=0.75\columnwidth]{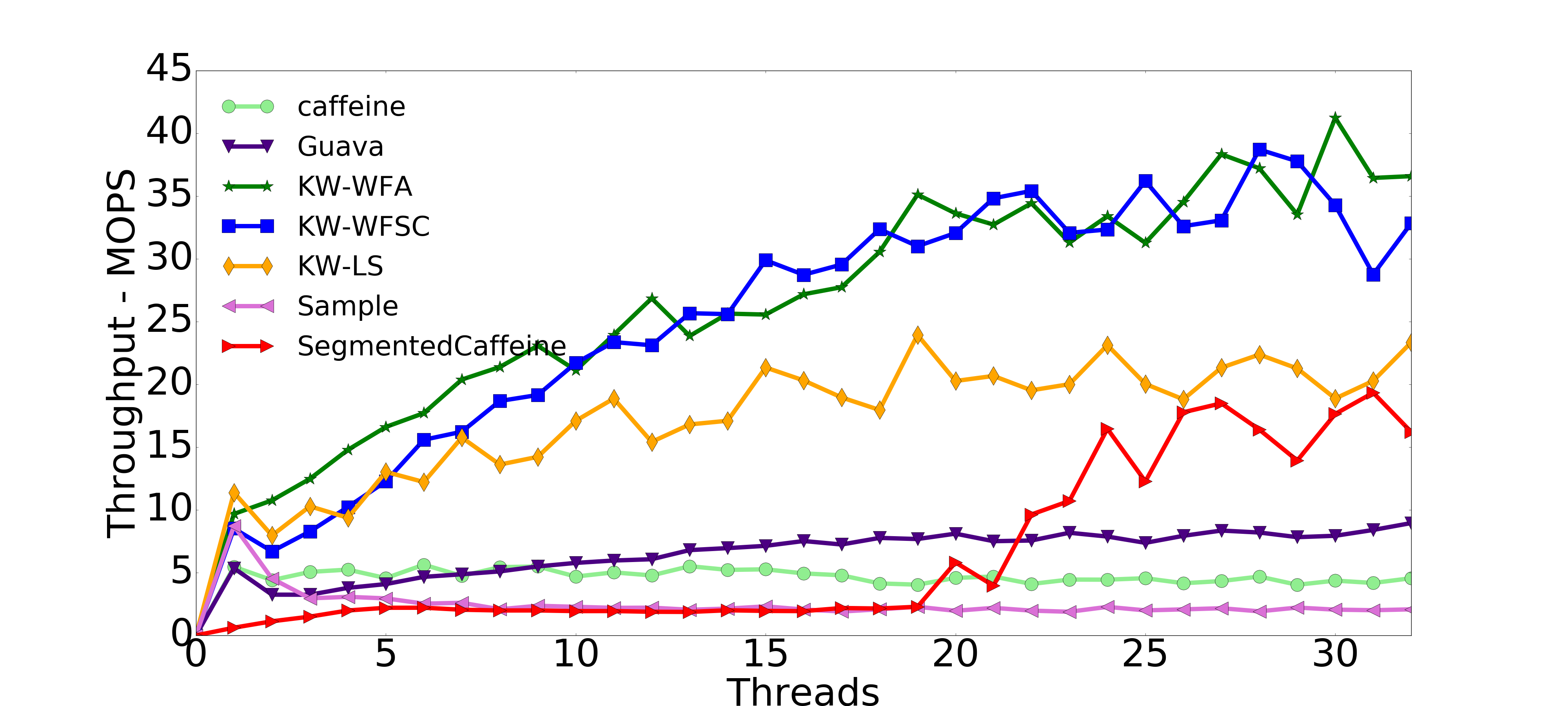}}
	\caption {wiki1190322952 trace with cache size of $2^{11}$ elements and duration of run of 1 second, run on AMD PowerEdge R7425. }
	\label{fig:wiki1190322952throughput}
\end{figure}

\begin{figure}[t]
	\center{
		\includegraphics[width=0.75\columnwidth]{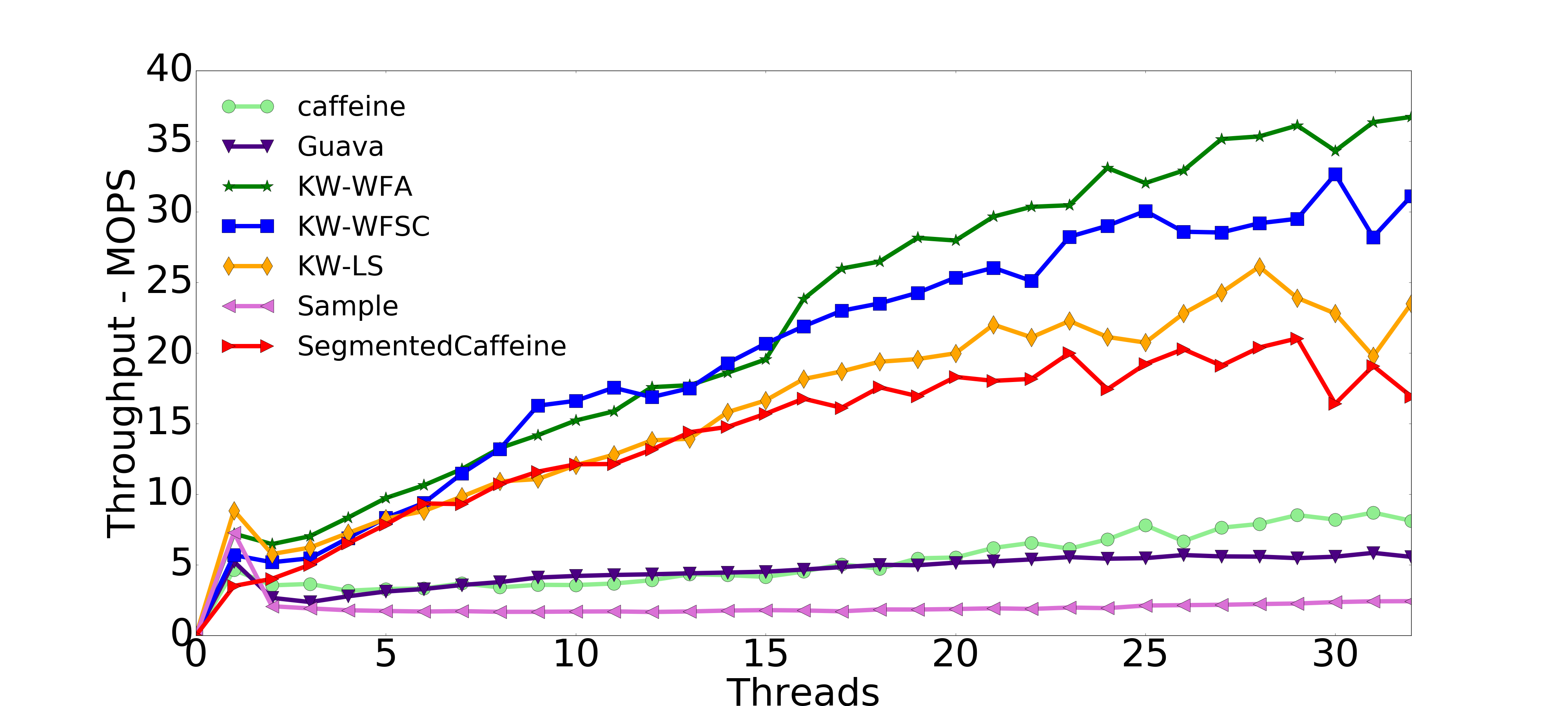}}
	\caption {OLTP trace with cache size of $2^{11}$ elements and duration of run of 1 second, run on AMD PowerEdge R7425. }
	\label{fig:OLTPthroughput}
\end{figure}
\begin{figure}[t]
	\center{
		\includegraphics[width=0.75\columnwidth]{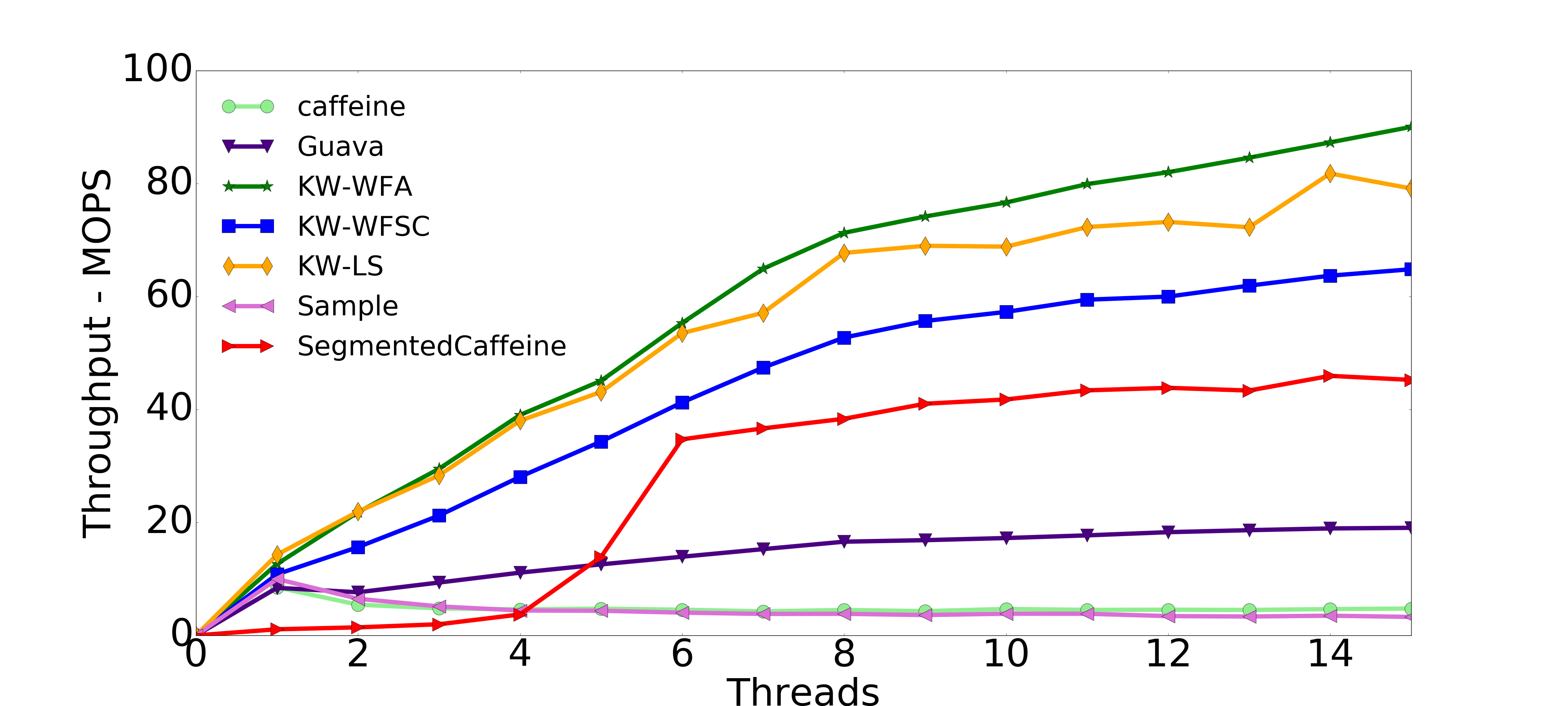}}
	\caption {F2 trace with cache size of $2^{11}$ elements and duration of run of 1 second, run on Intel Xeon E5-2667. }
	\label{fig:f2Throughput}
\end{figure}
\begin{figure}[t]
	\center{
		\includegraphics[width=0.75\columnwidth]{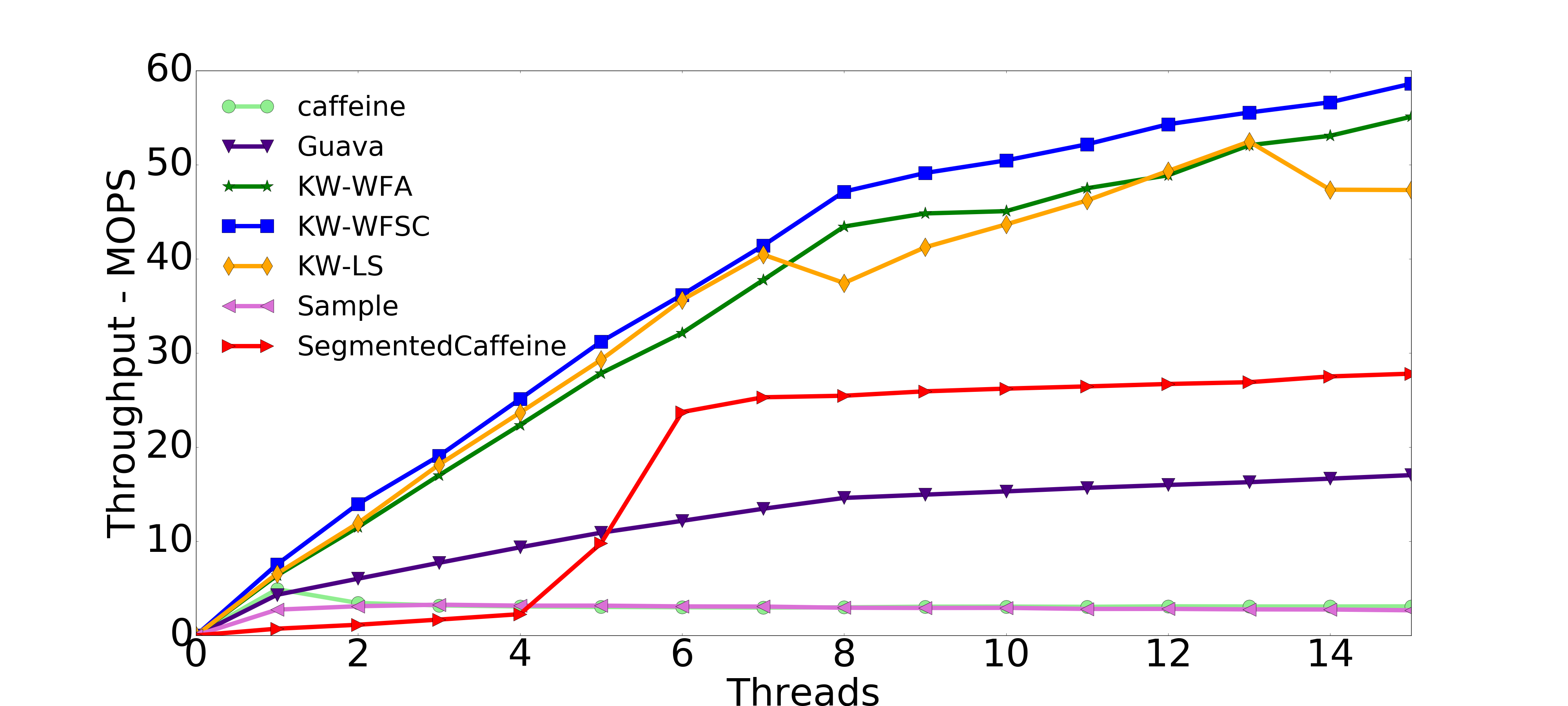}}
	\caption {W3 trace with cache size of $2^{19}$ elements and duration of run of 4 second, run on Intel Xeon E5-2667. }
	\label{fig:W3Throughput}
\end{figure}
\begin{figure}[h]
	\center{
		\includegraphics[width=0.75\columnwidth]{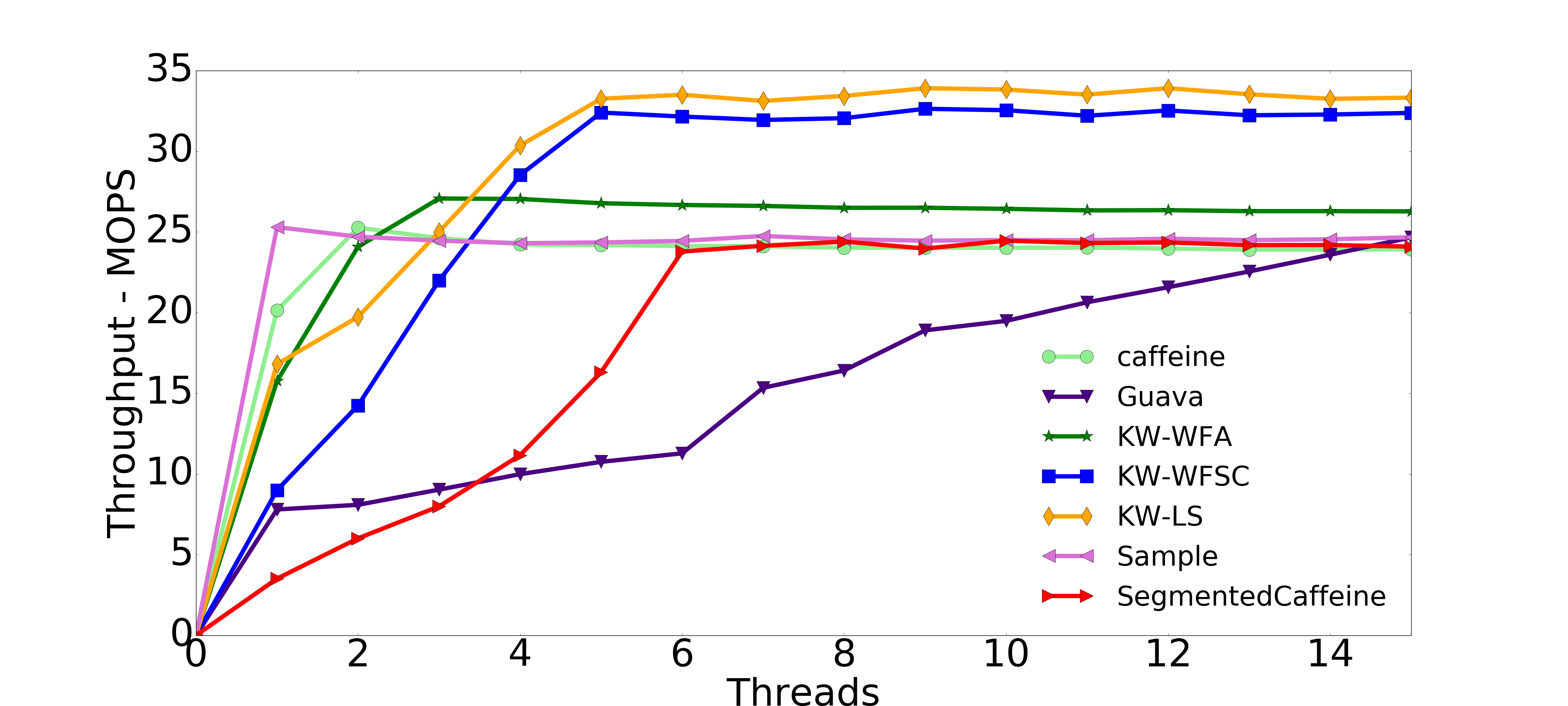}}
	\caption {multi1 trace with cache size of $2^{11}$ elements and duration of run of 1 second, run on Intel Xeon E5-2667. }
	\label{fig:multi1Throughput}
\end{figure}
\begin{figure}[t]
	\center{
		\includegraphics[width=0.75\columnwidth]{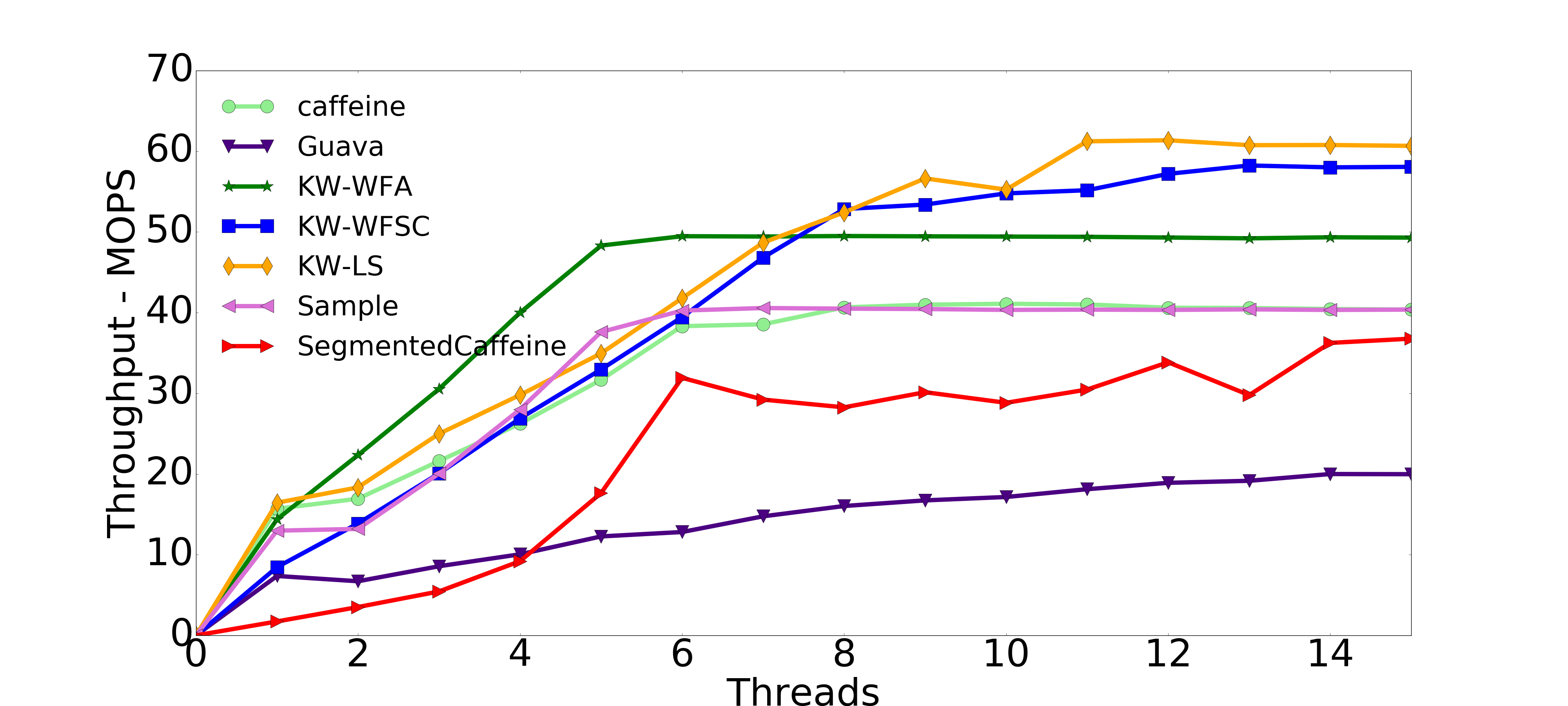}}
	\caption {multi2 trace with cache size of $2^{11}$ elements and duration of run of 1 second, run on Intel Xeon E5-2667. }
	\label{fig:multi2Throughput}
\end{figure}
\begin{figure}[t]
	\center{
		\includegraphics[width=0.75\columnwidth]{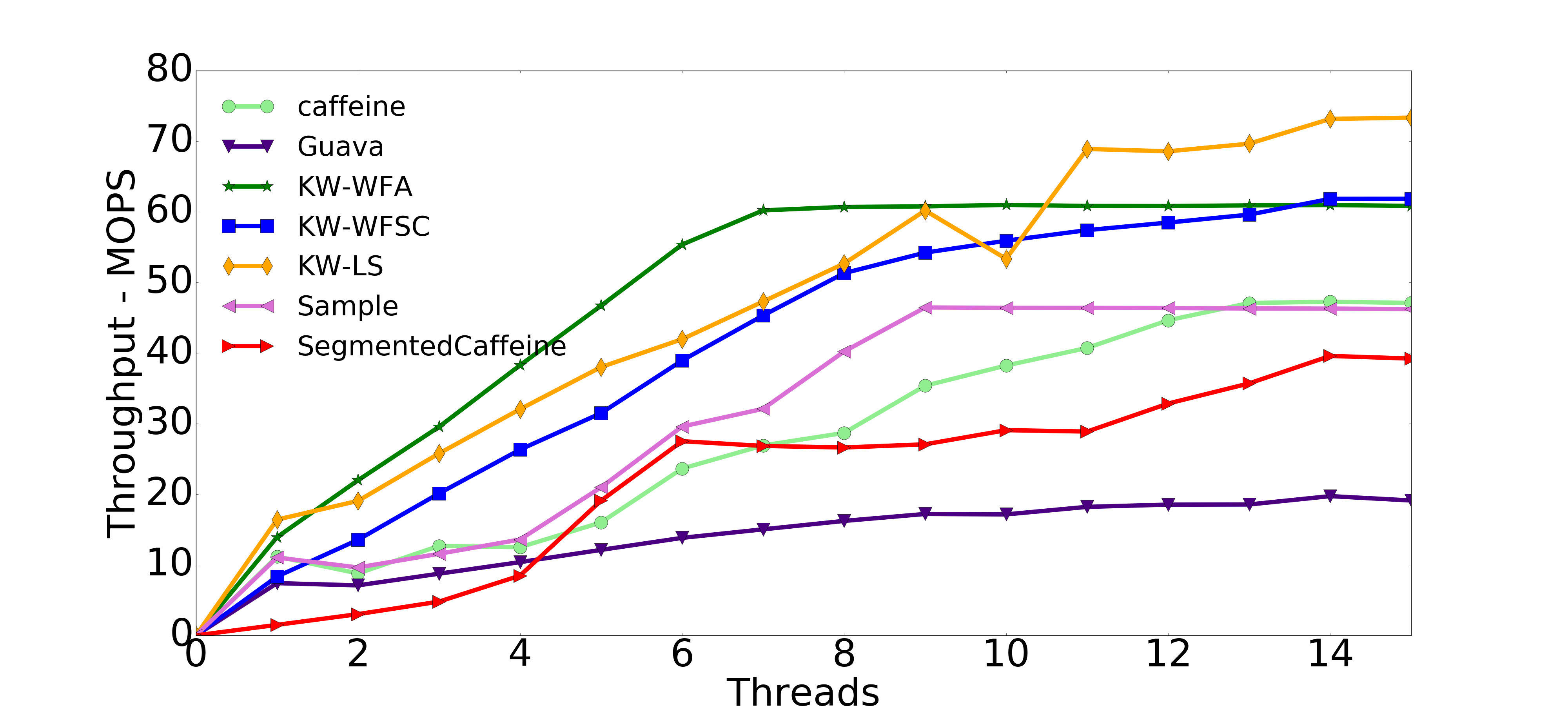}}
	\caption {multi3 trace with cache size of $2^{11}$ elements and duration of run of 1 second, run on Intel Xeon E5-2667. }
	\label{fig:multi3Throughput}
\end{figure}
 \begin{figure}[t]
 	\center{
 		\includegraphics[width=0.75\columnwidth]{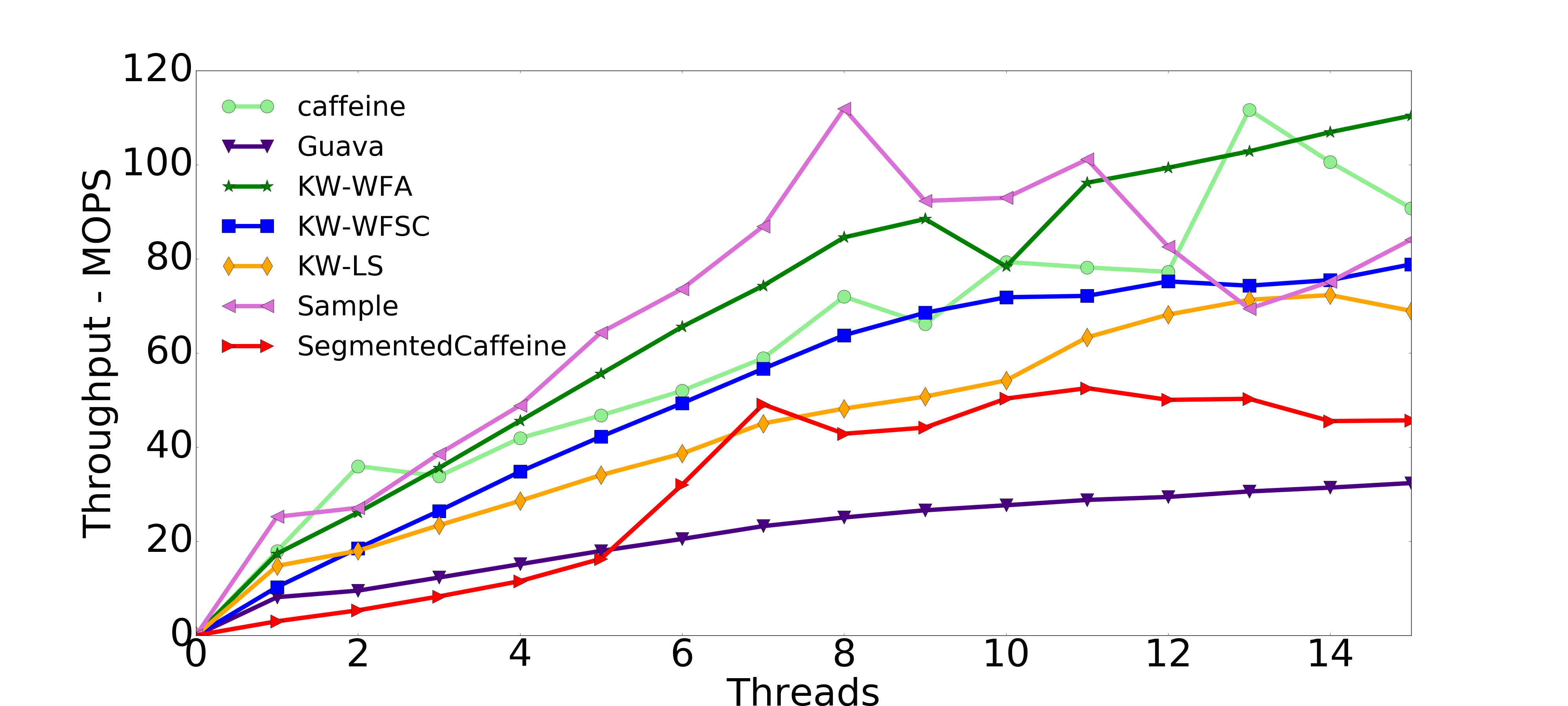}}
 	\caption {sprite trace with cache size of $2^{11}$ elements and duration of run of 1 second, run on Intel Xeon E5-2667.  }
 	\label{fig:spriteThroughput}
 \end{figure}
\begin{figure}[t]
	\center{
		\includegraphics[width=0.75\columnwidth]{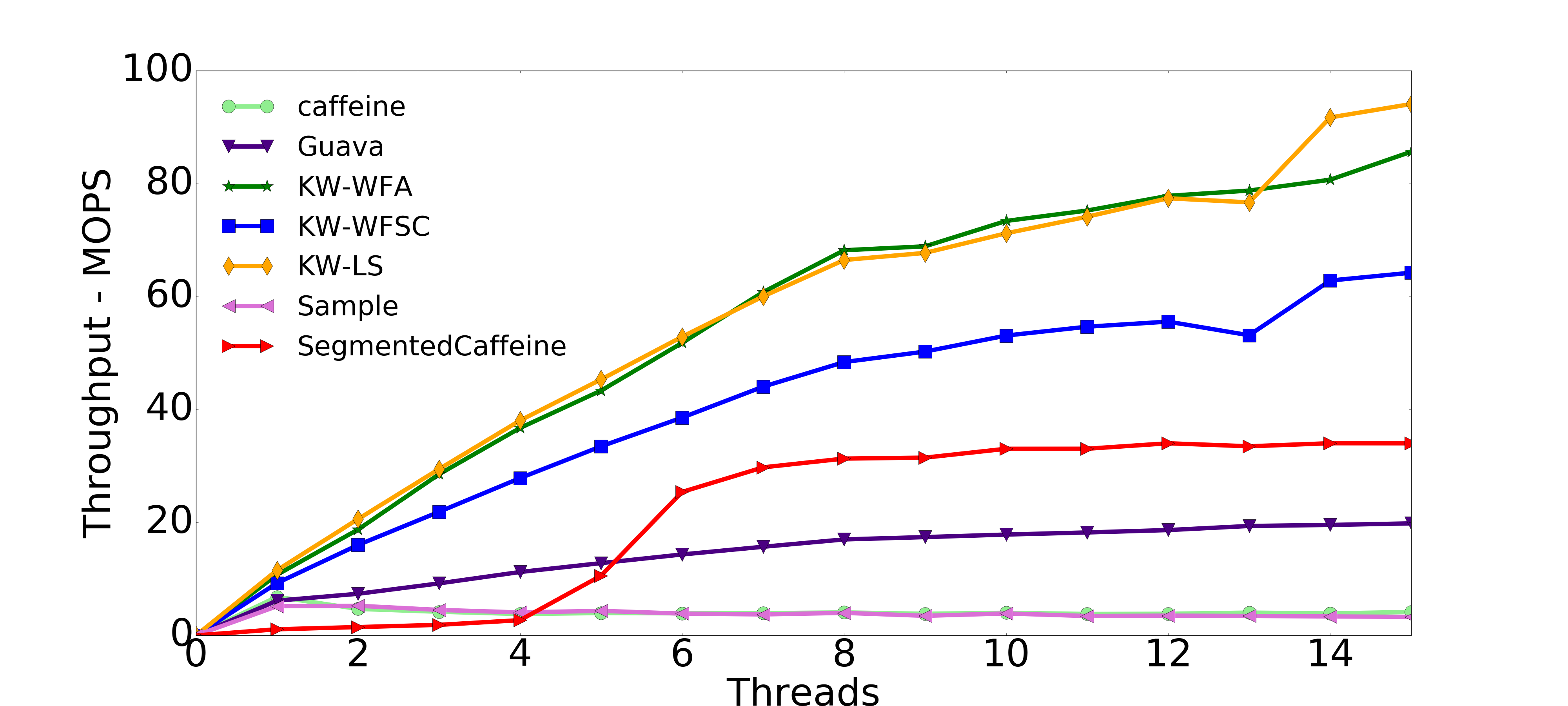}}
	\caption {P12 trace with cache size of $2^{17}$ elements and duration of run of 2 second, run on Intel Xeon E5-2667. }
	\label{fig:P12Throughput}
\end{figure}
\begin{figure}[t]
	\center{
		\includegraphics[width=0.75\columnwidth]{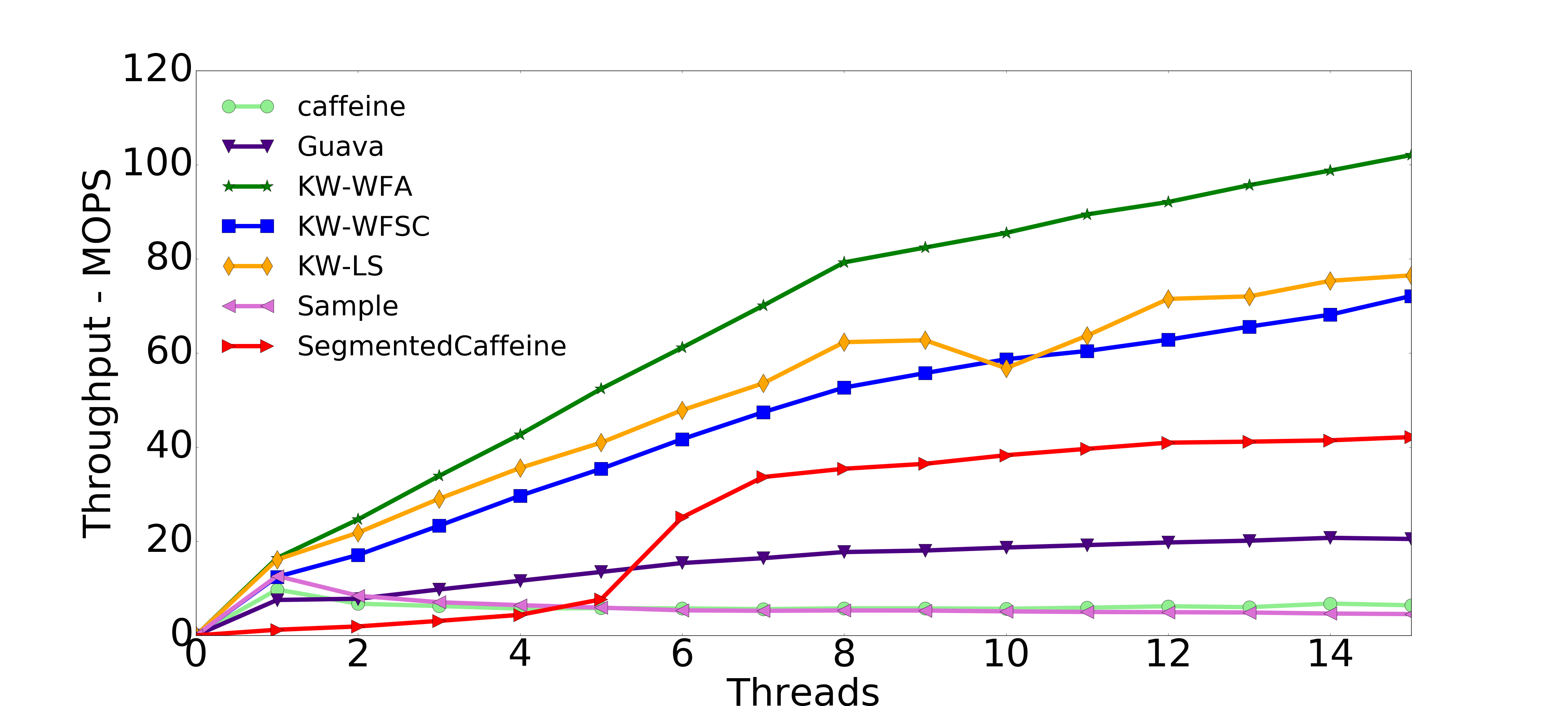}}
	\caption {wiki1191277217 trace with cache size of $2^{11}$ elements and duration of run of 1 second, run on Intel Xeon E5-2667. }
	\label{fig:wiki1191277217Throughput}
\end{figure}

The throughput measurements, in Millions of Get/Put operations per second, with different number of threads are shown in
\iftoggle{SMALL}{\Cref{fig:f1Throughput,fig:S3Throughput,fig:wiki1190322952throughput,fig:OLTPthroughput}}
{\Cref{fig:f1Throughput,fig:S3Throughput,fig:S1Throughput,fig:wiki1190322952throughput,fig:OLTPthroughput}}
for the AMD PowerEdge R7425 server and~\cref{fig:f2Throughput,fig:W3Throughput,fig:multi1Throughput,fig:multi2Throughput,fig:multi3Throughput,fig:spriteThroughput,fig:P12Throughput,fig:wiki1191277217Throughput} for the Intel Xeon E5-2667 server.

As can be observed, the \NAMNECACHE{} variations \FREEAR{}, \COUNTER{} and \LOCKSET{} gain the highest throughput.
Caffeine and sampled exhibited the worst worst-case performance, and it does not improve with additional threads for all traces except Sprite~(\cref{fig:spriteThroughput}), which will be addressed shortly.

For Caffeine, this is explained as put operations in Caffeine are performed in the background by a single thread. 
Guava is considerably faster than Caffeine in traces with significant number of misses because it performs put operations in the foreground in parallel.
As for the sampled approach, on every miss it needs to calculate $K$ random numbers and access these $K$ random locations to create the sample.
Even when $K$ is as small as 8, this already incurs a significant overhead.

Our limited associativity design therefore serves misses much faster than sampled (and Caffeine).
However, sampled may be quicker in serving a hit since it only accesses the (single) accessed object's meta-data, while our limited associativity design always accesses the entire set.
Further, the larger the cache is, the higher is the number of sets (since $K$ is independent of the cache size).
Hence, the hash function of the limited associativity approach is likely to yield better load balancing for larger caches.

For Sprite~(\cref{fig:spriteThroughput}), the hit ratio is high and the caches are small.
Hence, our limited associativity under-performs compared to sampled.
Yet, it still delivers tens of millions of operations per second.

\begin{figure}[t]
\center{
\ifdefined\ICPP
\includegraphics[width=0.8\columnwidth]{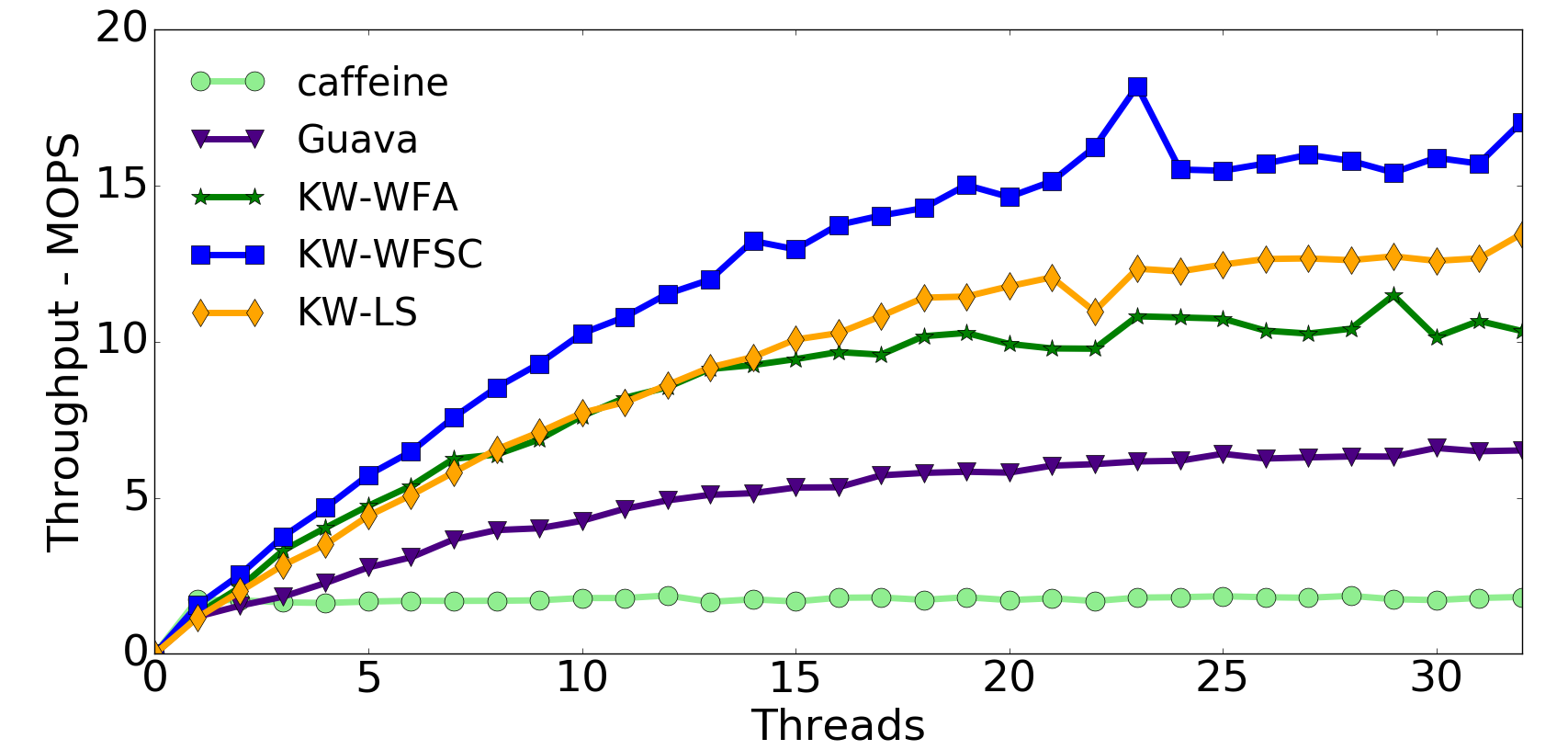}}
\else
\includegraphics[width=1\columnwidth]{graph/GetPut50.png}}
\fi
\caption {Synthetic trace of read and write, 100\% miss ratio, with cache size of $2^{21}$ elements, and duration of run of 5 second,   run on AMD PowerEdge R7425.}

\label{fig:miss}
\end{figure}

\subsection{Synthetic Throughput Evaluation}
\ifdefined\ICPP
\else
To better understand the factors that affect performance of our limited associativity solution and of Caffeine and Guava, we run synthetic traces that mimic various hit ratios.
We measure 5 seconds with a cache of $2^{21}$ elements on the AMD PowerEdge R7425 server.
The number of operations performed during this time interval varies between algorithms according to their speed.
\fi 

\paragraph{100\% miss ratio:}
The worst-case performance for any cache is the case of 100\% misses. 
In this evaluation, we perform a single get and a single put operation for each item and each item is requested only once. 
This extreme case helps us characterize the operational framework of various cache libraries. 

Figure~\ref{fig:miss} shows the attained throughput (measured in Millions of Get/Put operations per second). 
As can be observed, Caffeine has the worst worst-case performance and it does not improve with additional threads. 
This is explained as put operations in Caffeine are performed in the background by a single thread. 
Guava is considerably faster than Caffeine, because it performs put operations in the foreground, and next are our own three implementations of limited associativity caches. 
Here, the \COUNTER{} is the fastest of our algorithms because it is optimized towards fast replacement of the cache victim.

\paragraph{100\% hit ratio:}
Next, we repeat the last experiment but only perform Get operations to mimic having 100\% hit ratio. 
Our results are in Figure~\ref{fig:onlyGet}. 
Notice that here, Caffeine is considerably faster than all the alternatives since its read operations are simple hash table read operations. 
Next is Guava, which is very fast, and our algorithms are last. 
Interestingly, our algorithms attain very similar performance to the 100\% miss case, as we always scan the set - a full scan on a miss and half the set on average for a hit.
This implies that their performance is less sensitive to the workload characteristics. 
The performance of Caffeine and Guava varies greatly according to the cache effectiveness. 
We conclude that the break-even point between our approach vs. Caffeine and Guava depends on the hit ratio of the trace. 
So our next evaluation aims to determine at what hit ratios limited associativity caches are expected to outperform these established libraries. 

\begin{figure}[t]
	\center{
\ifdefined\ICPP
	\includegraphics[width=0.8\columnwidth]{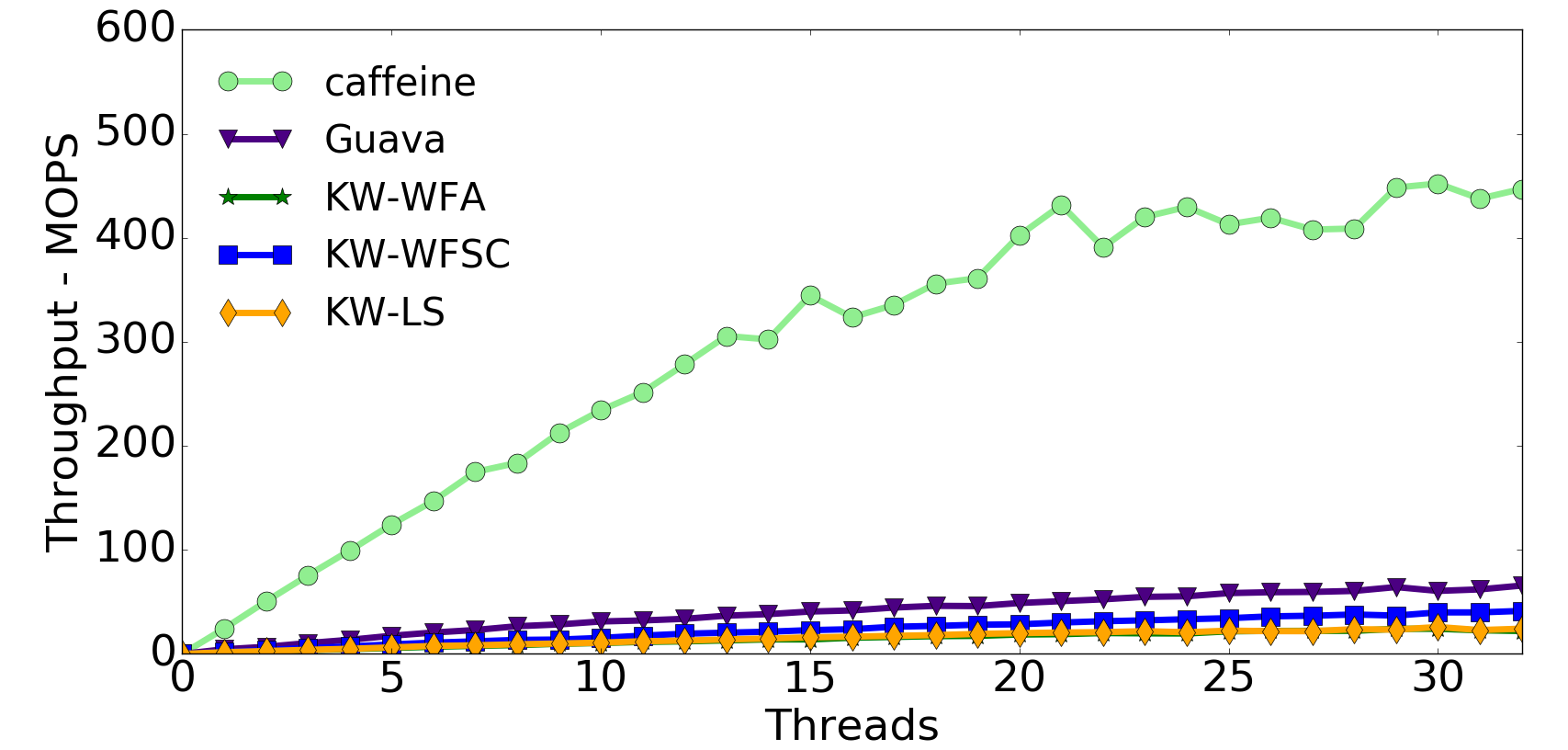}}
\else
	\includegraphics[width=1\columnwidth]{graph/onlyGet.png}}
\fi	
	\caption {Synthetic trace to mimic 100\% hit ratio with cache size of $2^{21}$ elements, and duration of run of 5 second,   run on AMD PowerEdge R7425. }
	
	\label{fig:onlyGet}
\end{figure}

\paragraph{95\% hit ratio:}
We continue and simulate 95\% hit ratio by performing one put operations for every 20 read operations. 
The results in Figure~\ref{fig:getput5} show that in 95\% hit ratio Guava is the fastest algorithm, while Caffeine is fastest for a small number of threads. 
Here, again Caffeine does not scale with the number of threads because the bottleneck still becomes the write buffer. 

\begin{figure}[h]
	\center{
\ifdefined\ICPP
	\includegraphics[width=0.8\columnwidth]{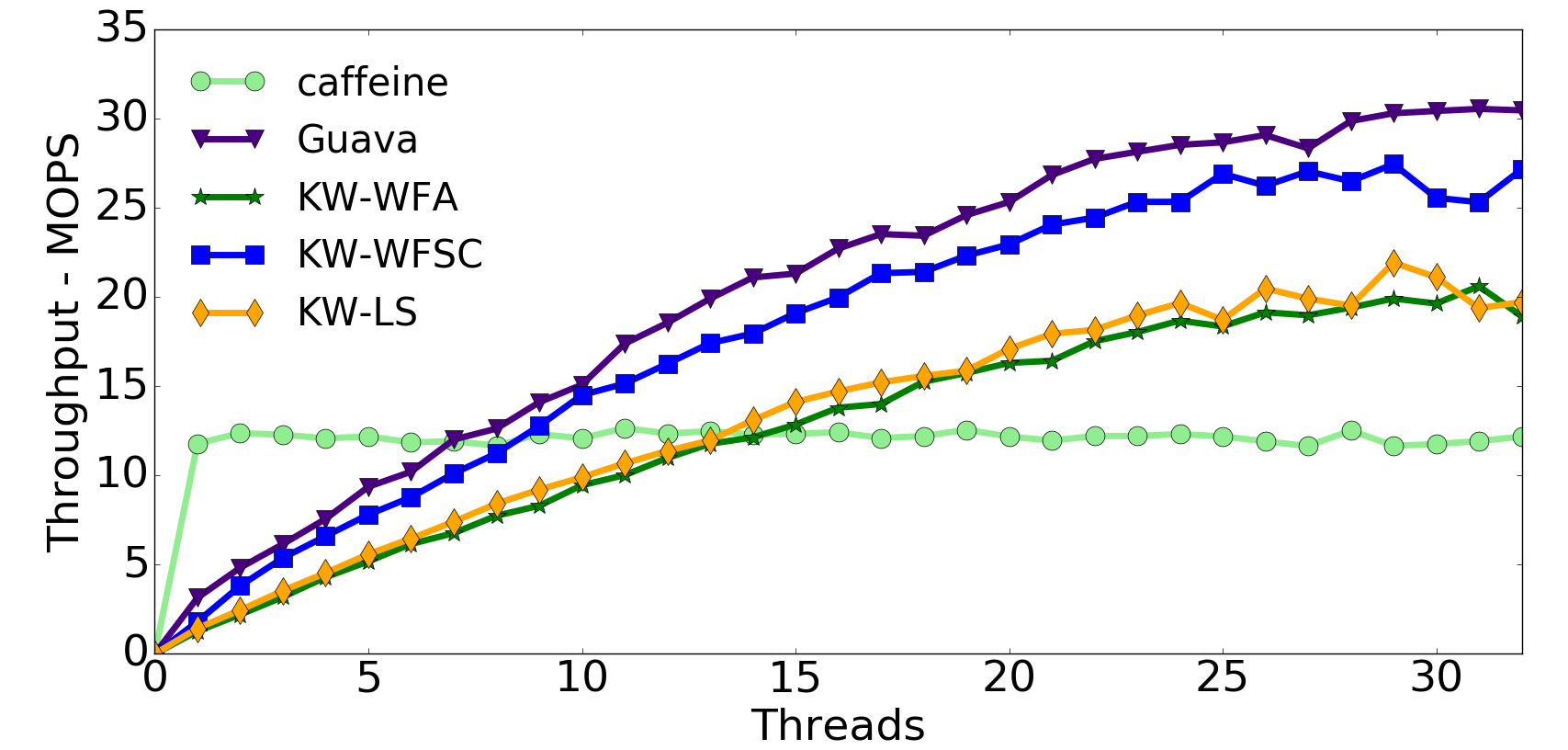}}
\else
		\includegraphics[width=1\columnwidth]{graph/getPut5.png}}
\fi	
	\caption {Synthetic trace to mimic 95\% hit ratio with cache size of $2^{21}$ elements, and duration of run of 5 second,   run on AMD PowerEdge R7425 }	
	\label{fig:getput5}
\end{figure}

\paragraph{90\% hit ratio:}
Next, we increase the number of writes to 1 in 10 to simulate 90\% hit ratio. 
The results in Figure~\ref{fig:getput10} show that \COUNTER{} is already faster than Guava and Caffeine when there are enough threads.  
That is, we conclude that our approach attains higher throughput than Caffeine and Guava as long as the hit ratio is less than 90\%, which implies our attractiveness for a very large variety of workloads and cache sizes. 

\begin{figure}[t]
	\center{
\ifdefined\ICPP
	\includegraphics[width=0.8\columnwidth]{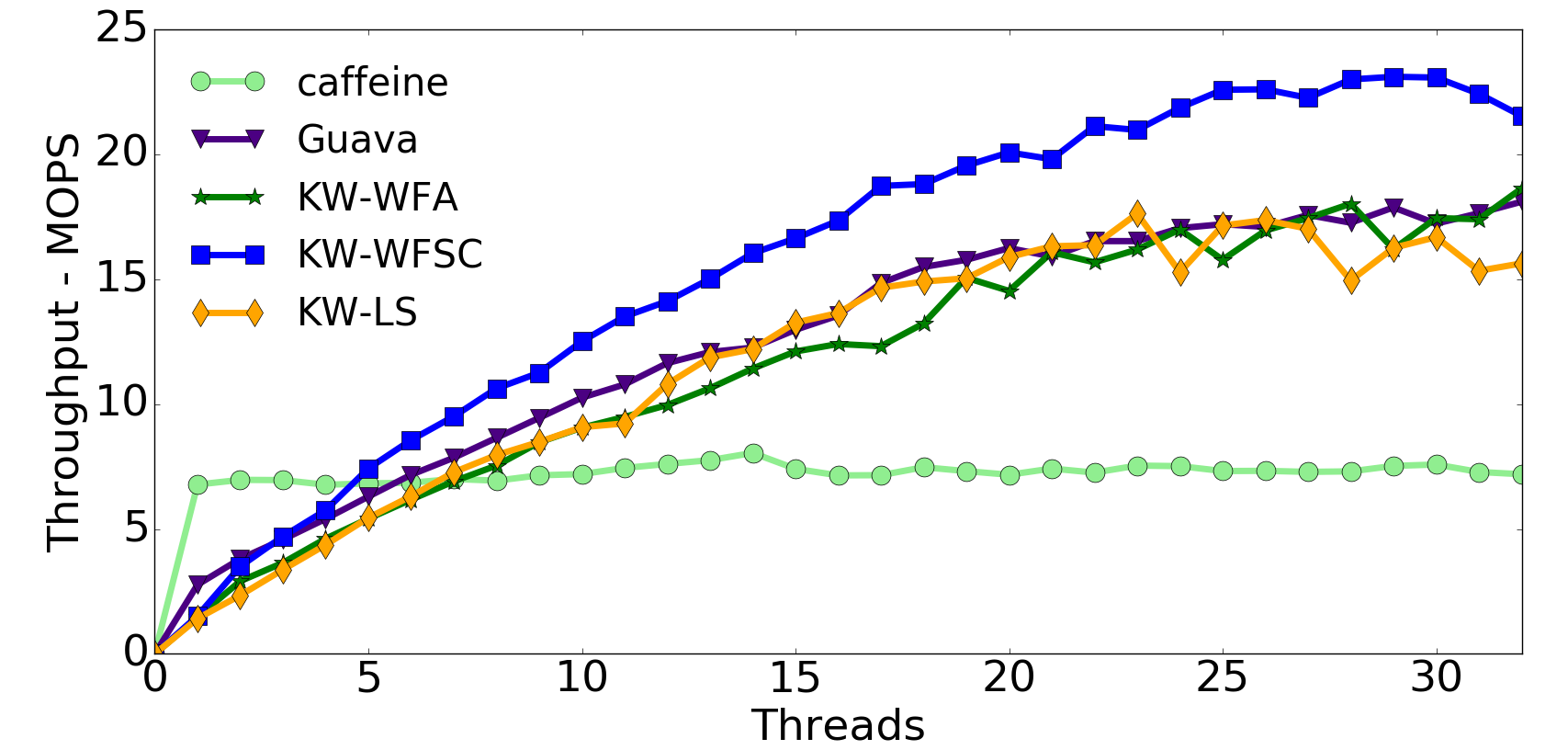}}
\else
	\includegraphics[width=1\columnwidth]{graph/getPut10.png}}
\fi	
	\caption {Synthetic trace to mimic 90\% hit ratio with cache size of $2^{21}$ elements, and duration of run of 5 second,   run on AMD PowerEdge R7425 }
	
	\label{fig:getput10}
\end{figure}

\paragraph{Conclusions:}
In our limited associativity design, get and put have similar performance.
In contrast, in sampled, Caffeine and Guava gets are fast and puts are slow.
Our performance is expected to be better than sampled, Caffeine and Guava when the hit ratio is below 90\%, which captures most real traces and cache sizes. 
Our approach scales better with the number of threads in almost all~scenarios.

\section{Discussion}
\label{sec:discussion}
\ifdefined\ICPP
\else
Non distributed software caches are almost always fully associative and often research papers that discuss such caches do not even explicitly mention this design choice. 
In hardware systems, the vast majority of caches are designed with limited associativity and the benefits of such a design are well understood.
\fi
Our work argues that software cache associativity is an important design choice that should be given consideration. 
We note that limited associativity caches are easier and simpler to implement than fully associative caches. 
They are embarrassingly parallel and have lower memory overheads compared to fully associative caches. 
On the other hand, limited associativity may negatively impact the hit ratio, which is an important quality measure for every cache. 

\ifdefined\ICPP
Sampled techniques, such as Sampled LRU of Redis~\cite{redis-lru} and the sampling performed by Hyperbolic caching~\cite{Hyperbolic}, reduce accuracy in favor of speed.
\else
Trading-off hit ratio for speed is a common compromise in cache design, e.g., MemC3~\cite{MemC3} shows that the performance of MemCacheD can be improved by using ``dumber caching'' that offers better performance.
Similarly, sampled techniques, such as Sampled LRU of Redis~\cite{redis-lru} and the sampling performed by Hyperbolic caching~\cite{Hyperbolic}, reduce accuracy in favor of speed.
\fi
This is because one can always improve the hit ratio by increasing the cache size, but there are no simple techniques to improve the operation speed of the cache. 
On the other hand, caches' utility is due to performance gaps between different ways to fetch data. 
\ifdefined\ICPP
For example, main memory is faster than secondary storage.
\else
For example, main memory is faster than secondary storage and accessing compressed data is slower than accessing uncompressed data.
\fi 
Therefore, the potential benefits from caching are based on the gap between the operation time of the cache and the operation time of the alternative method. 
Thus, improving the operation time of caches makes every cache hit more beneficial to the application (as long as we do not significantly reduce the hit~ratio). 

Our work evaluated limited associativity caches on diverse (real) application workloads over multiple popular cache management policies. 
We showed that limited associativity caches attain a similar (and often nearly identical) hit ratio compared to fully associative caches for many workloads and policies. 
Our work also suggested three different implementations of limited associativity caches.
We showed an improvement of up to x$5$ compared to the Caffeine and Guava libraries that together dominate the Java ecosystem, especially in multi-threaded settings.
Given these findings, we believe that such caching libraries would benefit from adopting limited associativity cache designs.

Among our different implementations, when the trace is uniformly separated without heavy hitters, so the insertions are usually to different buckets, it is better to use the KW-LS.
If there are many reads and not many writes, it is better to use the KW-WFSC as the reads are continuous in memory and hence faster than KW-WFA where the array is an array of pointers so the performance of the reads is compromised.
Otherwise, it is best to use the KW-WFA as it uses only one atomic operation in contrast to KW-WFC with three atomic operations, which slow down the cache updates. 

\nottoggle{SMALL}{
Interestingly, for hardware caches, the seminal work of Qureshi, Thompson and Patt~\cite{v-way} has found that to fully optimize the cache hit ratio, one should vary the associativity of a cache on a per-set basis in response to the demands of the program.
Supporting such dynamic associativity level in software is straight-forward and could help improve the hit-ratio gaps between $8$ ways and fully associative in traces like WS2.
Exploring this direction is left for future work.
}

\paragraph*{Acknowledgments:}
This work was partially	funded by ISF grant \#1505/16 and the Technion-HPI research school as well as the Cyber Security Research Center and the Data Science Research Center at Ben-Gurion University




 \bibliographystyle{abbrv}
	\bibliography{mybibfile}

\begin{thebibliography}{10}

\bibitem{accumulo}
{Apache Accumulo}.
\newblock 2020.

\bibitem{cassandra}
{Apache Cassandra}.
\newblock 2020.

\bibitem{hbase}
{Apache HBase}.
\newblock 2020.

\bibitem{javaMap}
{Java ConcurrentMap}.
\newblock https://en.wikipedia.org/wiki/Java\_ConcurrentMap, 2020.

\bibitem{dgraph}
{The DGraph Website}.
\newblock 2020.

\bibitem{neo4j}
{The Neo4j Website}.
\newblock 2020.

\bibitem{redis-lru}
{Using Redis as an LRU Cache}.
\newblock 2020.

\bibitem{Hifi}
S.~Akhtar, A.~Beck, and I.~Rimac.
\newblock {HiFi: A Hierarchical Filtering Algorithm for Caching of Online
  Video}.
\newblock In {\em ACM MM}, MM, pages 421--430, 2015.

\bibitem{LFUDA}
M.~Arlitt, L.~Cherkasova, J.~Dilley, R.~Friedrich, and T.~Jin.
\newblock {Evaluating Content Management Techniques for Web Proxy Caches}.
\newblock In {\em In Proc. of the 2nd Workshop on Internet Server Performance},
  1999.

\bibitem{LFUAGING}
M.~Arlitt, R.~Friedrich, and T.~Jin.
\newblock {Performance Evaluation of Web Proxy Cache Replacement Policies}.
\newblock {\em Perform. Eval.}, 39(1-4):149--164, Feb. 2000.

\bibitem{hashcache}
A.~Badam, K.~Park, V.~S. Pai, and L.~L. Peterson.
\newblock {HashCache: Cache Storage for the Next Billion}.
\newblock In {\em NSDI}, pages 123--136, 2009.

\bibitem{LHD}
N.~Beckmann, H.~Chen, and A.~Cidon.
\newblock {LHD: Improving Cache Hit Rate by Maximizing Hit Density}.
\newblock In {\em NSDI}, pages 389--403, Apr. 2018.

\bibitem{WCSS}
R.~Ben-Basat, G.~Einziger, R.~Friedman, and Y.~Kassner.
\newblock {Heavy Hitters in Streams and Sliding Windows}.
\newblock In {\em IEEE INFOCOM}, 2016.

\bibitem{Hyperbolic}
A.~Blankstein, S.~Sen, and M.~J. Freedman.
\newblock {Hyperbolic Caching: Flexible Caching for Web Applications}.
\newblock In {\em ATC}, pages 499--511, 2017.

\bibitem{clock}
F.~J. Corbato.
\newblock {A Paging Experiment with the Multics System}.
\newblock Technical Report Project MAC Report MAC-M-384, MIT, May 1968.

\bibitem{EEFM18}
G.~Einziger, O.~Eytan, R.~Friedman, and B.~Manes.
\newblock {Adaptive Software Cache Management}.
\newblock In {\em Middleware}, pages 94--106, 2018.

\bibitem{EFM17}
G.~Einziger, R.~Friedman, and B.~Manes.
\newblock {TinyLFU: A Highly Efficient Cache Admission Policy}.
\newblock {\em ACM ToS}, 2017.

\bibitem{MemC3}
B.~Fan, D.~G. Andersen, and M.~Kaminsky.
\newblock {MemC3: Compact and Concurrent MemCache with Dumber Caching and
  Smarter Hashing}.
\newblock In {\em NSDI}, pages 371--384, 2013.

\bibitem{CountingBloom}
L.~Fan, P.~Cao, J.~Almeida, and A.~Z. Broder.
\newblock Summary cache: a scalable wide-area web cache sharing protocol.
\newblock {\em IEEE/ACM Trans. Netw.}, 8(3):281--293, June 2000.

\bibitem{FM2020-fast}
R.~Friedman and B.~Manes.
\newblock {Tutorial: Designing Modern Software Caches}.
\newblock FAST, 2020.

\bibitem{guava-cache}
Google.
\newblock {Guava: Google Core Libraries for Java}.
\newblock {\em https://github.com/google/guava}, 2016.

\bibitem{LRU}
J.~L. Hennessy and D.~A. Patterson.
\newblock {\em {Computer Architecture - A Quantitative Approach (5. ed.)}}.
\newblock Morgan Kaufmann, 2012.

\bibitem{hopscotchhashing}
M.~Herlihy, N.~Shavit, and M.~Tzafrir.
\newblock {Hopscotch Hashing}.
\newblock DISC, pages 350--364, 2008.

\bibitem{HS89}
M.~D. {Hill} and A.~J. {Smith}.
\newblock {Evaluating Associativity in CPU Caches}.
\newblock {\em IEEE Transactions on Computers}, 38(12):1612--1630, Dec 1989.

\bibitem{clock-pro}
S.~Jiang, F.~Chen, and X.~Zhang.
\newblock {CLOCK-Pro: An Effective Improvement of the CLOCK Replacement}.
\newblock In {\em USENIX Annual Technical Conference}, ATEC, pages 35--35,
  2005.

\bibitem{LIRS}
S.~Jiang and X.~Zhang.
\newblock {LIRS: An Efficient Low Inter-reference Recency Set Replacement
  Policy to Improve Buffer Cache Performance}.
\newblock {\em ACM SIGMETRICS}, pages 31--42, 2002.

\bibitem{WLFU}
G.~Karakostas and D.~N. Serpanos.
\newblock {Exploitation of Different Types of Locality for Web Caches}.
\newblock In {\em ISCC}, 2002.

\bibitem{SLRU}
R.~Karedla, J.~S. Love, and B.~G. Wherry.
\newblock {Caching Strategies to Improve Disk System Performance}.
\newblock {\em IEEE Computer}, 1994.

\bibitem{LFUIMPl}
A.~Ketan~Shah and M.~D. Matani.
\newblock {An O(1) Algorithm for Implementing the LFU Cache Eviction Scheme}.
\newblock Technical report, 2010.
\newblock "http://dhruvbird.com/lfu.pdf".

\bibitem{MaierSD19}
T.~Maier, P.~Sanders, and R.~Dementiev.
\newblock Concurrent hash tables: Fast and general(?)!
\newblock {\em {ACM} Trans. Parallel Comput.}, 5(4):16:1--16:32, 2019.

\bibitem{CaffeineProject}
B.~Manes.
\newblock {Caffeine: A High Performance Caching Library for Java 8}.
\newblock {\em https://github.com/ben-manes/caffeine}, 2017.

\bibitem{Seg-Caffeine}
B.~Manes.
\newblock {Segmented Caffeine -- Private Communications}.
\newblock 2020.

\bibitem{ARC}
N.~Megiddo and D.~S. Modha.
\newblock {ARC: A Self-Tuning, Low Overhead Replacement Cache}.
\newblock In {\em FAST}, pages 115--130, 2003.

\bibitem{SpaceSavings}
A.~Metwally, D.~Agrawal, and A.~E. Abbadi.
\newblock {Efficient Computation of Frequent and Top-k Elements in Data
  Streams}.
\newblock In {\em ICDT}, 2005.

\bibitem{michael2002high}
M.~M. Michael.
\newblock High performance dynamic lock-free hash tables and list-based sets.
\newblock In {\em SPAA}, pages 73--82, 2002.

\bibitem{Mitzenmacher2005}
M.~Mitzenmacher and E.~Upfal.
\newblock {\em {Probability and Computing: Randomized Algorithms and
  Probabilistic Analysis}}.
\newblock Cambridge University Press, New York, NY, USA, 2005.

\bibitem{LRUK}
E.~J. O'Neil, P.~E. O'Neil, and G.~Weikum.
\newblock {The LRU-K Page Replacement Algorithm for Database Disk Buffering}.
\newblock {\em ACM SIGMOD Rec.}, 22(2):297--306, June 1993.

\bibitem{FRD}
S.~Park and C.~Park.
\newblock {FRD: A Filtering based Buffer Cache Algorithm that Considers both
  Frequency and Reuse Distance}.
\newblock In {\em MSST}, 2017.

\bibitem{SurveyOfCacheReplecmentStrategies}
S.~Podlipnig and L.~B\"{o}sz\"{o}rmenyi.
\newblock A survey of web cache replacement strategies.
\newblock {\em ACM Comput. Surv.}, 35(4):374--398, Dec. 2003.

\bibitem{v-way}
M.~K. {Qureshi}, D.~{Thompson}, and Y.~N. {Patt}.
\newblock {The V-Way Cache: Demand-Based Associativity via Global Replacement}.
\newblock In {\em ISCA}, pages 544--555, June 2005.

\bibitem{AdaptiveCacheReplacement}
G.~Tewari and K.~Hazelwood.
\newblock Adaptive web proxy caching algorithms.
\newblock Technical Report TR-13-04, "Harvard University", 2004.

\bibitem{xxHash}
V.~Tolstopyatov.
\newblock {xxHash}.
\newblock {\em https://github.com/OpenHFT/Zero-Allocation-Hashing}.

\bibitem{wikipedia}
G.~Urdaneta, G.~Pierre, and M.~Van~Steen.
\newblock {Wikipedia Workload Analysis for Decentralized Hosting}.
\newblock {\em Computer Networks}, 53(11):1830--1845, 2009.

\bibitem{umasstrace}
G.~Weikum.
\newblock {UMassTrace Repository}.
\newblock {\em http://traces. cs. umass. edu/index. php/Storage}, 2011.

\end{thebibliography}

\end{document}